\documentclass[11pt,journal,draftcls,onecolumn,peerreviewca]{IEEEtran}
\usepackage{color}
\usepackage{float}
\usepackage{amsthm}
\usepackage{amsmath}
\usepackage{amssymb}
\usepackage{stackrel}
\usepackage{graphicx}
\usepackage{graphics}
\usepackage{caption}
\usepackage{esint}
\usepackage{subfigure}
\usepackage{algorithm}
\usepackage{algorithmicx}
\usepackage{algpseudocode}
\usepackage{mathrsfs}

\newtheorem{lemma}{Lemma}

\usepackage[T1]{fontenc}
\floatstyle{ruled}\newfloat{algorithm}{tbp}{loa}
\providecommand{\algorithmname}{Algorithm}
\floatname{algorithm}{\protect\algorithmname}
\newtheorem{proposition}{Proposition}

\pdfpageheight\paperheight
\pdfpagewidth\paperwidth
\begin{document}
\title{Line Spectrum Estimation and Detection with Few-bit ADCs: Theoretical Analysis and Generalized NOMP Algorithm}
\author{Jiang Zhu, Hansheng Zhang, Ning Zhang, Jun Fang and Fengzhong Qu
\thanks{
Jiang Zhu, Hansheng Zhang and Fengzhong Qu are with the Ocean College, Zhejiang University, and are also with the engineering research center of oceanic sensing technology and equipment, Ministry of Education, No.1 Zheda Road, Zhoushan, 316021, China (email: \{jiangzhu16, hanshengzhang, jimqufz\}@zju.edu.cn). Ning Zhang is with the Nanjing Marine Radar Institute, Nanjing, China (email: zhangn\_ee@163.com). Jun Fang is with the National Key Laboratory of Science and Technology on Communications, University of Electronic Science and Technology of China, Chengdu 611731, China (email: JunFang@uestc.edu.cn).}}
\maketitle
\begin{abstract}
As radar systems will be equipped with thousands of antenna elements and wide bandwidth, the associated costs and power consumption become exceedingly high, and a potential solution is to adopt low-resolution quantization technology, which not only reduces data storage needs but also lowers power and hardware costs. This paper focuses on line spectral estimation and detection (LSE\&D) with few-bit ADCs (typically 1-4 bits) by investigating the signal-to-noise ratio (SNR) loss, establishing a framework to understand the impact of intersinusoidal interference, the bit-depth of the quantizer, and the noise variance on weak signal detection in scenarios involving multiple sinusoids under low-resolution quantization. Additionally, a low-complexity, super-resolution, and constant false alarm rate (CFAR) algorithm, named generalized Newtonized orthogonal matching pursuit (GNOMP), is proposed. Extensive numerical simulations are conducted to validate the theoretical findings, particularly in terms of the detection probability bound. The effectiveness of GNOMP is demonstrated through comparisons with state-of-the-art algorithms, the Cram\'{e}r Rao bound, and the detection probability bound. Real data acquired by mmWave radar further substantiates the effectiveness of GNOMP in practical applications.
\end{abstract}
\begin{IEEEkeywords}
Generalized Newtonized orthogonal matching pursuit (GNOMP), low resolution quantization, constant false alarm rate (CFAR), Rao test, weak signal detection, line spectral estimation and detection.
\end{IEEEkeywords}

\section{Introduction}

Low-resolution quantization (e.g., 1-4 bits) emerges as a promising technique for future digital radar systems equipped with large arrays and high bandwidth, as highlighted in previous studies \cite{Mishra19SPM, KumariICASSP2020, Ciuonzo2013SPL, NiTSP2023, Xu2021TSIPN}. When compared to traditional radars employing high-resolution analog-to-digital converters (ADCs) (e.g., 10-12 bits), radar systems utilizing low-resolution quantization present both advantages and challenges.

On the positive side, adopting low-resolution quantization significantly reduces power consumption, hardware costs, and data rates generated by the radar. Conversely, the nonlinearity inherent in low-resolution quantization introduces substantial distortion to the original signal. This poses a challenge to conventional signal processing methods such as matched filtering (MF)-based linear approaches and compressed sensing (CS)-based iterative methods, leading to performance degradation by neglecting the nonlinear quantization effects. Theoretical analyses indicate that coarsely quantized line spectral exhibits abundant harmonics, potentially causing false alarms when using conventional fast Fourier transform (FFT) methods.

Moreover, the low-resolution ADC exacerbates the near-far problem by reducing the instantaneous dynamic range (DR)\footnote{In \cite{DWR2016p25}, the instantaneous dynamic range concerning two simultaneous signals is defined as the power ratio of the maximum and minimum simultaneously received pulses that can be properly detected by the receiver.}. Consequently, investigating the detection performance limits for weak signals in scenarios involving multiple sinusoids becomes meaningful. It is crucial to unveil the impacts of quantizer bit-depth, noise variance, and intersinusoidal interference on weak signal detection and to design practical algorithms that approach these performance limits.
\subsection{Related Work}
Range estimation and target detection in linear frequency modulated continuous wave (LFMCW) radar can be conceptualized as line spectral estimation and detection (LSE\&D). Existing research on LSE\&D from coarsely quantized samples can be broadly categorized into two groups: signal-reconstruction-based linear approaches and parameter-estimation-and-detection-based nonlinear methods \cite{LFMCWTAES20}. From the signal reconstruction perspective, several criteria of the analog-to-digital converter (ADC), including the dynamic range (DR) and spurious free dynamic range (SFDR), have been thoroughly investigated. The spectrum of the low-resolution quantized signal, particularly in cases of 1-bit quantization, has been analyzed, shedding light on the effects of linear processing on target estimation and detection. On the other hand, the parameter-estimation-and-detection-based method aims to directly perform target detection and estimation using a parameterized model through nonlinear processing, such as the CS-based algorithm. In the following sections, we delve into the specifics of these research endeavors.

Extensive analysis has been conducted on the criteria of ADCs when subjected to sinusoidal inputs. The study reveals that the DR of a general ADC is $6B+1.72$ dB, where $B$ represents the bit-depth. Additionally, it has been established that ADCs generate harmonics, imposing limitations on their SFDR. In the context of radar applications, it has been demonstrated that binary data contains abundant self-generated harmonics \cite{SAR1991} and cross-generated harmonics \cite{LFMCWTAES20}. In scenarios with low signal-to-noise ratio (SNR), harmonic strengths decay rapidly, and conventional FFT methods demonstrate robust performance. However, in high SNR scenarios, FFT may lead to an overestimation of the model order. A notable approach presented in \cite{onebitDBF} involves forming fundamental and harmonic beams separately within subbands. Fine angular-resolution harmonic beams are then leveraged, bypassing the need for suppression in such scenarios.

Regarding parameter estimation, the  Cram\'{e}r Rao bound (CRB) has been employed to analyze the performance limits of maximum likelihood estimators (MLE). In-depth analysis has been carried out in \cite{Papatit} on scalar parameter estimation under various control inputs, including additive noise control input, threshold control input, and feedback control input. It was unveiled that appropriately tuned noise positively influences the estimator's performance, a phenomenon well-known as stochastic resonance. The study also highlighted the significant impact of the threshold input on MLE performance. When the threshold is close to the true value compared to the noise, the performance degradation, compared to the unquantized system, is minimal; otherwise, the degradation is substantial. In the context of single-tone frequency estimation \cite{SingletoneTSP}, it was discovered that 1-bit quantization results in a dramatic increase in variance at certain frequencies and slightly worse performance at other frequencies. Meanwhile, in \cite{DOA1bit02}, a stochastic Gaussian point source model was assumed to study direction of arrival (DOA) estimation, analogous to the temporal line spectral estimation (LSE) problem with multiple measurement vectors (MMVs). The closed-form CRB for a two-sensor array case was derived, revealing weak dependency on SNR for estimation error. Notably, two singular DOA angles, $0^{\circ}$ and $30^{\circ}$, were identified, showing that higher SNR leads to improved estimation performance. In the study presented in \cite{NingTAES22}, LSE with MMVs from coarsely quantized samples was investigated. For single-tone frequency estimation with multisnapshots and deterministic parameter modeling, the performance bound and its asymptotic behavior under 1-bit quantization were provided. The asymptotic analysis indicated that the CRB is inversely proportional to the number of snapshots and the cubic of the number of samples per snapshot. In lower SNR scenarios, the CRB is inversely proportional to SNR, while in higher SNR scenarios, the CRB is inversely proportional to the square root of the SNR.

The above work focuses on the theoretical aspects of LSE. Next, we will introduce related work in the algorithmic part. On-grid methods, which discretize frequencies into a finite number of grids, have been commonly employed \cite{Yu, MengZhu, LFMCWTAES20}. However, these methods often suffer from model mismatch \cite{mismatch}. To overcome this issue incurred by on-grid assumptions \cite{Yangzaireview}, a gridless atomic norm minimization (ANM) approach has been introduced, involving solving semidefinite programming and exhibiting high computational complexity \cite{Fu, Wen}. In efforts to mitigate computational burdens, alternative methods such as a gridless support vector machine (SVM) approach \cite{Gaoyu} and a relaxation-based approach known as 1bRELAX have been proposed \cite{Gianelli2}. 1bRELAX utilizes grid refinement to address model mismatch issues \cite{LiJian18SPL, LiJian19TSP, Zhang2019}. Additionally, the Bayesian information criterion (BIC) is employed to determine the number of scatters \cite{LiJian18SPL}. This approach has been extended to handle range estimation and range-Doppler imaging in LFMCW radar \cite{Zhang2019}. However, the computational complexity of the 1bRELAX algorithm remains high, as it solves convex optimization problems for all the grids in a single iteration. To enhance computational efficiency, a Majorization–Minimization (MM)-based variant, denoted as 1bMMRELAX, has been proposed. In efforts to enhance computational efficiency, low-complexity algorithms, VALSE-EP and MVALSE-EP, have been proposed \cite{VALSEEP, NingTAES22, zhuTWC}. These methods incorporate Newton steps to refine frequencies and utilize the expectation maximization (EM) algorithm to automatically estimate noise variance for bit-depth greater than 1. While these methods provide an efficient Bayesian inference framework for LSE\&D from quantized samples, they have limitations. Specifically, VALSE-EP and MVALSE-EP may act like black boxes, making it challenging to ascertain if a sinusoid is missed. Furthermore, Bayesian algorithms, designed for estimation, may not be a good choice for target detection.

In summary, the existing studies mentioned do not delve into the impact of low-resolution quantization on the DR for the recovery of multiple sinusoidal signals. Additionally, there is a gap in the availability of a fast algorithm that simultaneously preserves CFAR behavior, boasts lower computational complexity, attains super-resolution capabilities, and ensures high estimation accuracy. These identified gaps serve as the motivations for our current work.
\begin{table}[]
	\centering
        \resizebox{\textwidth}{25mm}{
	\begin{tabular}{c  c | c c}
		\hline
		ADC &  Analog-to-Digital Converter &  LSE\&D & line spectral estimation and detection  \\
		SNR & Signal-to-noise-ratio & CFAR & constant false alarm rate\\
		NOMP & Newtonized orthogonal matching pursuit&GNOMP & generalized Newtonized orthogonal matching pursuit\\
		FIM & Fisher information matrix & CRB &  Cram\'{e}r Rao bound   \\
		FFT &  fast Fourier transform &  IFFT  & inverse  fast Fourier transform  \\
		DR&dynamic range&MLE&maximum likelihood estimator\\
		BIC & Bayesian information criterion&CDF &  Cumulative distribution function\\
		AM & alternating minimization & ANM &  atomic norm minimization  \\
		1bRELAX& one bit RELAX& 1bMMRELAX&one bit majorization-minimization based RELAX\\
		VALSE-EP & varational line spectral estimation and expectation propagation   & MVALSE-EP & multisnapshot VALSE-EP \\
        AST & atomic norm soft-thresholding & FMCW & frequency modulated continuous wave \\
        MIMO & multiple input multiple output & PDF & probability density function \\ \hline
	\end{tabular}}
	\caption{ACRONYMS}
\end{table}
\subsection{Main Contributions}
In this paper, the LSE\&D have been thoroughly investigated, and the main contributions can be summarized as follows: Firstly, a comprehensive examination of the inherent DR limitations associated with a receiver employing low-resolution ADCs, particularly in the case of 1-bit ADCs, has been conducted. The study reveals that in scenarios involving multiple sinusoids, the SNR loss of the weakest signal is attributed to both quantization effects and intersinusoidal interference, considered as the synthesis of the remaining sinusoids. The Rao test is proposed and its false alarm probability and detection probability are provided. It is also shown that the proposed Rao test can be efficiently implemented through FFT and inverse FFT (IFFT). Secondly, inspired by the remarkable super-resolution and rapid performance of Newtonalized orthogonal matching pursuit (NOMP) \cite{Madhow16TSP}, a fast, superresolution and CFAR perserved generalized NOMP (GNOMP) is proposed \footnote{It's noteworthy that the GNOMP can be readily extended to handle various settings, such as compressive measurement scenarios, multisnapshot measurement scenarios, multidimensional LSE, signed measurements from time-varying thresholds, measurements from mixed-resolution ADCs, unknown noise variance scenarios, etc. Although these details are not explicitly presented, the GNOMP is applied to address scenarios involving signed measurements from time-varying thresholds, unknown noise variance and random linear measurements in numerical simulations.}. Finally, extensive numerical simulations and real experiments are conducted to verify the theoretical results and showcase the efficiency and outstanding performance of the GNOMP. The comparison is made against theoretical bounds in terms of estimation and detection, as well as state-of-the-art algorithms. The results demonstrate the effectiveness of the GNOMP in various scenarios.

The remainder of this article is structured as follows. In Section \ref{ProbSet}, we introduce the signal model. Section \ref{Signalsignal} delves into the study of the effects of low-resolution quantization in a single-signal scenario with nonidentical thresholds, and its relationship with multiple sinusoids scenarios is discussed. The proposed GNOMP approach is detailed in Section \ref{GNOMP}. Subsequently, in Section \ref{NS}, we present substantial numerical experiments to illustrate the frequency estimation and detection performances of GNOMP. The efficacy of GNOMP is further demonstrated using real data in Section \ref{RE}. Finally, Section \ref{Cons} provides the conclusion for this article.


For a complex vector ${\mathbf x}\in {\mathbb C}^M$, let ${\Re}\{{\mathbf x}\}$ and $\Im\{\mathbf x\}$ denote the real and imaginary part of $\mathbf x$, respectively. For the square matrix $\mathbf A$, let ${\rm diag}(\mathbf A)$ return a vector with elements being the diagonal of $\mathbf A$. Meanwhile, for a vector ${\mathbf a}$, let ${\rm diag}(\mathbf a)$ return a diagonal matrix with the diagonal being $\mathbf a$. Thus ${\rm diag}({\rm diag}(\mathbf A))$ returns a diagonal matrix. Let ${\rm j}$ denote the imaginary unit. For the matrix $\mathbf A\in\mathbb C^{N\times N}$, let ${\mathbf A}^{*}$, ${\mathbf A}^{\rm T}$ and ${\mathbf A}^{\rm H}$ be the conjugate, transpose and Hermitian transpose operator of $\mathbf A$, respectively. For the matrix $\mathbf A$, let $|\mathbf A|$ denote the elementwise absolute value of $\mathbf A$. Let $\mathbf I_L$ denote the identity matrix of dimension $L$. Let ${\mathcal {CN}}({\mathbf x};{\boldsymbol \mu},{\boldsymbol \Sigma})$ denote the complex normal (CN) distribution of ${\mathbf x}$ with mean ${\boldsymbol \mu}$ and covariance ${\boldsymbol \Sigma}$. Let $\phi(x)=\exp(-x^2/2)/{\sqrt{2\pi}}$ and $\Phi(x)=\int_{-\infty}^x\phi(t){\rm d}t$ denote the standard normal probability density function (PDF) and cumulative distribution function (CDF), respectively. For reading convenience, Table I shows the acronyms employed throughout this paper.

\section{Problem Setup}\label{ProbSet}
The LSE\&D from low-resolution quantized samples is formulated as
\begin{align}\label{defGLMCSLSE}
	\mathbf{y}=\mathcal{Q}_{\rm C}\left(\sum\limits_{k=1}^K{\mathbf a}(\omega_k)x_k+{\boldsymbol \epsilon}\right),
\end{align}
where
\begin{align}\label{defatom}
{\mathbf{a}}({\omega })=[1,{\rm e}^{{\rm j}\omega},\cdots,{\rm e}^{{\rm j}(N-1)\omega}]^{\rm T}
\end{align}
is the atom or array manifold vector, $\mathcal{Q}_{\rm C}({\mathbf x})=\mathcal{Q}(\Re\{{\mathbf x}\})+{\rm j}\mathcal{Q}(\Im\{{\mathbf x}\})$, ${\mathbf y}\in{\mathbb C}^N$ are the measurements, ${\boldsymbol\epsilon}\in{\mathbb C}^N$ is the additive white Gaussian noise (AWGN) satisfying ${\boldsymbol\epsilon}\sim {\mathcal {CN}}({\mathbf 0},\sigma^2{\mathbf I}_N)$, $\sigma^2$ is the variance of the noise, $N$ is the number of measurements, ${\mathcal Q}(\cdot)$ is a uniform/nonuniform quantizer which maps the continuous-valued observations into a finite number of bits. A uniform quantizer with bit-depth $B$ is used, defined as \cite{Nir2019}:
\begin{align}
{\mathcal Q}(x)=
\begin{cases}
&-\gamma+\frac{2\gamma}{b}\left(d+\frac{1}{2}\right),\quad x-d\frac{2\gamma}{b}+\gamma\in[0,\frac{2\gamma}{b}],\\&\qquad \quad \qquad\quad\qquad\quad d\in \{0,1,\cdots,b-1\},\\
&{\rm sign}(x)\left(\gamma-\frac{\gamma}{b}\right),\qquad |x|>\gamma,
\end{cases}
\end{align}
where $b=2^B$ is the cardinality of the output of the quantizer, $\gamma$ is the maximum full-scale range. Assume that the quantization intervals for the quantizer ${\mathcal Q}(\cdot)$ are $\{(\tau_d,\tau_{d+1})\}_{d=0}^{b-1}$, where $\tau_0=-\infty$, $\tau_d=d\frac{2\gamma}{b}-\gamma$, $d=1,2,\cdots,b-1$, $\tau_{b}=\infty$. For example, one-bit quantization refers to $B=1$, $b=2$, $\tau_0=-\infty$, $\tau_1=0$, $\tau_{2}=\infty$, and ${\mathcal Q}(\cdot)$ reduces to the signum function, i.e., ${\mathcal Q}(\cdot)={\rm sign}(\cdot)\gamma/2$. To better describe the quantizer, we define two functions $l(\cdot)$ and $u(\cdot)$ which return the componentwise lower thresholds $l({\mathbf y})$ and upper thresholds $u({\mathbf y})$ of the measurements ${\mathbf y}$. For example, $u\left(-\gamma+\frac{2\gamma}{b}\left(d+\frac{1}{2}\right)\right)=(d+1)\frac{2\gamma}{b}-\gamma$ and $l\left(-\gamma+\frac{2\gamma}{b}\left(d+\frac{1}{2}\right)\right)=d\frac{2\gamma}{b}-\gamma$.

The goal of this work is to perform line spectrum estimation and detection from the coarsely quantized measurements $\mathbf{y}$, i.e., inferring the unknown parameters $\{\omega_k\}_{k=1}^K$, $\{x_k\}_{k=1}^K$, and $K$, while maintaining the CFAR property.
\section{A Single Sinusoid Estimation and Detection with Nonzero Thresholds}\label{Signalsignal}
Directly analyzing the theoretical estimation and detection performance limits in a multiple targets scenario is a very challenging task due to intersinusoidal interference. Therefore, we adopt an oracle theoretical perspective: For a given sinusoid with frequency $\omega_k$ and amplitude $x_k$ that we aim to estimate and detect, assuming the remaining signals are perfectly known or estimated, this scenario can be likened to a single sinusoidal signal estimation and detection with non-zero thresholds, as defined later in (\ref{BHT}). Here, we set $\omega=\omega_k$, $x=x_k$, and ${\boldsymbol\zeta}=\sum\limits_{k^{'}=1,k^{'}\neq k}^K{\mathbf{a}}({\omega }_{k^{'}}){x_{k^{'}}}$, which can be regarded as the intersinusoidal interference for the current sinusoid under detection in a multiple sinusoid scenario. In this context, the estimation bound for the given sinusoid is expected to be lower than that in the true model (\ref{defGLMCSLSE}). Furthermore, the detection probability of the given sinusoid in such a scenario provides an upper bound for the detection probability in the true model (\ref{defGLMCSLSE}). It is anticipated that the estimation and detection bounds will be accurate provided that the designed algorithm exhibits very high estimation accuracy, as demonstrated in Sec.\ref{NS}.

From an algorithmic perspective, the proposed GNOMP, detailed in Section \ref{GNOMP}, follows a greedy approach by sequentially estimating the sinusoids. Ideally, GNOMP begins by detecting the strongest signal while ignoring the presence of other signals. Subsequently, GNOMP incorporates the first detected strong signal and assesses the existence of the second signal. Assuming GNOMP exhibits high estimation accuracy, detecting the presence of the current signal can be modeled as a binary hypothesis testing problem, as defined later in (\ref{BHT}), where the thresholds can be regarded as a synthesis of the estimated/detected signals. Consequently, it is meaningful to explore single sinusoid estimation and detection with non-zero thresholds.

This section introduces the mathematical model for the estimation and detection of a single sinusoid with nonzero thresholds in the case of a known frequency. Subsequently, the results are extended to address the scenario where the frequency is unknown.
\subsection{Complex Amplitude Unknown}
Consider the following binary hypothesis testing (BHT) problem
\begin{align}\label{BHT}
	\begin{cases}
		{\mathcal H}_0:{\mathbf y}=\mathcal{Q}_{\rm C}\left({\boldsymbol\zeta}+{\boldsymbol\epsilon}\right),\\
		{\mathcal H}_1:{\mathbf y}=\mathcal{Q}_{\rm C}\left({\boldsymbol\zeta}+{\mathbf{a}}({\omega }){{x}}+{\boldsymbol\epsilon}\right),
	\end{cases}
\end{align}
where ${\boldsymbol\zeta}\in{\mathbb C}^N$ are the nonzero thresholds. For the known frequency $\omega$ case, let ${\boldsymbol\theta}=[\Re\{x\},\Im\{x\}]^{\rm T}$ denote the unknown deterministic parameters and ${\boldsymbol\theta}_0={\mathbf 0}_{2\times 1}$. Then one could obtain the generalized likelihood ratio test (GLRT), Wald test and Rao test \cite{KayDet}. Since the asymptotic performances of GLRT and Wald test are the same as that of Rao test, and both GLRT and Wald Test depend on the ML estimate $\hat{\boldsymbol\theta}_1$ under hypothesis ${\mathcal H}_1$, making the computation complexity high especially when the frequency is unknown. We focus on the Rao test, similar to \cite{onebitMIMOradaedet} but only under the zero threshold case. Note that the Rao Test depends on the Fisher information matrix (FIM) under hypothesis ${\mathcal H}_0$, and the asymptotic distribution of Rao test depends on the FIM under ${\mathcal H}_1$ \cite{KayEst,KayDet}, we provide the FIM in the following proposition.
\begin{proposition}
For the quantization with bit-depth $B$, the FIM ${\mathbf I}_B({\boldsymbol\theta})$ and the CRB ${\mathbf I}_B^{-1}({\boldsymbol\theta})$ are
\begin{align}\label{FIMsingle0}
{\mathbf I}_B({\boldsymbol\theta})&=\frac{2}{\sigma^2}\left({\mathbf a}^{\rm H}{\rm diag}({\mathbf h}_{B,+}({\mathbf a}{x}+{\boldsymbol\zeta})){\mathbf a}
\left[
  \begin{array}{cc}
    1 & 0 \\
    0 & 1\\
  \end{array}\right]
+\Re\{{\mathbf a}^{\rm T}{\rm diag}({\mathbf h}_{B,-}({\mathbf a}{x}+{\boldsymbol\zeta})){\mathbf a}\}
\left[
  \begin{array}{cc}
    1 & 0 \\
    0 & -1\\
  \end{array}\right]\right.\notag\\
&\left.
-\Im\{{\mathbf a}^{\rm T}{\rm diag}({\mathbf h}_{B,-}({\mathbf a}{x}+{\boldsymbol\zeta})){\mathbf a}\}
\left[
  \begin{array}{cc}
    0 & 1 \\
    1 & 0\\
  \end{array}\right]\right),
\end{align}
and
\begin{align}\label{CRBsingle0}
&{\mathbf I}_B^{-1}({\boldsymbol\theta})=\frac{\sigma^2}{2\left(\left({\mathbf a}^{\rm H}{\rm diag}({\mathbf h}_{B,+}({\mathbf a}{x}+{\boldsymbol\zeta})){\mathbf a}\right)^2-\Re^2\{\left({\mathbf a}^{\rm T}{\rm diag}({\mathbf h}_{B,-}({\mathbf a}{x}+{\boldsymbol\zeta})){\mathbf a}\right)\}-\Im^2\{\left({\mathbf a}^{\rm T}{\rm diag}({\mathbf h}_{B,-}({\mathbf a}{x}+{\boldsymbol\zeta})){\mathbf a}\right)\}\right)}\notag\\
&\left({\mathbf a}^{\rm H}{\rm diag}({\mathbf h}_{B,+}({\mathbf a}{x}+{\boldsymbol\zeta})){\mathbf a}
\left[
  \begin{array}{cc}
    1 & 0 \\
    0 & 1\\
  \end{array}\right]
-\Re\{{\mathbf a}^{\rm T}{\rm diag}({\mathbf h}_{B,-}({\mathbf a}{x}+{\boldsymbol\zeta})){\mathbf a}\}
\left[
  \begin{array}{cc}
    1 & 0 \\
    0 & -1\\
  \end{array}\right]\right.\notag\\
&\left.
+\Im\{{\mathbf a}^{\rm T}{\rm diag}({\mathbf h}_{B,-}({\mathbf a}{x}+{\boldsymbol\zeta})){\mathbf a}\}
\left[
  \begin{array}{cc}
    0 & 1 \\
    1 & 0\\
  \end{array}\right]\right),
\end{align}
respectively, where
\begin{subequations}
	\begin{align}
		{\mathbf h}_{B,+}({\boldsymbol \eta})=\frac{h_B(\Re\{{\boldsymbol \eta}\},{\sigma}^2)+h_B(\Im\{{\boldsymbol \eta}\},{\sigma}^2)}{2},\label{hBplus}\\
		{\mathbf h}_{B,-}({\boldsymbol \eta})=\frac{h_B(\Re\{{\boldsymbol \eta}\},{\sigma}^2)-h_B(\Im\{{\boldsymbol \eta}\},{\sigma}^2)}{2},
	\end{align}
\end{subequations}
$h_B(x,\sigma^2)$ is given by
\begin{align}\label{hxsigma}
	h_B(x,\sigma^2)=\sum\limits_{d=0}^{b-1}
	\frac{\left[\phi\left(\frac{\tau_{d+1}-x}{\sigma/\sqrt{2}}\right)-\phi\left(\frac{\tau_{d}-x}{\sigma/\sqrt{2}}\right)\right]^2}
	{\Phi\left(\frac{\tau_{d+1}-x}{\sigma/\sqrt{2}}\right)-\Phi\left(\frac{\tau_{d}-x}{\sigma/\sqrt{2}}\right)},
\end{align}
and $h_B({\mathbf x},{\sigma}^2)=[h_B(x_1,{\sigma}^2),h_B(x_2,{\sigma}^2),\cdots h_B(x_N,{\sigma}^2)]^{\rm T}$ for ${\mathbf x}\in {\mathbb R}^N$.
Here we use ${\mathbf a}$ instead of ${\mathbf a}(\omega)$ for brevity. Besides, $h_{B=\infty}(x,\sigma^2)=\lim_{B\rightarrow\infty}h_B(x,\sigma^2)=1$, and the FIM ${\mathbf I}_{\infty}({\boldsymbol\theta})$ is
\begin{align}\label{FIMunq}
{\mathbf I}_{\infty}({\boldsymbol\theta})=\frac{2}{\sigma^2}{\mathbf a}^{\rm H}{\mathbf a}
\left[
  \begin{array}{cc}
    1 & 0 \\
    0 & 1\\
  \end{array}\right]=\frac{2N}{\sigma^2}\left[
  \begin{array}{cc}
    1 & 0 \\
    0 & 1\\
  \end{array}\right].
\end{align}
For one-bit quantization where $B=1$, $b=2^B=2$, $\tau_0=-\infty$, $\tau_1=0$, $\tau_2=\infty$, $h_{B=1}(x,\sigma^2)$ simplifies to be
\begin{align}\label{onebith}
h_{B=1}(x,\sigma^2)&=\frac{\phi^2\left(\frac{x}{\sigma/\sqrt{2}}\right)}{\Phi\left(\frac{x}{\sigma/\sqrt{2}}\right)\Phi\left(\frac{-x}{\sigma/\sqrt{2}}\right)}=\frac{1}{2\pi}\frac{{\rm e}^{-\frac{2x^2}{\sigma^2}}}{\Phi\left(\frac{x}{\sigma/\sqrt{2}}\right)\Phi\left(\frac{-x}{\sigma/\sqrt{2}}\right)}.
\end{align}
\end{proposition}
\begin{proof}
The proof is postponed to Appendix \ref{FIM0single}.
\end{proof}
\subsubsection{Estimation Performance}
We investigate the estimation performance by evaluating the CRB which equals to the trace of ${\mathbf I}_{B}^{-1}({\boldsymbol\theta})$, i.e., ${\rm tr}\left({\mathbf I}_B^{-1}({\boldsymbol\theta})\right)$ \cite{KayEst}.
\subsubsection{Detection Performance}
The Rao test $T_{{\rm R},B}(\mathbf y,{\boldsymbol\zeta},\omega)$ decides ${\mathcal H}_1$ if \cite{KayDet}
\begin{small}
	\begin{align}\label{Raoorig}
		T_{{\rm R},B}(\mathbf y,{\boldsymbol\zeta},\omega)&= \left.\frac{\partial \ln p({\mathbf y},{\boldsymbol\zeta};{\boldsymbol \theta})}
		{\partial {\boldsymbol\theta}}\right|_{{\boldsymbol\theta}={\boldsymbol\theta}_0}^{\rm T}
		{\mathbf I}_B^{-1}({\boldsymbol\theta}_0) \left.\frac{\partial \ln p({\mathbf y},{\boldsymbol\zeta};{\boldsymbol \theta})}
		{\partial {\boldsymbol\theta}}\right|_{{\boldsymbol\theta}={\boldsymbol\theta}_0}\notag\\
		&\geq \gamma_{{\rm th}, B},
	\end{align}
\end{small}
where
${\boldsymbol\theta}=[\Re\{x\},\Im\{x\}]^{\rm T}$, ${\boldsymbol\theta}_0=[0,0]^{\rm T}$,
\begin{small}
\begin{align}\label{lnpdef}
    \ln p({\mathbf y},{\boldsymbol\zeta};{\boldsymbol \theta})=&\sum\limits_{n=1}^N\left(\log\left(\Phi\left(\frac{\Re\{u(y_n)\}-\Re\{\zeta_n+{a}_n(\omega)x\}}{\sigma/{\sqrt{2}}}\right)-\Phi\left(\frac{\Re\{l(y_n)\}-\Re\{\zeta_n+{ a}_n(\omega)x\}}{\sigma/{\sqrt{2}}}\right)\right)\right.\notag\\
    &\left.+\log\left(\Phi\left(\frac{\Im\{u(y_n)\}-\Im\{\zeta_n+{ a}_n(\omega)x\}}{\sigma/{\sqrt{2}}}\right)-\Phi\left(\frac{\Im\{l(y_n)\}-\Im\{\zeta_n+{ a}_n(\omega)x\}}{\sigma/{\sqrt{2}}}\right)\right)\right),
\end{align}
\end{small}
\begin{small}
    \begin{align}\label{lnppart}
        \frac{\partial \ln p({\mathbf y},{\boldsymbol\zeta};{\boldsymbol \theta})}
        {\partial {\boldsymbol\theta}}
        &=\sum\limits_{n=1}^N-\frac{\phi\left(\frac{\Re\{u(y_n)\}-\Re\{\zeta_n+{a}_n(\omega)x\}}{\sigma/{\sqrt{2}}}\right)-\phi\left(\frac{\Re\{l(y_n)\}-\Re\{\zeta_n+{a}_n(\omega)x\}}{\sigma/{\sqrt{2}}}\right)}{\Phi\left(\frac{\Re\{u(y_n)\}-\Re\{\zeta_n+{a}_n(\omega)x\}}{\sigma/{\sqrt{2}}}\right)-\Phi\left(\frac{\Re\{l(y_n)\}-\Re\{\zeta_n+{a}_n(\omega)x\}}{\sigma/{\sqrt{2}}}\right)}\frac{[\Re\{{a}_n(\omega)\},-\Im\{{a}_n(\omega)\}]^{\rm T}}{\sigma/{\sqrt{2}}}\notag\\
        &-\sum\limits_{n=1}^N\frac{\phi\left(\frac{\Im\{u(y_n)\}-\Im\{\zeta_n+{a}_n(\omega)x\}}{\sigma/{\sqrt{2}}}\right)-\phi\left(\frac{\Im\{l(y_n)\}-\Im\{\zeta_n+{a}_n(\omega)x\}}{\sigma/{\sqrt{2}}}\right)}{\Phi\left(\frac{\Im\{u(y_n)\}-\Im\{\zeta_n+{a}_n(\omega)x\}}{\sigma/{\sqrt{2}}}\right)-\Phi\left(\frac{\Im\{l(y_n)\}-\Im\{\zeta_n+{a}_n(\omega)x\}}{\sigma/{\sqrt{2}}}\right)}\frac{[\Im\{{a}_n(\omega)\},\Re\{{a}_n(\omega)\}]^{\rm T}}{\sigma/{\sqrt{2}}},
    \end{align}
\end{small}
and $\gamma_{{\rm th}, B}$ is the detection threshold.


Now we simplify the Rao test $T_{{\rm R},B}({\mathbf y},{\boldsymbol\zeta},\omega)$. We first compute ${\mathbf I}_B^{-1}({\boldsymbol\theta}_0)$ and $\frac{\partial \ln p({\mathbf y},{\boldsymbol\zeta};{\boldsymbol \theta})}
{\partial {\boldsymbol\theta}}|_{{\boldsymbol\theta}={\boldsymbol\theta}_0}$. According to (\ref{CRBsingle0}), ${\mathbf I}_B^{-1}({\boldsymbol\theta}_0)$ is 
\begin{small}
    \begin{align}\label{Iinvtheta0}
        &{\mathbf I}_B^{-1}({\boldsymbol\theta}_0)=\frac{\sigma^2}{2\left(\left({\mathbf a}^{\rm H}{\rm diag}({\mathbf h}_{B,+}({\boldsymbol\zeta})){\mathbf a}\right)^2-\Re^2\{\left({\mathbf a}^{\rm T}{\rm diag}({\mathbf h}_{B,-}({\boldsymbol\zeta})){\mathbf a}\right)\}-\Im^2\{\left({\mathbf a}^{\rm T}{\rm diag}({\mathbf h}_{B,-}({\boldsymbol\zeta})){\mathbf a}\right)\}\right)}\notag\\
        &\left({\mathbf a}^{\rm H}{\rm diag}({\mathbf h}_{B,+}({\boldsymbol\zeta})){\mathbf a}
        \left[
        \begin{array}{cc}
            1 & 0 \\
            0 & 1\\
        \end{array}\right]
        -\Re\{{\mathbf a}^{\rm T}{\rm diag}({\mathbf h}_{B,-}({\boldsymbol\zeta})){\mathbf a}\}
        \left[
        \begin{array}{cc}
            1 & 0 \\
            0 & -1\\
        \end{array}\right]
        +\Im\{{\mathbf a}^{\rm T}{\rm diag}({\mathbf h}_{B,-}({\boldsymbol\zeta})){\mathbf a}\}
        \left[
        \begin{array}{cc}
            0 & 1 \\
            1 & 0\\
        \end{array}\right]\right).
    \end{align}
\end{small}
$\frac{\partial \ln p({\mathbf y},{\boldsymbol\zeta};{\boldsymbol \theta})}
{\partial {\boldsymbol\theta}}\big|_{{\boldsymbol\theta}={\boldsymbol\theta}_0}$ is
\begin{align}\label{lnptheta0}
\begin{split}
\left.\frac{\partial \ln p({\mathbf y},{\boldsymbol\zeta};{\boldsymbol \theta})}
{\partial {\boldsymbol\theta}}\right|_{{\boldsymbol\theta}={\boldsymbol\theta}_0}
=-\frac{1}{\sigma/{\sqrt{2}}}
\left[
  \begin{array}{cc}
    \Re\{{{\mathbf a}^{\rm H}}{\boldsymbol \varphi}_B({\mathbf y},{\boldsymbol\zeta},\sigma)\}\\
    \Im\{{{\mathbf a}^{\rm H}}{\boldsymbol \varphi}_B({\mathbf y},{\boldsymbol\zeta},\sigma)\}
  \end{array}\right],
\end{split}
\end{align}
where ${\boldsymbol \varphi}_B({\mathbf y},{\boldsymbol\zeta},\sigma)=[[{\boldsymbol \varphi}_B({\mathbf y},{\boldsymbol\zeta},\sigma)]_1,[{\boldsymbol \varphi}_B({\mathbf y},{\boldsymbol\zeta},\sigma)]_2,\cdots,[{\boldsymbol \varphi}_B({\mathbf y},{\boldsymbol\zeta},\sigma)]_N]^{\rm T}$ and
\begin{align}\label{verphingeneral}
	[{\boldsymbol \varphi}_B({\mathbf y},{\boldsymbol\zeta},\sigma)]_n =& -\frac{\phi\left(\frac{\Re\{u(y_n)-\zeta_n\}}
		{\sigma/{\sqrt{2}}}\right)-\phi\left(\frac{\Re\{l(y_n)-\zeta_n\}}
		{\sigma/{\sqrt{2}}}\right)}{\Phi\left(\frac{\Re\{u(y_n)-\zeta_n\}}
		{\sigma/{\sqrt{2}}}\right)-\Phi\left(\frac{\Re\{l(y_n)-\zeta_n\}}
		{\sigma/{\sqrt{2}}}\right)}-{\rm j}\frac{\phi\left(\frac{\Im\{u(y_n)-\zeta_n\}}
		{\sigma/{\sqrt{2}}}\right)-\phi\left(\frac{\Im\{l(y_n)-\zeta_n\}}
		{\sigma/{\sqrt{2}}}\right)}{\Phi\left(\frac{\Im\{u(y_n)-\zeta_n\}}
		{\sigma/{\sqrt{2}}}\right)-\Phi\left(\frac{\Im\{l(y_n)-\zeta_n\}}
		{\sigma/{\sqrt{2}}}\right)}.
\end{align}
For simplicity, we use ${\boldsymbol \varphi}_B$ instead of ${\boldsymbol \varphi}_B({\mathbf y},{\boldsymbol\zeta},\sigma)$ for brevity. Inserting (\ref{Iinvtheta0}) and (\ref{lnptheta0}) into (\ref{Raoorig}) yields the simplified Rao test as
\begin{small}
    \begin{align}\label{generalRao}
        T_{{\rm R},B}(\mathbf y,{\boldsymbol\zeta},\omega)
        =&\left.\frac{\partial \ln p({\mathbf y},{\boldsymbol\zeta};{\boldsymbol \theta})}
        {\partial {\boldsymbol\theta}}\right|_{{\boldsymbol\theta}={\boldsymbol\theta}_0}^{\rm T}
        {\mathbf I}_B^{-1}({\boldsymbol\theta}_0) \left.\frac{\partial \ln p({\mathbf y},{\boldsymbol\zeta};{\boldsymbol \theta})}
        {\partial {\boldsymbol\theta}}\right|_{{\boldsymbol\theta}={\boldsymbol\theta}_0}\notag \\
        =&\frac{\overbrace{\left({\mathbf a}^{\rm H}{\rm diag}({\mathbf h}_{B,+}({\boldsymbol\zeta})){\mathbf a}\right)|{{\mathbf a}^{\rm H}}{\boldsymbol \varphi}_B|^2}^{\text{the first term}}
            -\overbrace{\Re\{\left({\mathbf a}^{\rm T}{\rm diag}({\mathbf h}_{B,-}({\boldsymbol\zeta})){\mathbf a}\right){({{\mathbf a}^{\rm H}}{\boldsymbol \varphi}_B)^2}\}}^{\text{the second term}}}{\left({\mathbf a}^{\rm H}{\rm diag}({\mathbf h}_{B,+}({\boldsymbol\zeta})){\mathbf a}\right)^2-\vert {\mathbf a}^{\rm T}{\rm diag}({\mathbf h}_{B,-}({\boldsymbol\zeta})){\mathbf a} \vert^2}
        \triangleq T_{{\rm R},B}(\boldsymbol{\varphi}_B,{\boldsymbol\zeta},\omega)
    \end{align}
\end{small}
Note that the Rao test $T_{{\rm R},B}(\mathbf y,{\boldsymbol\zeta},\omega)$ consists of two terms: The first term can be viewed as the MF based test with pseudo measurements ${\boldsymbol\varphi}_B$, the second term can be viewed as a regularization term due to the nonidentical thresholds in Inphase/Quadrature (I/Q) channels. By viewing ${\boldsymbol\varphi}_B$ as the pseudo measurements, we can use $T_{{\rm R},B}(\boldsymbol{\varphi}_B,{\boldsymbol\zeta},\omega)$ shown in (\ref{generalRao}) to represent the Rao test $T_{{\rm R},B}(\mathbf y,{\boldsymbol\zeta},\omega)$. When $\mathbf a$ is an array manifold vector, we can use $\mathbf a^{\rm H}\mathbf a=N$ to further simplify the Rao test $T_{{\rm R},B}(\boldsymbol{\varphi}_B,{\boldsymbol\zeta},\omega)$ as
\begin{small}
\begin{align}\label{simpleRAO}
&T_{{\rm R},B}(\boldsymbol{\varphi}_B,{\boldsymbol\zeta},\omega)\notag \\
&=\frac{\left({\mathbf 1}^{\rm T}{\mathbf h}_{B,+}({\boldsymbol\zeta})\right)|{{\mathbf a}^{\rm H}}{\boldsymbol \varphi}_B|^2
				-\Re\{\left({\mathbf a}^{\rm T}{\rm diag}({\mathbf h}_{B,-}({\boldsymbol\zeta})){\mathbf a}\right){({{\mathbf a}^{\rm H}}{\boldsymbol \varphi}_B)^2}\}}{\left({\mathbf 1}^{\rm T}{\mathbf h}_{B,+}({\boldsymbol\zeta})\right)^2-\vert {\mathbf a}^{\rm T}{\rm diag}({\mathbf h}_{B,-}({\boldsymbol\zeta})){\mathbf a} \vert^2}.
\end{align}
\end{small}
As stated in \cite{KayDet}, when the data record is large and the signal is weak, and when the MLE attains its asymptotic PDF, the Rao test $T_{{\rm R},B}(\boldsymbol{\varphi}_B,{\boldsymbol\zeta},\omega)$ has the PDF \cite{KayDet}
\begin{align}
T_{{\rm R},B}(\boldsymbol{\varphi}_B,{\boldsymbol\zeta},\omega)\stackrel{a}{\sim}\begin{cases}
\chi_{2}^2, ~~~~~{\rm under}~{\mathcal H}_0,\\
\chi_{2}^{\prime 2}(\lambda_{{\rm R},B}), ~{\rm under}~{\mathcal H}_1,
\end{cases}
\end{align}
where $\chi_{\nu}^2$ and $\chi_{\nu}^{\prime 2}(\lambda)$ denote the chi-squared PDF with $\nu$ degrees of freedom and the noncentral chi-squared PDF with $\nu$ degrees of freedom and noncentrality parameter $\lambda$, respectively, $\stackrel{a}{\sim}$ denotes an asymptotic PDF, $\lambda_{{\rm R},B}$ reduces to
\begin{align}\label{lambdageneralRB}
    \lambda_{{\rm R},B}&=[\Re\{x\},\Im\{x\}]{\mathbf I}_B({\boldsymbol\theta}_0)[\Re\{x\},\Im\{x\}]^{\rm T}\notag\\
    &=\frac{2}{\sigma^2}{\mathbf a}^{\rm H}{\rm diag}({\mathbf h}_{B,+}({\boldsymbol\zeta})){\mathbf a}({\boldsymbol\zeta})|x|^2+\frac{2}{\sigma^2}\Re\{{\mathbf a}^{\rm T}{\rm diag}({\mathbf h}_{B,-}({\boldsymbol\zeta})){\mathbf a}{{x}^2}\}.
\end{align}
Consequently, the false alarm probability ${\rm P}_{{\rm FA},B}$ and detection probability ${\rm P}_{{\rm D},B}$ are
\begin{align}\label{PFAPDknownfre}
	{\rm P}_{{\rm FA},B} &= \int_{\gamma_{{\rm th}, B}}^\infty\frac{1}{2}{\rm e}^{-\frac{x}{2}}{\rm d}x
	={\rm e}^{-\frac{\gamma_{{\rm th}, B}}{2}},\\
	{\rm P}_{{\rm D},B}&=Q_1\left(\sqrt{\lambda_{{\rm R},B}},\sqrt{\gamma_{{\rm th}, B}}\right)=Q_1\left(\sqrt{\lambda_{{\rm R},B}},\sqrt{-2\ln{{\rm P}_{{\rm FA},B}}}\right),
\end{align}
due to $\gamma_{{\rm th}, B} = -2\ln{{\rm P}_{{\rm FA},B}}$, where $Q_1(\cdot,\cdot)$ is the Marcum Q-function.

We now hope to provide some insights to reveal the relationship between the proposed general Rao detector (\ref{generalRao}), the detector without quantization and the detector under $1$ bit quantization. For unquantized measurements one has ${\mathbf h}_{\infty,+}({\boldsymbol\zeta})={\mathbf 1}_N$
and ${\mathbf h}_{\infty,-}({\boldsymbol\zeta},\omega)={\mathbf 0}_N$. $T_{{\rm R},\infty}(\boldsymbol{\varphi}_\infty,{\boldsymbol\zeta},\omega)$ (\ref{generalRao}) reduces to
\begin{align}\label{unqRaodet}
T_{{\rm R},\infty}(\boldsymbol{\varphi}_\infty,{\boldsymbol\zeta},\omega)=\frac{1}{{\mathbf a}^{\rm H}{\mathbf a}}|{{\mathbf a}^{\rm H}}{\boldsymbol \varphi}_{\infty}|^2.
\end{align}
According to
\begin{align}
    {\lim\limits_{\Delta \to 0^{+}}\frac{\phi(x+\Delta)-\phi(x)}{\Phi(x+\Delta)-\Phi(x)}}
    =\frac{\phi'(x)}{\phi(x)}=-x,
\end{align}
we let $B\rightarrow \infty$ in (\ref{verphingeneral}) and simplify ${\boldsymbol \varphi}_{\infty}$ as
\begin{align}
    {\boldsymbol \varphi}_{\infty} = \frac{\mathbf y-\boldsymbol{\zeta}}{\sigma/{\sqrt{2}}}.
\end{align}
Therefore $T_{{\rm R},{\infty}}(\boldsymbol{\varphi}_\infty,{\boldsymbol\zeta},\omega)$ (\ref{unqRaodet}) is simplified to be
\begin{align}\label{Raounq}
T_{{\rm R},{\infty}}(\boldsymbol{\varphi}_\infty,{\boldsymbol\zeta},\omega)=\frac{1}{{\mathbf a}^{\rm H}{\mathbf a}}\left| {{\mathbf a}^{\rm H}}{\frac{\mathbf y-\boldsymbol{\zeta}}
{\sigma/{\sqrt{2}}}} \right|^2=\frac{2}{\sigma^2}\frac{1}{{\mathbf a}^{\rm H}{\mathbf a}}\vert {{\mathbf a}^{\rm H}}\left({\mathbf y-\boldsymbol{\zeta}}\right)
 \vert^2.
\end{align}
In addition, $\lambda_{{\rm R},\infty}$ (\ref{lambdageneralRB}) is simplified to be
\begin{align}\label{lambdaunq}
    \lambda_{{\rm R},\infty}=\frac{2}{\sigma^2}{\mathbf a}^{\rm H}{\mathbf a}|x|^2.
\end{align}
It can be concluded that the Rao test $T_{{\rm R},\infty}(\boldsymbol{\varphi}_\infty,\boldsymbol{\zeta},\omega)$ (\ref{Raounq}) and the concentrality parameter $\lambda_{{\rm R},\infty}$ (\ref{lambdaunq}) are consistent with the results directly obtained with the unquantized measurement model \cite{Madhow16TSP}. Besides, the effects caused by the nonzero thresholds can easily be cancelled and the asymptotic distribution under either hypothesis ${\mathcal H}_0$ or hypothesis ${\mathcal H}_1$ is irrelevant with respect to the thresholds $\boldsymbol{\zeta}$. This reflects that in the oracle setting where the thresholds $\boldsymbol{\zeta}$ are perfectly known, estimating and detecting a single sinusoidal signal with a general $\boldsymbol{\zeta}$ is equivalent to estimating and detecting a pure sinusoidal signal, which makes sense.

For quantized measurement model, in general, ${\mathbf h}_{B,-}({\boldsymbol\zeta})$ is usually very small, compared to ${\mathbf h}_{B,+}({\boldsymbol\zeta})$. Dropping those terms involved with ${\mathbf h}_{B,-}({\boldsymbol\zeta})$ or letting ${\mathbf h}_{B,-}({\boldsymbol\zeta})={\mathbf 0}$, a simplified Rao test $T_{{\rm R},B}^{\prime}({\boldsymbol{\varphi}_B},{\boldsymbol\zeta},\omega)$
\begin{align}\label{TRprimey}
	T_{{\rm R},B}^{\prime}(\boldsymbol{\varphi}_B,{\boldsymbol\zeta},\omega)&=\frac{1}{{\mathbf a}^{\rm H}{\rm diag}({\mathbf h}_{B,+}({\boldsymbol\zeta})){\mathbf a}}|{{\mathbf a}^{\rm H}}{\boldsymbol \varphi}_B|^2
\end{align}
can be obtained, which is the MF filter with pseudo measurements ${\boldsymbol\varphi}_B$. According to (\ref{lambdageneralRB}), a simplified $\lambda_{{\rm R},B}^{\prime}$ omitting ${\mathbf h}_{B,-}({\boldsymbol\zeta})$ is calculated to be
\begin{align}
\lambda_{{\rm R},B}^{\prime}=\frac{2}{\sigma^2}\left({\mathbf a}^{\rm H}{\rm diag}({\mathbf h}_{B,+}({\boldsymbol\zeta})){\mathbf a}\right)|x|^2
\end{align}
Compared to the unquantized measurements, one can define the SNR loss incurred by the quantization as
\begin{align}\label{SNRloss}
{\rm SNR}_{{\rm loss},B}=\frac{\lambda_{{\rm R},\infty}}{\lambda_{{\rm R},B}^{\prime}}=\frac{N}{{\mathbf a}^{\rm H}{\rm diag}({\mathbf h}_{B,+}({\boldsymbol\zeta})){\mathbf a}}.
\end{align}
It can be seen that the SNR loss ${\rm SNR}_{{\rm loss},B}$ is related to the noise variance $\sigma^2$, the bit-depth $B$ caused by the low resolution quantization and the nonzero thresholds ${\boldsymbol\zeta}$ which can be regarded as the intersinusoidal interference in multiple sinusoids scenario. In unquantized scenario, a target is typically considered reliably detectable when its integrated SNR (time domain SNR $+10\log(N)$) exceeds about 12 dB. While in quantized scenario, taking the SNR loss into account,
it can be shown that a target is reliably detectable when its effective integrated SNR (time domain SNR $+10\log(N)-{\rm SNR}_{{\rm loss},B}$) exceeds about 12 dB.

Next we consider 1-bit quantization. According to (\ref{verphingeneral}), $[{\boldsymbol \varphi}_{B=1}]_n$ is simplified as
\begin{align}
	[{\boldsymbol \varphi}_{B=1}]_n = \frac{{\rm{sign}}(\Re\{y_n\})\phi\left(\frac{\Re\{\zeta_n\}}
		{\sigma/{\sqrt{2}}}\right)}{\Phi\left(\frac{{\rm{sign}}(\Re\{y_n\})\Re\{\zeta_n\}}
		{\sigma/{\sqrt{2}}}\right)}+{\rm j}\frac{{\rm{sign}}(\Im\{y_n\})\phi\left(\frac{\Im\{\zeta_n\}}
		{\sigma/{\sqrt{2}}}\right)}{\Phi\left(\frac{{\rm{sign}}(\Im\{y_n\})\Im\{\zeta_n\}}
		{\sigma/{\sqrt{2}}}\right)}.\notag
\end{align}
We discuss two cases in which $|{\zeta}_n|/\sigma$ is near zero and $|{\zeta}_n|/\sigma$ is very large.
For ${\zeta}_n/\sigma\approx 0$, using the first order Taylor series expansion, $[{\boldsymbol \varphi}_{B=1}]_n$ can be approximated as
\begin{align}
	[{\boldsymbol \varphi}_{B=1}]_n
	&\approx \sqrt{\frac{2}{\pi}}{\rm{sign}}(\Re\{y_n\})\left(1-{\rm{sign}}(\Re\{y_n\})\frac{2\Re\{\zeta_n\}}{\sqrt{\pi}\sigma}\right)\notag\\
	&+{\rm j}\sqrt{\frac{2}{\pi}}{\rm{sign}}(\Im\{y_n\})\left(1-{\rm{sign}}(\Im\{y_n\})\frac{2\Im\{\zeta_n\}}{\sqrt{\pi}\sigma}\right).\notag
\end{align}
It can be seen that $[{\boldsymbol \varphi}_{B=1}]_n\approx \sqrt{\frac{2}{\pi}}({\rm{sign}}(\Re\{y_n\})+{\rm j}{\rm{sign}}(\Im\{y_n\}))$ is independent of noise variance. For $|{\zeta}_n|\gg \sigma$, $[{\boldsymbol \varphi}_{B=1}]_n$ can be approximated as
\begin{small}
	\begin{align}
		[{\boldsymbol \varphi}_{B=1}]_n
		=\frac{{\rm{sign}}(\Re\{y_n\})}{\sqrt{2\pi}}{\rm e}^{-\frac{\Re^2\{\zeta_n\}}{\sigma^2}}
		+{\rm j}\frac{{\rm{sign}}(\Im\{y_n\})}{\sqrt{2\pi}}{\rm e}^{-\frac{\Im^2\{\zeta_n\}}{\sigma^2}},\notag
	\end{align}
\end{small}and the real (or imaginary) part of $[{\boldsymbol \varphi}_{B=1}]_n$ decays very quickly and approaches to $0$ as the absolute value of the real (or imaginary) part of $\zeta_n$ increases. This demonstrates that those measurements with large absolute thresholds compared to the noise standard deviation will not contribute too much for signal detection.

For one-bit quantization, we use the following two approximations \cite{Proakis}
\begin{equation}
\Phi(x)\Phi(-x)\approx \begin{cases}\frac{1}{4}{\rm e}^{-\frac{x^2}{2}},x\leq \sqrt{\frac{8}{\pi}}\\
\frac{1}{\sqrt{2\pi}|x|}{\rm e}^{-x^2/2},x\geq  \sqrt{\frac{8}{\pi}}
\end{cases}
\end{equation}
to approximate $h_{B=1}(x,\sigma^2)$ (\ref{onebith}) as
\begin{equation}\label{Chernoffbound}
h_{B=1}(x,\sigma^2)\approx\begin{cases}
\frac{2}{\pi}{\rm e}^{-x^2/\sigma^2},x\leq \sqrt{\frac{8}{\pi}}\sigma\\
\frac{|x|}{\sqrt{\pi}\sigma}{\rm e}^{-x^2/\sigma^2}, x\geq \sqrt{\frac{8}{\pi}}\sigma.
\end{cases}
\end{equation}
Let ${\mathcal I}_{{\rm R},\leq}\subseteq {\mathcal N}\triangleq \{1,2,\cdots,N\}$ and ${\mathcal I}_{{\rm I},\leq}\subseteq {\mathcal N}$ be the index set such that $\forall n\in {\mathcal I}_{{\rm R},\leq}$, $|\Re\{\zeta_n\}|\leq \sqrt{\frac{8}{\pi}}\sigma$, $\forall n\in {\mathcal I}_{{\rm I},\leq}$, $|\Im\{\zeta_n\}|\leq \sqrt{\frac{8}{\pi}}\sigma$, respectively. Using the approximation (\ref{Chernoffbound}) and the definition ${\mathbf h}_{B,+}({\boldsymbol \eta})$ (\ref{hBplus}), ${\rm SNR}_{{\rm loss},B=1}$ can be approximated as 
\begin{align}\label{SNRloss1bit0}
    {\rm SNR}_{{\rm loss},B=1}&\approx {2N}/\left(\underbrace{\frac{2}{\pi}\sum\limits_{n\in {\mathcal I}_{{\rm R},\leq}}{\rm e}^{-|\Re^2\{\zeta_n\}|/\sigma^2}}_{\text{The first term}}+\underbrace{\frac{2}{\pi}\sum\limits_{n\in {\mathcal I}_{{\rm I},\leq}}{\rm e}^{-|\Im^2\{\zeta_n\}|/\sigma^2}}_{\text{The second term}}\right.\notag\\
    &\left.+\underbrace{\frac{1}{\sqrt{\pi}{\sigma}}\sum\limits_{n\in {\mathcal N}\setminus{\mathcal I}_{{\rm R},\leq}}|\Re\{\zeta_n\}|{\rm e}^{-|\Re^2\{\zeta_n\}|/\sigma^2}}_{\text{The third term}}+\underbrace{\frac{1}{\sqrt{\pi}\sigma}\sum\limits_{n\in {\mathcal N}\setminus{\mathcal I}_{{\rm I},\leq}}|\Im\{\zeta_n\}|{\rm e}^{-|\Im^2\{\zeta_n\}|/\sigma^2}}_{\text{The fourth term}}
    \right).
\end{align}
We now analyze the SNR loss in two extreme cases corresponding to the zero threshold case, i.e., ${\boldsymbol \zeta}={\mathbf 0}_N$, and the identical large threshold case where ${\boldsymbol \zeta}/\sigma \gg {\mathbf 1}_N$. For the first case, it can be shown that ${\rm SNR}_{{\rm loss},B=1}\approx \pi/2$, which demonstrates that for weak signal detection under zero threshold, one achieves the minimal performance loss compared to the quantized system. For the second case, ${\mathcal I}_{{\rm R},\leq}=\emptyset$ and ${\mathcal I}_{{\rm I},\leq}=\emptyset$, and define $\zeta_{\rm min}=\min\{|\Re\{\zeta_n\}|,|\Im\{\zeta_n\}|,n=1,2,\cdots,N\}$, the righthand term of (\ref{SNRloss1bit0}) can be upper bounded as
\begin{align}\label{SNRloss1bit1}
{\rm SNR}_{{\rm loss},B=1}\geq  \frac{\sigma}{\zeta_{\rm min}}\sqrt{\pi}{\rm e}^{\frac{\zeta_{\rm min}^2}{\sigma^2}}.
\end{align}
The SNR loss increases rapidly as ${\zeta_{\rm min}}/{\sigma}$ increases. Another common case is that either $|\Re\{\zeta_n\}|/\sigma\approx 0$ or $|\Re\{\zeta_n\}|/\sigma\gg 1$, either $|\Im\{\zeta_n\}|/\sigma\approx 0$ or $|\Im\{\zeta_n\}|/\sigma\gg 1$. Let $N_{\rm R}$ and $N_{\rm I}$ denote the number of terms such that $|\Re\{\zeta_n\}|/\sigma\approx 0$ and $|\Im\{\zeta_n\}|/\sigma\approx 0$, respectively. In this setting, the SNR loss ${\rm SNR}_{{\rm loss},B=1}$ is dominated by the first and the second terms, which can be further approximated as
\begin{align}\label{SNRloss1bit2}
{\rm SNR}_{{\rm loss},B=1}\approx \frac{\pi}{2}\times \frac{N}{(N_{\rm R}+N_{\rm I})/2}.
\end{align}
The term $\frac{\pi}{2}$ in (\ref{SNRloss1bit2}) is the minimal SNR loss incurred by the $1$ bit quantization. The term $\frac{N}{(N_{\rm R}+N_{\rm I})/2}$ in (\ref{SNRloss1bit2}) is the loss due to the nonzero thresholds with large absolute value compared to the noise standard deviation $\sigma$ which, in our case, is the synthesized signal except the current signal. This demonstrates that if the synthesized signal is comparable to the noise standard deviation, the performance degradation is small.

\subsection{Frequency and Complex Amplitude are Unknown}\label{freqcampunksubSec}
For the frequency unknown case, we also evaluate the performance of unbiased estimator by deriving the CRB. Besides, we propose a detector and analyze its theoretical performance. We discrete the frequency into a number of grids ${\Omega}_{\rm {DFT}}=\{\frac{g}{N}2\pi, g=0,1,2,\cdots,N-1\}$. For each frequency $\omega_g=\frac{g}{N}2\pi\in {\Omega}_{\rm {DFT}}$,  we define the Rao test $T_{{\rm R},B}({\boldsymbol{\varphi}_B},{\boldsymbol\zeta},\omega_g)$ as 
\begin{small}
\begin{align}\label{TRByzetaomegag}
    T_{{\rm R},B}({\boldsymbol{\varphi}_B},&{\boldsymbol\zeta},\omega_g)
    =\left.\frac{\partial \ln p({\mathbf y},{\boldsymbol\zeta};{\boldsymbol \theta})}
    {\partial {\boldsymbol\theta}}\right|_{{\boldsymbol\theta}={\boldsymbol\theta}_0}^{\rm T}
    {\mathbf I}_B^{-1}({\boldsymbol\theta}_0) \left.\frac{\partial \ln p({\mathbf y},{\boldsymbol\zeta};{\boldsymbol \theta})}
    {\partial {\boldsymbol\theta}}\right|_{{\boldsymbol\theta}={\boldsymbol\theta}_0}\notag\\
    =&\frac{\left({\mathbf a}^{\rm H}(\omega_g){\rm diag}({\mathbf h}_{B,+}({\boldsymbol\zeta})){\mathbf a}(\omega_g)\right)|{{\mathbf a}^{\rm H}(\omega_g)}{\boldsymbol \varphi}_B|^2
        -\Re\{\left({\mathbf a}^{\rm T}(\omega_g){\rm diag}({\mathbf h}_{B,-}({\boldsymbol\zeta})){\mathbf a}(\omega_g)\right){({{\mathbf a}^{\rm H}}(\omega_g){\boldsymbol \varphi}_B)^2}\}}{\left({\mathbf a}^{\rm H}(\omega_g){\rm diag}({\mathbf h}_{B,+}({\boldsymbol\zeta})){\mathbf a}(\omega_g)\right)^2-\vert {\mathbf a}^{\rm T}(\omega_g){\rm diag}({\mathbf h}_{B,-}({\boldsymbol\zeta})){\mathbf a}(\omega_g) \vert^2}.
\end{align}
\end{small}We conjecture that the distribution of $T_{{\rm R},B}({\boldsymbol{\varphi}_B},{\boldsymbol\zeta},\omega_g)$ follows
\begin{align}
T_{{\rm R},B}({\boldsymbol{\varphi}_B},{\boldsymbol\zeta},\omega_g)\stackrel{a}{\sim}\begin{cases}
\chi_{2}^2, {\rm under}~{\mathcal H}_0,\\
\chi_{2}^{\prime 2}(\lambda_{{\rm R},B,g}), {\rm under}~{\mathcal H}_1,
\end{cases}
\end{align}
where $\lambda_{{\rm R},B,g}$ reduces to
\begin{align}\label{glambdageneral}
	\lambda_{{\rm R},B,g}=&\frac{2}{\sigma^2}\left|\frac{{\mathbf a}^{\rm H}(\omega_g){\mathbf a}(\omega)}{|{\mathbf a}(\omega_g)|_2^2}\right|^2\left({\mathbf a}^{\rm H}{\rm diag}({\mathbf h}_{B,+}({\boldsymbol\zeta})){\mathbf a}|x|^2+\Re\{{\mathbf a}^{\rm T}(\omega){\rm diag}({\mathbf h}_{-}({\boldsymbol\zeta})){\mathbf a}(\omega){{x}^2}\}\right)\notag\\
	=&\frac{2}{\sigma^2}\beta_g\left({\mathbf a}^{\rm H}{\rm diag}({\mathbf h}_{B,+}({\boldsymbol\zeta})){\mathbf a}|x|^2+\Re\{{\mathbf a}^{\rm T}(\omega){\rm diag}({\mathbf h}_{B,-}({\boldsymbol\zeta})){\mathbf a}(\omega){{x}^2}\}\right),
\end{align}
$\beta_g$ is
\begin{align}\label{betagcal}
	\beta_g=\left|{\sin\left(\frac{N(\omega-\omega_g)}{2}\right)}/\left({N\sin\left({\frac{\omega-\omega_g}{2}}\right)}\right)\right|^2
\end{align}
due to the mismatch between $\omega$ and $\omega_g$ and $\beta_g\leq 1$. Note that the term ${{\mathbf a}^{\rm H}(\omega_g)}{\boldsymbol \varphi}_B$ in $\{T_{{\rm R},B}({\boldsymbol{\varphi}_B},{\boldsymbol\zeta},\omega_g)\}_{g=0}^{N-1}$ can be evaluated efficiently through FFT, while the term ${\mathbf a}^{\rm T}(\omega_g){\rm diag}({\mathbf h}_{B,-}({\boldsymbol\zeta})){\mathbf a}(\omega_g)$ can also be evaluated through IFFT, which simplifies the computation complexity significantly. We propose the following Rao test $T_{{\rm R},B,\max}({\boldsymbol{\varphi}_B},{\boldsymbol\zeta})$ as
\begin{align}\label{Raotestfreunk}
\begin{split}
T_{{\rm R},B,\max}({\boldsymbol{\varphi}_B},{\boldsymbol\zeta})
=&\underset{{\omega_g}\in\Omega_{\rm DFT}}{\max}~T_{{\rm R},B}({\boldsymbol{\varphi}_B},{\boldsymbol\zeta},\omega_g).
\end{split}
\end{align}
Note that $T_{{\rm R},\rm B}({\boldsymbol{\varphi}_B},{\boldsymbol\zeta},\omega_g)$ is only related to ${{\mathbf a}^{\rm H}(\omega_g)}{\boldsymbol \varphi}_B$. In addition, one has
\begin{equation}
{\rm E}\left[{\boldsymbol \varphi}_B\right]={\mathbf 0},\notag\\
{\rm E}\left[[{\boldsymbol \varphi}_B]_n[{\boldsymbol \varphi}_B]_m\right]=0, {\rm E}\left[[{\boldsymbol \varphi}_B]_n[{\boldsymbol \varphi}_B]_m^*\right]=0,n\neq m,
\end{equation}
Provided $\Re\{\boldsymbol\zeta\}=\Im\{\boldsymbol\zeta\}$, one has
\begin{subequations}
	\begin{align}
		&{\rm E}\left[[{\boldsymbol \varphi}_B]_n[{\boldsymbol \varphi}_B]_n\right]=0, \\
		&{\rm E}\left[[{\boldsymbol \varphi}_B]_n[{\boldsymbol \varphi}_B]_n^*\right]=h_B(\Re\{\zeta_n\},\sigma^2)+h_B(\Im\{\zeta_n\},\sigma^2).
	\end{align}
\end{subequations}
In addition, suppose that
\begin{align}
h_B(\Re\{\zeta_n\},\sigma^2)+h_B(\Im\{\zeta_n\},\sigma^2)={\rm const}
\end{align}
independent of $n$ where $h_B(\cdot,\sigma^2)$ is defined in (\ref{hxsigma}), then the covariance matrix of ${\boldsymbol \varphi}_B$ is a scaled identity matrix. Due to ${{\mathbf a}^{\rm H}(\omega_g)}{{\mathbf a}(\omega_{g^{\prime}})}=\delta_{gg^{\prime}}$, it can be shown that ${{\mathbf a}^{\rm H}(\omega_g)}{\boldsymbol \varphi}_B$ is uncorrelated with ${{\mathbf a}^{\rm H}(\omega_{g^{\prime}})}{\boldsymbol \varphi}_B$ for $g\neq g^{\prime}$, thus $T_{{\rm R},B}({\boldsymbol{\varphi}_B},{\boldsymbol\zeta},\omega_g)$ is uncorrelated with $T_{{\rm R},B}({\boldsymbol{\varphi}_B},{\boldsymbol\zeta},\omega_{g^{\prime}})$. Here we make an assumption that $\{T_{{\rm R},B}({\boldsymbol{\varphi}_B},{\boldsymbol\zeta},\omega_g)\}_{g=1}^G$ are independent\footnote{Provided that ${\boldsymbol \varphi}_B$ follows Gaussian distribution, the assumption holds.}, the false alarm probability $\tilde{\rm P}_{{\rm {FA}},B}$ is
\begin{align}\label{Pfadef}
	&\tilde{\rm P}_{{\rm FA},B}={\rm {Pr}}\left\{\underset{g=0, \cdots, N-1}{\rm {max}}~T_{{\rm R},B}({\boldsymbol{\varphi}_B},{\boldsymbol\zeta},\omega_g)>\tau_{\rm th}\right\}\notag\\
	&= 1-{\rm {Pr}}\left( \underset{g=0, \cdots, N-1}{\rm {max}}~T_{{\rm R},B}({\boldsymbol{\varphi}_B},{\boldsymbol\zeta},\omega_g)\leq\tau_{\rm th}\right) \notag\\
	&=1 - \left({\rm {Pr}}\left(T_{{\rm R},B}({\boldsymbol{\varphi}_B},{\boldsymbol\zeta},\omega_g)\leq\tau_{\rm th}\right)\right)^N\notag\\
	&=1-F_{\chi_{2}^2}^N\left(\tau_{\rm th}\right)=1-(1-{\rm e}^{-\frac{\tau_{\rm th}}{2}})^N,
\end{align}
where $F_{\chi_{2}^2}(\cdot)$ is the cumulative distribution function of the chi-squared distribution with $2$ degrees of freedom,
and the threshold is
\begin{align}
    \tau_{\rm th} = -2\ln{(1-(1-\tilde{\rm P}_{{\rm FA},B})^{\frac{1}{N}})}.
\end{align}
To find the detection probability $\tilde{\rm P}_{{\rm D},B}$, we first define a detection as a threshold crossing in the \emph{correct} frequency bin $g^*$ corresponding to the frequency $\omega_{g^*}$ closest to the true frequency $\omega$ in wrap-around distance. Hence $\tilde{\rm P}_{\rm D}$ is defined as the probability that the peak of the spectrum occurs in the \emph{correct} frequency bin $g^*$ and crosses the threshold $\tau_{\rm th}$. With this definition and for a given $\tilde{\rm P}_{{\rm FA},B}$, we have
\begin{small}
	\begin{align}\label{unknownfreq}
		\tilde{\rm P}_{{\rm D},B}&=Q_1\left(\sqrt{\lambda_{{\rm R},B,g^*}},\sqrt{\tau_{\rm th}}\right)\notag\\
		&=Q_1\left(\sqrt{\lambda_{{\rm R},B,g^*}},\sqrt{-2\ln{(1-(1-\tilde{\rm P}_{{\rm FA},B})^{\frac{1}{N}})}}\right).
	\end{align}
\end{small}
where $Q_1(\cdot,\cdot)$ is the Marcum Q-function. According to the definition of $\lambda_{{\rm R},B,g^*}$, one could also define the SNR loss similar to the frequency known case. It can be easily seen that the SNR loss is the same as (\ref{SNRloss}).

It is worth noting that once the signal is detected, the gradient descent or Newton method is adopted to jointly refine the frequency and amplitude estimates. In order to accelerate the GNOMP approach developed later in Sec. \ref{GNOMP}, one could use oversampling to evaluate $T_{{\rm R},B}({\boldsymbol{\varphi}_B},{\boldsymbol\zeta},\omega_g)$, $\omega_g\in \Omega_{\rm os}$, where $\gamma_{\rm os}$ is the oversampling factor and $\Omega_{\rm os}=\{2\pi g/(\gamma_{\rm{os}}N),g=0,1,\cdots,\gamma_{\rm{os}}N-1\}$. The coarse estimate of the frequency can be obtained via finding the maximum of $T_{{\rm R},B}({\boldsymbol{\varphi}_B},{\boldsymbol\zeta},\omega_g)$, $\omega_g\in \Omega_{\rm os}$.

\subsection{Further Discussion on the Detection Probability of Multiple Sinusoids}
Although we focus on a single signal detection, the analysis can be used to provide an upper bound of the detection probabilities of all the targets. Let $A_i$ denote the event that the the $i$th target is detected. Define $\lambda_{g^*,i}$ as
\begin{align}\label{lambdagenerali}
	\lambda_{g^*,i}&=\frac{2}{\sigma^2}\beta_{g^*,i}{\mathbf a}^{\rm H}(\omega_i){\rm diag}({\mathbf h}_{B,+}({\boldsymbol\zeta}_{\setminus i})){\mathbf a}(\omega_i)|x_i|^2+\frac{2}{\sigma^2}\Re\{{\mathbf a}^{\rm T}(\omega_i){\rm diag}({\mathbf h}_{B,-}({\boldsymbol\zeta}_{\setminus i})){\mathbf a}(\omega_i){{x}_i^2}\},
\end{align}
where $\beta_{g^*,i}$ is calculated through (\ref{betagcal}) by replacing $\omega$ with $\omega_i$, ${\boldsymbol\zeta}_{\setminus i}=\sum\limits_{k=1,k\neq i}^K{\mathbf a}(\omega_k)x_k$, $x_i$ denotes the complex amplitude of the $i$th sinusoid. According to (\ref{unknownfreq}), the detection probability of the $i$th sinusoid with all the other signals being perfectly known is
\begin{align}\label{unknownfreqi}
\tilde{\rm P}_{{\rm D},i}=Q_1\left(\sqrt{\lambda_{g^*,i}},\sqrt{-2\ln{(1-(1-\tilde{\rm P}_{{\rm FA},B})^{\frac{1}{N}})}}\right).
\end{align}
Consequently, the detection probability ${\rm P}_{\rm D}^{\rm all}={\rm Pr}(A_1A_2\cdots A_K)$ of all the targets can be upper bounded as
\begin{align}
	{\rm P}_{\rm D}^{\rm all}={\rm Pr}(A_1A_2\cdots A_K)={\rm Pr}(A_1){\rm Pr}(A_2|A_1){\rm Pr}(A_K|A_1A_2\cdots A_{K-1})\leq	\prod\limits_{k=1}^K\tilde{\rm P}_{{\rm D},k}.\label{unknownfreqall}
\end{align}
A particular case is that all the $K-1$ targets except the $K$th target are strong such that ${\rm P}_{{\rm D},k}\approx 1$, $k=1,2,\cdots,K-1$, and the detection probability ${\rm P}_{\rm D}^{\rm all}$ (\ref{unknownfreqall}) can be simplified as
\begin{align}\label{unknownfreqallsimplified}
{\rm P}_{\rm D}^{\rm all}
\leq (\approx) \tilde{\rm P}_{{\rm D},K}.
\end{align}
\begin{figure}[htbp]
	\centering
	\includegraphics[width=16cm]{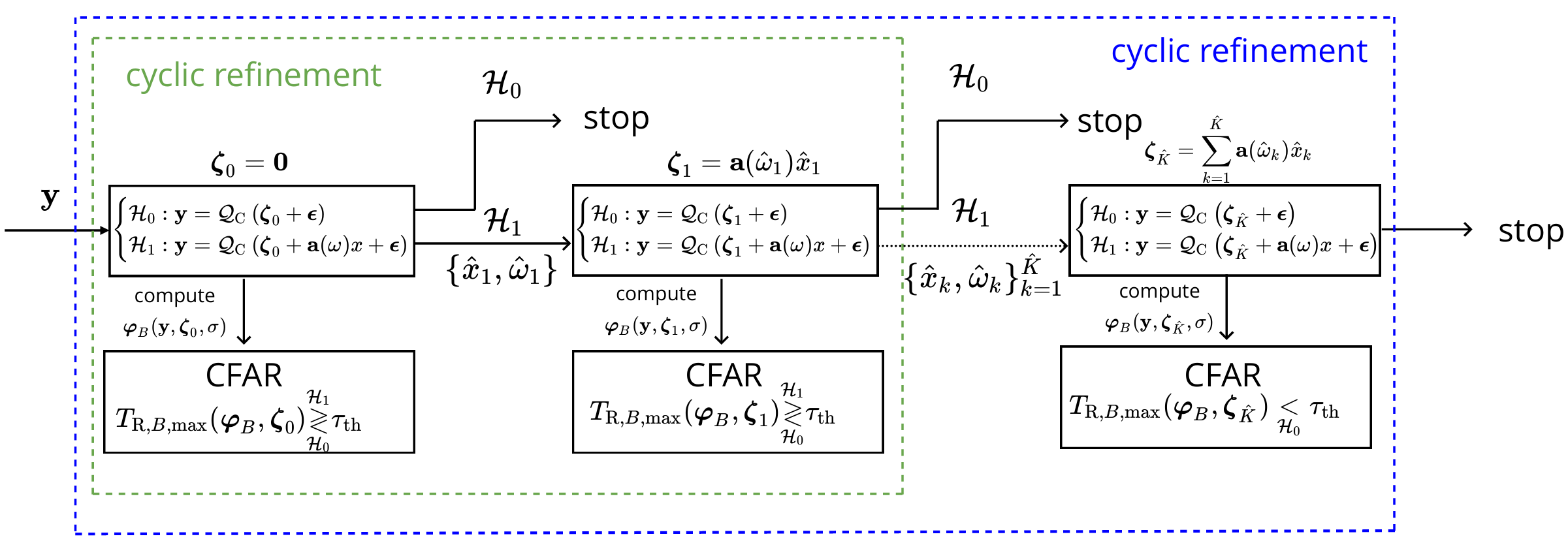}
	\caption{A bird's eye view of the GNOMP and the relationship between the theoretical parts established in Subsection \ref{freqcampunksubSec} and the GNOMP, where the GNOMP stops when the number of detected targets is $\hat{K}$.}
 \label{GNOMPBirdviewfigure}
\end{figure}
This demonstrates that the detection probability of all the targets is dominated by the detection probability of the weakest target, which makes sense.
\section{Generalized NOMP}\label{GNOMP}
This section develops a fast GNOMP for LSE\&D. The GNOMP mainly consists of two steps: Detection and Estimation. The signals that have been estimated can be synthesized to create  the thresholds $\boldsymbol\zeta$, and the BHT problem (\ref{BHT}) studied in Sec. \ref{Signalsignal} can be used to detect whether an additional sinusoid is present or not, and a coarse estimate of the frequency can be obtained to provide a good initial point for the Newton step. Fig. \ref{GNOMPBirdviewfigure} shows the steps of GNOMP and how the theoretical parts established in Subsection \ref{freqcampunksubSec} is used to develop the GNOMP. In the following, we will first demonstrate how to handle a single sinusoid, and then extend these findings to scenarios involving multiple sinusoids.

Given the number of sinusoids $K$ and provided that the real and imaginary parts of $y_n$ lie in the interval $[l(\Re\{y_n\}),u(\Re\{y_n\}))$ and $[l(\Im\{y_n\}),u(\Im\{y_n\}))$, respectively, the MLE of the frequencies and amplitudes are
\begin{align}
\underset{{\boldsymbol\omega},{\mathbf x}}{\operatorname{maximize}}~l({\mathbf y};{\boldsymbol\omega},{\mathbf x}),
\end{align}
where $l({\mathbf y};{\boldsymbol\omega},{\mathbf x})$ is defined as
\begin{small}
    \begin{align}\label{loglikelihoodly}
        &l({\mathbf y};{\boldsymbol \omega},{\mathbf x})\triangleq \sum\limits_{n=1}^N\left(\log\left(\Phi\left(\frac{u(\Re\{y_n\})-\Re\{\sum\limits_{k=1}^K{ a}_n(\omega_k)x_k\}}{\sigma/{\sqrt{2}}}\right)-\Phi\left(\frac{l(\Re\{y_n\})-\Re\{\sum\limits_{k=1}^K{ a}_n(\omega_k)x_k\}}{\sigma/{\sqrt{2}}}\right)\right)\right.\notag\\
        &\left.+\log\left(\Phi\left(\frac{u(\Im\{y_n\})-\Im\{\sum\limits_{k=1}^K{ a}_n(\omega_k)x_k\}}{\sigma/{\sqrt{2}}}\right)-\Phi\left(\frac{l(\Im\{y_n\})-\Im\{\sum\limits_{k=1}^K{ a}_n(\omega_k)x_k\}}{\sigma/{\sqrt{2}}}\right)\right)\right).
    \end{align}
    \end{small}Directly solving the above problem needs a $K$ dimensional search of the frequencies by restricting the frequencies onto the grids. With the frequencies being fixed, the amplitudes can be solved via
\begin{align}\label{ampeachest}
\underset{{\mathbf x}}{\operatorname{maximize}}~l({\mathbf y};{\boldsymbol\omega}_g,{\mathbf x}).
\end{align}
It has been shown that (\ref{ampeachest}) is a convex optimization problem \cite{CSquant}, which can be solved efficiently. In total, the number of convex optimization problems needed to be solved is $N_g^K$, where $N_g$ denotes the number of grids. Then a gradient descent or Newton method can be adopted to eliminate the offgrid effects. The computation complexity of this method is huge especially when the number of frequencies $K$ is large. In addition, the number of frequencies $K$ is usually unknown. Therefore we propose a greedy based low complexity approach named GNOMP, and use the CFAR based criterion to perform target detection and stop the GNOMP.

Motivated by the low complexity and high estimation accuracy of NOMP algorithm, we propose the GNOMP algorithm which iteratively cancels the interference in a nonlinear and greedy way.

{\bf Main Idea}: The proposed GNOMP is a sequential approach. We first treat the multiple targets estimation and detection as a single (or the strongest) target  estimation and detection problem by ignoring the other targets, which corresponds to a BHT problem (\ref{BHT}) with ${\boldsymbol\zeta}={\mathbf 0}$. If the detector decides that the target is absent, the algorithm stops. Otherwise, we estimate the frequency and amplitude of the target as $\hat{\theta}_1$ and $\hat{x}_1$. Then, we taking the estimated signal into account, and we address the BHT problem (\ref{BHT}) with ${\boldsymbol\zeta}={\mathbf a}(\hat{\theta}_1)\hat{x}_1$. The algorithm proceeds similarly in the single target scenario and stops until no target is detected. In the following, we provide the details.
\subsection{A Single Sinusoid Scenario}\label{singletarget}
At first, we treat the multiple targets estimation and detection as a single (or the strongest) target  estimation and detection problem by ignoring the other targets. In this way, the BHT problem (\ref{BHT}) is
\begin{align}\label{BHT0}
	\begin{cases}
		{\mathcal H}_0:{\mathbf y}=\mathcal{Q}_{\rm C}\left({\boldsymbol\zeta}_0+{\boldsymbol\epsilon}\right),\\
		{\mathcal H}_1:{\mathbf y}=\mathcal{Q}_{\rm C}\left({\boldsymbol\zeta}_0+{\mathbf{a}}({\omega }){{x}}+{\boldsymbol\epsilon}\right),
	\end{cases}
\end{align}
where ${\boldsymbol\zeta}_0={\mathbf 0}$. Note that the above BHT problem is a special case of (\ref{BHT}) corresponding to ${\boldsymbol\zeta}={\mathbf 0}$, which has been studied in depth in Sec. \ref{freqcampunksubSec} for a general ${\boldsymbol\zeta}$. We evaluate the proposed detector established in (\ref{Raotestfreunk}) to the threshold $\tau_{\rm th}$ calculated by (\ref{Pfadef}) with ${\boldsymbol\zeta}={\mathbf 0}$, we decide whether the target is present or not.

If no target is detected, the algorithm stops. Otherwise, we need to estimate the frequency and amplitude of the signal for the following model
\begin{align}\label{LSEgeneralmodel}
	{\mathbf y}={\mathcal Q}\left(\Re\{{\mathbf{a}}({{\omega}}){{x}}+{\boldsymbol\epsilon}\}\right)+{\rm j}{\mathcal Q}\left(\Im\{{\mathbf{a}}({{\omega}}){{x}}+{\boldsymbol\epsilon}\}\right).
\end{align}
The loglikelihood is $l({\mathbf y};\omega,{x})$. The MLE can be formulated as
\begin{align}\label{MLEsingle}
	(\hat{x}_{{\rm ML}},\hat \omega_{\rm {ML}})=\underset{{x},\omega}{\operatorname{argmax}}~l({\mathbf y};{\omega},{x}).
\end{align}
Directly solving the MLE of a single frequency is difficult as $l({\mathbf y};\omega,{x})$ is not concave with respect to the frequency and the real and imaginary parts of $x$. However, with the frequency $\omega$ being known, the loglikelihood function is concave with respect to the real and imaginary parts of $x$ \cite{CSquant}. Therefore we use the alternating minimization (AM) approach to solve (\ref{MLEsingle}). We first obtain a good initial point of $\omega$ as $\hat{\omega}$ and the amplitude ${x}$ as $\hat{x}$. Then we fix the amplitude estimate $\hat{x}$ and refine the estimate $\hat{\omega}$ as $\hat{\omega}^{\prime}$. The amplitude is further optimized as $\hat{x}^{\prime}$ by fixing the frequency as $\hat{\omega}^{\prime}$. To find a good initial point of $\omega$, the frequency $\omega\in [0,2\pi)$ is first discretized into a finite number of grids. As shown in \cite{Madhow16TSP}, an oversampling factor $\gamma_{\rm{os}}=4$ with respect to the Nyquist grid is preferable to ensure the convergence of Newton's method. For each grid $\omega_g=g2\pi/(\gamma_{\rm os}N)$, $g=0,1,2,\cdots,\gamma_{\rm os}N-1$, $\gamma_{\rm os}$ is the oversampling factor, the following subproblem
\begin{align}\label{subprob1}
	\hat{x}_g=\underset{{x}}{\operatorname{argmax}}~l({\mathbf y};{\omega}_g,{x}),
\end{align}
is solved globally. The frequency yielding the maximum loglikelihood is calculated as
\begin{align}
	\hat{\omega}=\underset{\omega_g\in {\Omega}_{\rm{os}}}{\operatorname{argmax}}~l({\mathbf y};{\omega}_g,\hat{x}_g).
\end{align}
The computation complexity of the above steps are very high, and we instead propose a very low complexity approach in the following.

We use oversampling to evaluate $T_{{\rm R},B}(\boldsymbol{\varphi}_B,{\boldsymbol\zeta},\omega_g)$, $\omega_g\in \Omega_{\rm os}$, where $\gamma_{\rm os}$ is the oversampling factor and $\Omega_{\rm os}=\{2\pi g/(\gamma_{\rm{os}}N),g=0,1,\cdots,\gamma_{\rm{os}}N-1\}$. The coarse estimate of the frequency can be obtained via finding the maximum of $T_{{\rm R},B}({\mathbf y},{\boldsymbol\zeta},\omega_g)$, $\omega_g\in \Omega_{\rm os}$. In particular, the steps can be summarized as follows:
\begin{itemize}
  \item  The Rao detector is adopted to obtain $\hat{\omega}$ efficiently through solving
\begin{align}
\hat{\omega}=\underset{\omega_g\in {\Omega}_{\rm{os}}}{\operatorname{argmax}}~T_{{\rm R},B}({\mathbf y},{\boldsymbol\zeta},\omega_g),
\end{align}
where ${\boldsymbol\zeta}={\mathbf 0}$ due to detecting the first signal. It is worth noting that $\hat{\omega}$ can be obtained via efficient FFT and IFFT. Additionally, the amplitude estimate $\hat{x}$ is obtained via solving
\begin{align}\label{subprobsimplified}
	\hat{x}=\underset{{x}}{\operatorname{argmax}}~l({\mathbf y};\hat{\omega},{x}).
\end{align}
  \item AM: With the coarse detection frequency $\hat{\omega}$ and amplitude $\hat{x}$, the Newton refinement is adopted to refine the frequency $\hat{\omega}$ with the amplitudes $\hat{x}$ being fixed, i.e.,
\begin{align}
\hat{\omega}^{\prime}=\hat{\omega}-\frac{\dot{l}({\mathbf y};\hat{\omega},\hat{x})}{\ddot{l}({\mathbf y};\hat{\omega},\hat{x})},
\end{align}
where $\dot{l}({\mathbf y};\hat{\omega},\hat{x})$ and $\ddot{l}({\mathbf y};\hat{\omega},\hat{x})$ denote the first and second order derivative of $l({\mathbf y};{\omega},\hat{x})$ with respect to ${\omega}$ evaluated at $\hat{\omega}$. For brevity, the detailed computations of $\dot{l}({\mathbf y};{\omega},\hat{x})$ and $\ddot{l}({\mathbf y};{\omega},\hat{x})$ are omitted. Then the amplitude $\hat{x}$ is refined via the Newton step with  ${\omega}$ fixed at $\hat{ \omega}^{\prime}$, i.e.,
\begin{align}\label{ampest}
\left[
\begin{array}{c}
  \Re\{\hat{x}^{\prime}\} \\
  \Im\{\hat{x}^{\prime}\}
\end{array}
\right]=\left[
\begin{array}{c}
  \Re\{\hat{x}\} \\
  \Im\{\hat{x}\}
\end{array}
\right]-[\nabla^2l({\mathbf y};{\hat \omega}^{\prime},\hat{x})]^{-1}\nabla l({\mathbf y};{\hat \omega}^{\prime},\hat{x}),
\end{align}
where $\nabla l({\mathbf y};{\omega},{x})$ and $\nabla^2l({\mathbf y};{\omega},{x})$ denote the gradient and Hessian of $l({\mathbf y};{\omega},{x})$ with respect to $[\Re\{x\};\Im\{x\}]^{\rm T}$.
\end{itemize}


\subsection{Multiple Sinusoids Scenario}
Suppose we have detected $L-1$ mixtures of sinusoids, and all the amplitudes-frequencies pairs are $\{(\hat{x}_k^{\prime},\hat{\omega}^{\prime}_k),k=1,\cdots,L-1\}$. Then we face the following BHT problem
\begin{align}\label{BHTL1}
	\begin{cases}
		{\mathcal H}_0:{\mathbf y}=\mathcal{Q}_{\rm C}\left({\boldsymbol\zeta}_{L-1}+{\boldsymbol\epsilon}\right),\\
		{\mathcal H}_1:{\mathbf y}=\mathcal{Q}_{\rm C}\left({\boldsymbol\zeta}_{L-1}+{\mathbf{a}}({\omega }){{x}}+{\boldsymbol\epsilon}\right),
	\end{cases}
\end{align}
where ${\boldsymbol\zeta}_{L-1}=\sum\limits_{k=1}^{L-1}\hat{x}_k^{\prime}{\mathbf a}(\hat{\omega}^{\prime}_k)$, which has been studied in Sec. \ref{freqcampunksubSec}. We still use the Rao detector to decide whether the $L$th target is present or not. If the $L$th target is absent, the GNOMP stops. Otherwise, we follow the steps described in Subsection \ref{singletarget} to obtain the refined estimates of the frequency and the amplitude of the $L$th target.

Now we have detected $L$ mixtures of sinusoids. Let ${\mathcal P} = \{({x}_k,\omega_k),k=1,\cdots,L\}$ denote the set in which the sinusoids have been detected. Then the block coordinate descent (BCD) is applied to refine all the amplitudes frequencies pairs, which amounts to solve
\begin{align}
\underset{\{\omega_k\}_{k=1}^L,{\mathbf x}}{\operatorname{maximize}}~l({\mathbf y};\{\omega_k\}_{k=1}^L,{\mathbf x}),
\end{align}
where $l({\mathbf y};\{\omega_k\}_{k=1}^L,{\mathbf x})$ denotes the loglikelihood whose form is similar to (\ref{loglikelihoodly}). The BCD proceeds as follows: For the $l$th sinusoid, the amplitudes frequencies pairs of the other sinusoids are fixed, and the amplitude frequency pair are optimized via AM similar to Subsection \ref{singletarget} does. After optimizing the $l$th sinusoid, we begin to optimize the other sinusoid in a cyclic way.

Once all the amplitudes-frequencies pairs have been updated and put into the list ${\mathcal P}^{\prime} = \{(\hat{x}_k^{\prime},\hat{\omega}^{\prime}_k),k=1,\cdots,L\}$, we reestimate all the amplitudes of the frequencies by fixing the frequencies to further improve the estimation accuracy, i.e.,
\begin{align}\label{refineallamps}
\underset{{\mathbf x}}{\operatorname{maximize}}~l({\mathbf y};\{\omega_k^{\prime}\}_{k=1}^L,{\mathbf x}).
\end{align}
The gradient and Hessian of $l({\mathbf y};\{\omega_k^{\prime}\}_{k=1}^L,{\mathbf x})$ (\ref{refineallamps}) with respect to $[\Re\{{\mathbf x}\};\Im\{{\mathbf x}\}]$ are calculated and the Newton method can be applied to refine the amplitudes' estimates. Note that all the above steps are accepted provided that these steps improve the loglikelihood.

We have found an interesting phenomenon in the numerical experiments. In the two-sinudoid coexistence scenario where one sinusoid's integrated SNR (defined later in Sec. \ref{NS}) is very large such as $60$ dB, we first detect this strong signal. Then we will detect a spurious component whose amplitude is small. Next we detect the  second sinusoid and its amplitude is stronger than that of the spurious component. We redo the CFAR detection for the spurious component and its test does not exceed the threshold. Therefore, we add a Spurious Component Suppression Step in GNOMP summarized in Algorithm \ref{Alg1}. It is worth noting that the cardinality of ${\mathcal P}_m$ is not always equal to $m$, i.e., $|{\mathcal P}_m|\neq m$ in some scenarios due to this step. Further details about GNOMP could be referred to the NOMP algorithm \cite{Madhow16TSP}.

Note that we have used the CFAR criterion to stop the GNOMP algorithm. We could use the one-bit Bayesian information criterion (1bBIC) proposed in \cite{LiJian18SPL} to select the model order $\hat{K}$ that minimizes the 1bBIC cost function
\begin{align}
{\rm 1bBIC}(\hat{K})=-2\ln p({\mathbf y},\hat{\boldsymbol\omega},\hat{\mathbf x})+5\hat{K}\ln N.
\end{align}
Still, we emphasize that using the Rao test (which can be implemented via FFT and IFFT) instead of 1bBIC reduces the computation complexity significantly.
\begin{algorithm}[ht!]
\caption{GNOMP.}\label{Alg1}
1: \textbf{Procedure} EXTRACTSPECTRUM $({{\mathbf y}, \tau_{\rm th}}):$\\
2: $m\leftarrow 0$, ${\mathcal P}_0 = \{\}$, ${\boldsymbol\zeta}_0=0$\\
3: $\textbf{while}$ {$\max\limits_{\omega\in \Omega_{\rm {DFT}}}T_{\rm R}({\boldsymbol{\zeta}}_B,{\boldsymbol \zeta_m},\omega)>\tau_{\rm th}$} $\textbf{do}$\\
4: \text{\quad}$m\leftarrow m+1$\\
5: \text{\quad}{\scshape Identify}\\
\text{\quad}\text{\quad}$\hat{\omega} = \arg\max\limits_{\omega \in {\Omega}_{\rm{os}}} T_{\rm R,B}({\boldsymbol{\zeta}}_B,{\boldsymbol \zeta_{m-1}},\omega)$\\
\text{\quad}\text{\quad}and its corresponding ${x}$ vector estimate \\
\text{\quad}\text{\quad}$\hat{x} = \arg\max\limits_{x} l({\mathbf y},{\boldsymbol \zeta}_{m-1};\hat{\omega} ,{x}) $\\
7: \text{\quad}{\scshape Single Refinement}: Refine $(\hat {x}, \hat \omega)$ using single frequency Newton update algorithm (${R}_s$ Newton \\
\text{\quad}\text{\quad~}steps) to obtain improved estimates $(\hat {x}', \hat {\omega}')$ and set ${\mathcal P}'_m={\mathcal P}_{m-1}\cup\{(\hat {x}', \hat {\omega}')\}$.\\
8: \text{\quad}${\boldsymbol \zeta}_{m}'={\boldsymbol \zeta}_{m-1}+ {\mathbf a}({\hat \omega}'){\hat {x}}'$\\
9: \text{\quad}{\scshape Cyclic Refinement}: Refine parameters in ${\mathcal P}'_m$ one at a time: For each $(x_l, \omega_l)$ inside ${\mathcal P}'_m$, we\\
\text{\quad}\text{\quad~}treat ${\boldsymbol \zeta}_l={\boldsymbol \zeta}_{m}'-{\mathbf a}({\omega_l}){{x}_l}$ as the pseudo thresholds ${\boldsymbol \zeta}$ and apply single frequency Newton
update\\
\text{\quad}\text{\quad~}algorithm. We perform ${R}_c$ rounds of cyclic refinements. Let ${\mathcal P}''_m$ denote the new set of parameters.\\
10: \text{~~}{\scshape Update} all ${\mathbf x}$ vector estimate in ${\mathcal P}''_m$ via solving (\ref{refineallamps}), i.e., ${\mathbf x}'' = \underset{{\mathbf x}}{\operatorname{argmax}}~l({\mathbf y};\{\omega_k\}_{k=1}^m,{\mathbf x})$. \\
11: \text{~~}{\scshape Spurious Component Suppression}\\
\text{\quad}\text{\quad\quad}\textbf{if}\text{~} $m>2$ and $|x_{m-1}|<0.5|x_{m}|$ \text{\quad} Let ${\mathcal P}'''_m$ denote the new set of parameters \\
 \text{\qquad\quad\quad     }\text{~}and  set ${\boldsymbol \zeta}_m = \sum\limits_{k\in {\mathcal P}'''_m}{\mathbf a }(\omega_k){x}_k$. Let ${\boldsymbol \zeta}_r= {\boldsymbol \zeta}_m-{\mathbf a}({\omega_{m-1}})x_{m-1}$\\
        \text{\qquad\quad\quad     }\text{~}\textbf{if} $T_{\rm R}({\mathbf y},{\boldsymbol \zeta_r},\omega_{m-1})>\tau$\text{\quad} Keep the $(m-1)$th sinusoid into the List ${\mathcal P}'''_m$.\\
        \text{\qquad\quad\quad     }\text{~}\textbf{else}\text{\quad} Remove the $(m-1)$th sinusoid from the List ${\mathcal P}'''_m$. Run {\scshape Cyclic Refinement} and  \\
  \text{\qquad}\text{\qquad~~\quad}\text{~}      \text{\quad}{\scshape Update} steps to optimize parameters in ${\mathcal P}'''_m$. \\
        \text{\qquad\quad\quad     }\text{~}\textbf{end}\\
\text{\quad}\text{\quad\quad}\textbf{end}\\
12: \text{~~}Let ${\mathcal P}_m$ denote the new set of parameters and set ${\boldsymbol \zeta}_m=\sum\limits_{(x_k,\omega_k)\in{\mathcal {P}_m}}{\mathbf a }(\omega_k){x}_k$.\\
13: \textbf{return} ${\boldsymbol \zeta}_{m}$ and ${\mathcal P}_m$.
\end{algorithm}

The computation complexity of GNOMP, assuming it runs for exactly $K$ iterations, can be analyzed as follows:
\begin{itemize}
	\item {\bf Detection Step}: The computation cost of the Detection step can be efficiently implemented via FFT and IFFT with a cost of $KN\log N$.
	\item {\bf Identify Step}: Involves obtaining the initial point of $\omega$, which can also be efficiently implemented via FFT and IFFT with a cost of $K\gamma_{\rm os}N\log(\gamma_{\rm os}N)$, where $\gamma_{\rm os}$ is the oversampling factor.
	\item {\bf $x\in {\mathbb C}$ Vector Estimate}:  Consists of calculating the gradient, Hessian matrix, and the Newton step with costs of $O(\text{Iter}\times KN)$, $O(\text{Iter}\times KN)$, and $O(\text{Iter}\times K)$, respectively, where $\text{Iter}$ denotes the number of Newton steps (usually about $2\sim 5$).
	\item {\bf Single Refinement Step}:  Takes only $O(\text{Iter}\times R_sN)$ operations per sinusoid, where $R_s$ is the number of single refinement steps (default is 1), and the total cost is $O(\text{Iter}\times R_sKN)$.
	\item {\bf Cyclic Refinement Step}: Involves refining all frequencies with an overall complexity of $O(\text{Iter}\times R_cR_sK^2N)$, where $R_c$ is the number of cyclic refinement steps (default is 3).
	\item {\bf Update Step}: Involves computing the gradient, Hessian matrix, and the Newton step of the amplitudes of the $K$ sinusoids. The costs are $O(\text{Iter}\times KN)$, $O(\text{Iter}\times K^2N)$, and $O(\text{Iter}\times (K^3+K^2))$ per outer iteration. The overall cost of the Update step is $O(\text{Iter}\times (K^2N+K^3N+K^4+K^3))$.
\end{itemize}

It's noteworthy that the computation complexities of 1bMMRELAX and MVALSE-EP are $O(KN^2)$ and $O(N^2+\text{Iter}\times NK^3)$ with both scaling with $N^2$ \cite{LiJian19TSP}, while GNOMP scales with $N$, resulting in a lower computation complexity.

We provide an example illustrating the results of GNOMP during iterations. The parameters are set as follows: $K=2$, $\omega_1=2.2$, $\omega_2=2.4$, $x_1=-1.505-0.497\rm{j}$, $x_2=-0.164-0.609\rm{j}$, the time domain SNRs are $4$ dB and $-4$ dB, respectively, $N=128$, $P_{\rm FA}=0.01$. The threshold can be calculated as $\tau_{\rm th}=25.5$ dB. Results are shown in Fig. \ref{Demo_GNOMP1} and Fig. \ref{Demo_GNOMP2} for $B=1$ and $B=2$, respectively. The upper graph in Fig. \ref{Demo_1iter_B_1} shows that the Rao test $41.4$ dB exceeds the threshold $25.5$ dB, and the alternative hypothesis is supported. The lower graph in Fig. \ref{Demo_1iter_B_1} provides an initial frequency estimate $2.197$, which is close to the true frequency $\omega_1=2.2$ corresponding to the strongest target. The Newton method and the AM method is used to refine the frequency and amplitude estimates. Taking the first estimated signal into account, we perform target detection in the 2nd iteration and results are shown in Fig. \ref{Demo_2iter_B_1}. It can be seen that the Rao test is $28.8$ dB, exceeding the threshold about $3.3$ dB, thus the second target is detected and the initial frequency estimate is $2.405$, close to the true frequency of the second target. After taking the two estimated targets into consideration, the Rao test is $24$ dB, smaller than the threshold, no target is detected and the GNOMP stops. In addition, we could also evaluate the theoretical $\rm{SNR}_{\rm {Loss}}$ of the two targets as $2.58$ dB and $5.21$ dB by treating the other signal as a perfectly known threshold, and we obtain the effective integrate SNRs of the two targets as $4+10\log 128-2.58=22.49$ dB and $-4+10\log 128-5.21=11.86$ dB, respectively. After estimating the frequency and amplitude of the two sinusoids as $\hat\omega_1=2.198$, $\hat\omega_2=2.401$ and $\hat x_1=-1.506-0.680\rm{j}$, $\hat x_2=-0.333-0.591\rm{j}$, we obtain the estimated SNR loss of the two targets as $\rm{SNR}^\prime_{\rm {Loss,1}}=2.67$ dB and $\rm{SNR}^\prime_{\rm {Loss,2}}=5.42$ dB, which are close to the theoretical SNR loss.

For $B=2$, results are shown in Fig. \ref{Demo_GNOMP2}. It can be seen that the phenomena are similar and the results are omitted here.

\begin{figure*}[htbp]
	\centering
	\subfigure[The 1st iteration]{
		\label{Demo_1iter_B_1}
		\includegraphics[width=5cm]{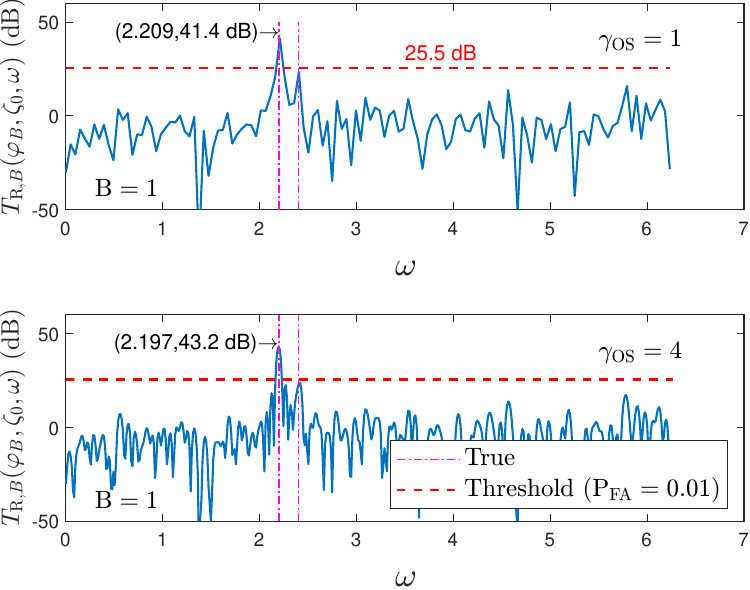}}
		\subfigure[The 2nd iteration]{
		\label{Demo_2iter_B_1}
		\includegraphics[width=5cm]{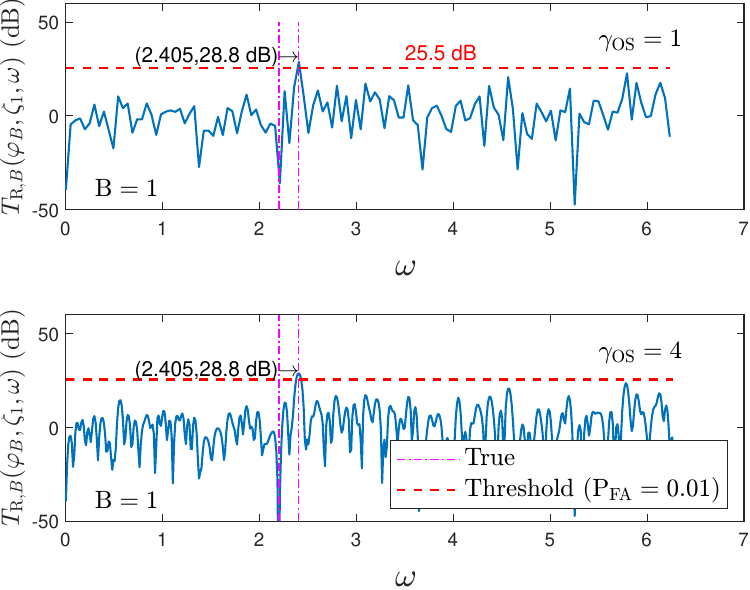}}
		\subfigure[The 3rd iteration]{
		\label{Demo_3iter_B_1}
		\includegraphics[width=5cm]{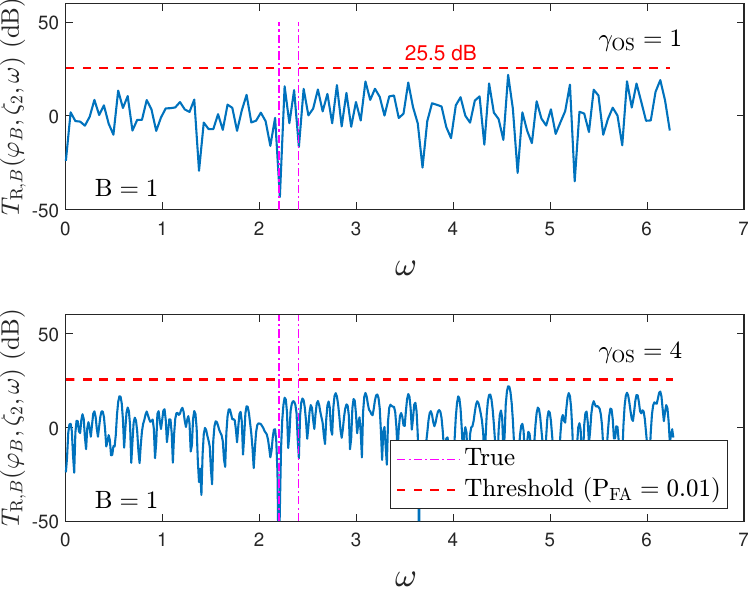}}
	\caption{The detection results of GNOMP under $B=1$.}
	\label{Demo_GNOMP1}
\end{figure*}
\begin{figure*}[htbp]
	\centering
	\subfigure[The 1st iteration]{
		\label{Demo_1iter_B_2}
		\includegraphics[width=5cm]{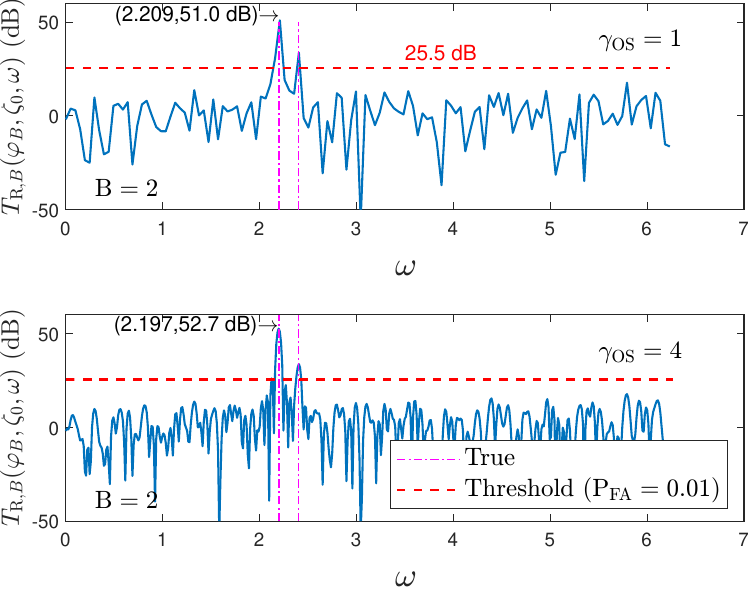}}\subfigure[The 2nd iteration]{
		\label{Demo_2iter_B_2}
		\includegraphics[width=5cm]{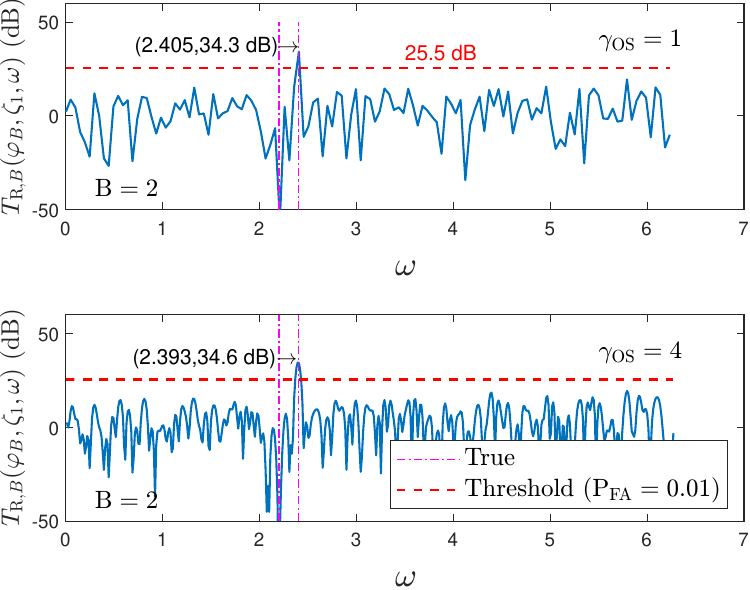}}\subfigure[The 3rd iteration]{
		\label{Demo_3iter_B_2}
		\includegraphics[width=5cm]{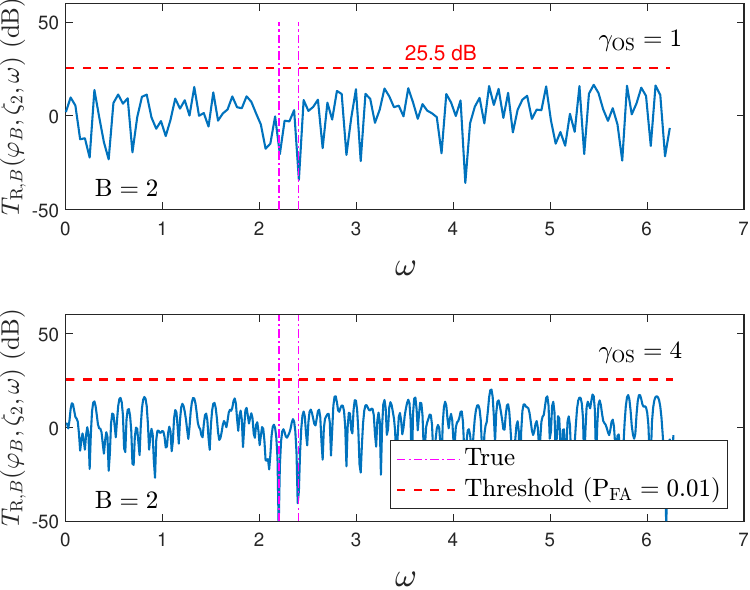}}
	\caption{The detection results of GNOMP under $B=2$.}
	\label{Demo_GNOMP2}
\end{figure*}
\section{Numerical Simulation}\label{NS}
The quantizer is designed with $\gamma$ set as $\gamma = \max(|x_1|,\cdots,|x_k|,\cdots,|x_K|,3\sigma/{\sqrt{2}})$, where $|x_k|$ denotes the magnitude of the $k$th sinusoidal signal\footnote{The design of the quantizer's maximum full-scale range is important for optimization of the estimation and detection performance. However, fine-tuning the quantizer for improved performance goes beyond the scope of this paper. Interested readers can refer to \cite{JianLiquant, guanyudet, JiangzhuTSP15} for further discussion on this topic.}. The choice of $\gamma$ for the quantizer's maximum full-scale range suggests that when all signal amplitudes are weaker than the noise standard deviation, the noise variance is used to design $\gamma$. In contrast, if any signal amplitude is stronger than the noise standard deviation, we design $\gamma$ based on the magnitude of the strongest signal. Additionally, the time domain SNR $10\log(|x|^2/\sigma^2)$ and the integrated SNR $10\log(N|x|^2/\sigma^2)$  are employed. The integrated SNR is $10\log N$ dB higher than the time domain SNR. Typically, it is considered that a signal can be reliably detected if its integrated SNR is greater than about $12$ dB based on unquantized measurements in AWGN environments. For unquantized measurements, the NOMP is run, but the notation GNOMP ($B=\infty$) is used instead. The false alarm probability or rate  is set as $\tilde{\rm P}_{{\rm FA},B}=0.01$ unless stated otherwise.

\subsection{Validate the Estimation Performance In a Single Sinusoid Scenario with Nonidentical Thresholds}
For the initial numerical simulation, we aim to validate the theoretical results of single signal estimation with nonidentical thresholds, as established in Section \ref{Signalsignal}. The nonidentical thresholds $\boldsymbol\zeta$ are set as $\boldsymbol\zeta={\mathbf a}(\omega_1)x_1$. We set $x_1=-0.96-1.75{\rm j}$, and the time domain SNR is approximately $6$ dB, with $\omega_1 = 0.15$. We consider two scenarios:
\begin{itemize}
	\item Scenario 1: Weak signal scenario, $x=-0.27+0.29{\rm j}$, and the time domain SNR is $-8$ dB.
	\item Scenario 2: Strong signal scenario, $x=-0.68+0.73{\rm j}$, and the time domain SNR is about $0$ dB.
\end{itemize}
The parameters are set as follows: $\omega=2.34$, $\sigma^2=1$. We assess the MSE and the CRB of the amplitude of the signal with known thresholds $\boldsymbol\zeta$ for the case where the frequency is unknown. The results are depicted in Fig. \ref{Sin_MSE_fig}. It can be observed that as the number of measurements $N$ increases, the algorithms asymptotically approach the CRBs.
\begin{figure}[htbp]
	\centering
	\subfigure[Scenario 1: Time domain ${\rm SNR}=-8$ dB]{
		\label{MSE_scenario_1}
		\includegraphics[width=2.6in]{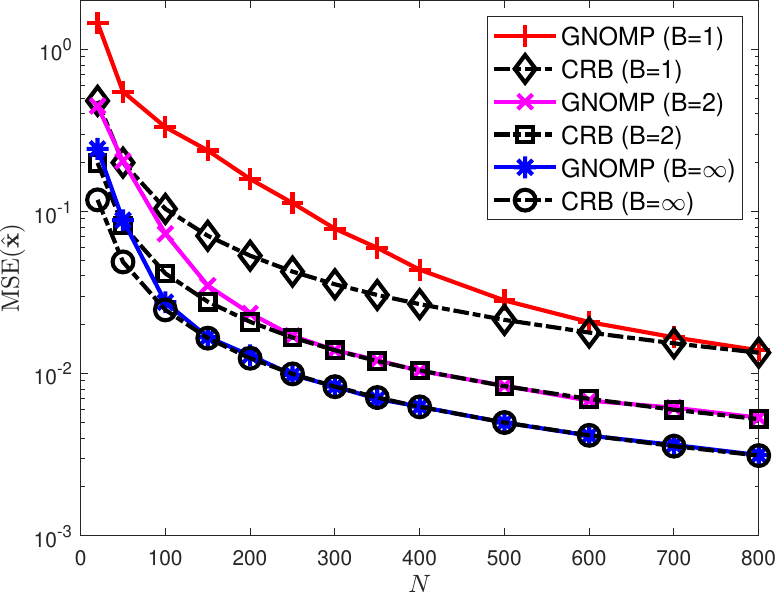}}\subfigure[Scenario 2: Time domain ${\rm SNR}=0$ dB]{
		\label{MSE_scenario_2}
		\includegraphics[width=2.6in]{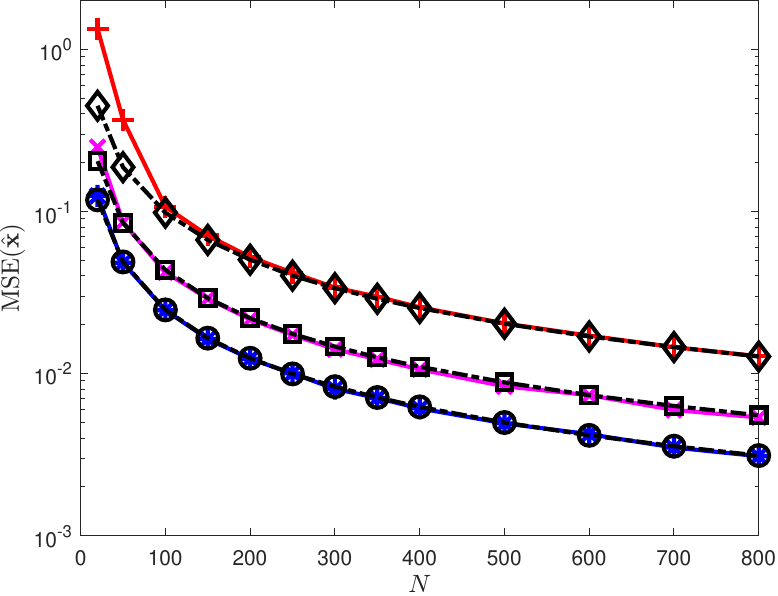}}
	\caption{Measured MSEs and CRBs of amplitude versus the number of measurements $N$ averaged over $5000$ Monte Carlo (MC) trials.}
	\label{Sin_MSE_fig}
\end{figure}

\subsection{Validate the Detection Performance In a Single Sinusoid Scenario with Nonidentical Thresholds}
We set $N=1024$. The threshold ${\boldsymbol \zeta}$ is defined as ${\boldsymbol \zeta}={\mathbf a}(\omega_1)x_1$, where $x_1=2$ and the time domain SNR is $6$ dB. We use two slightly different frequencies to generate $\boldsymbol \zeta$.
\begin{itemize}
	\item Scenario 1: $\omega_1 = \frac{\pi}{2}$, and the SNR losses under 1 bit, 2 bits, and 3 bits quantization are $4.8$ dB, $2.3$ dB, $1$ dB according to (\ref{SNRloss}).
	\item Scenario 2: $\omega_1 = \frac{\pi}{2}+0.1$, and the SNR losses under 1 bit, 2 bits, and 3 bits quantization are $6.4$ dB, $2.2$ dB, $0.9$ dB, according to (\ref{SNRloss}).
\end{itemize}
The detection probability versus the integrated SNR of the target signal is illustrated in Fig. \ref{Sin_P_D_figure}. The measured detection probability closely aligns with the theoretical detection probability, affirming the accuracy of the analysis. As shown in Fig. \ref{Sin_P_D_Scenario_1}, for a detection probability ${\rm P}_{\rm D}=0.5$, the integrated SNRs of 1 bit, 2 bits, 3 bits, and $\infty$ bits quantization are $17.2$ dB, $14.7$ dB, $13.3$ dB, $12.4$ dB, respectively. Therefore, the SNR losses of 1 bit quantization, 2 bits quantization, and 3 bits quantization compared to unquantized measurements are about $17.2-12.4=4.8$ dB, $14.7-12.4=2.3$ dB, and $13.3-12.4=0.9$ dB in Scenario 1. As shown in Fig. \ref{Sin_P_D_Scenario_2}, similarly, the SNR losses under 1 bit, 2 bits, and 3 bits quantization compared to unquantized measurements are about $18.7-12.4=6.3$ dB, $14.6-12.4=2.2$ dB, and $13.2-12.4=0.8$ dB in Scenario 2. These simulation results are consistent with the theoretical analysis.
\begin{figure}[htbp]
	\centering
	\subfigure[$\omega_1=\frac{\pi}{2}$]{
		\label{Sin_P_D_Scenario_1}
		\includegraphics[width=2.6in]{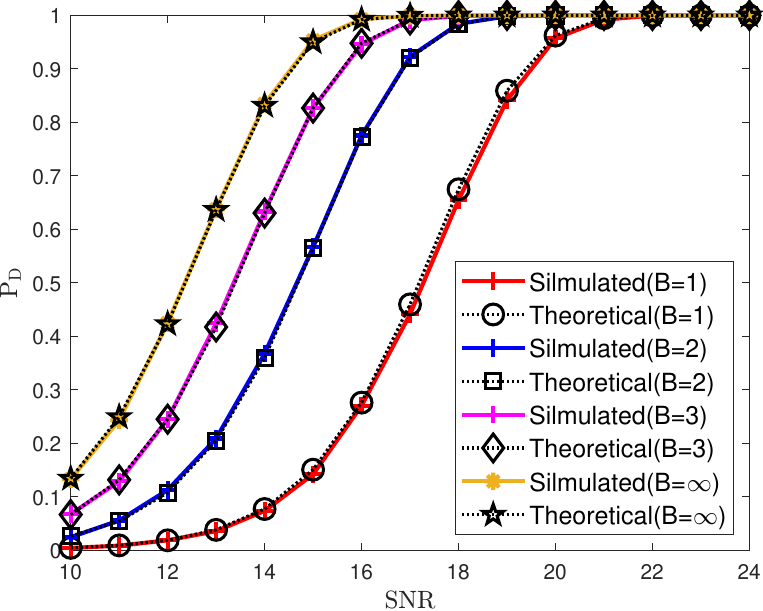}}
	\subfigure[$\omega_1=\frac{\pi}{2}+0.1$]{
		\label{Sin_P_D_Scenario_2}
		\includegraphics[width=2.6in]{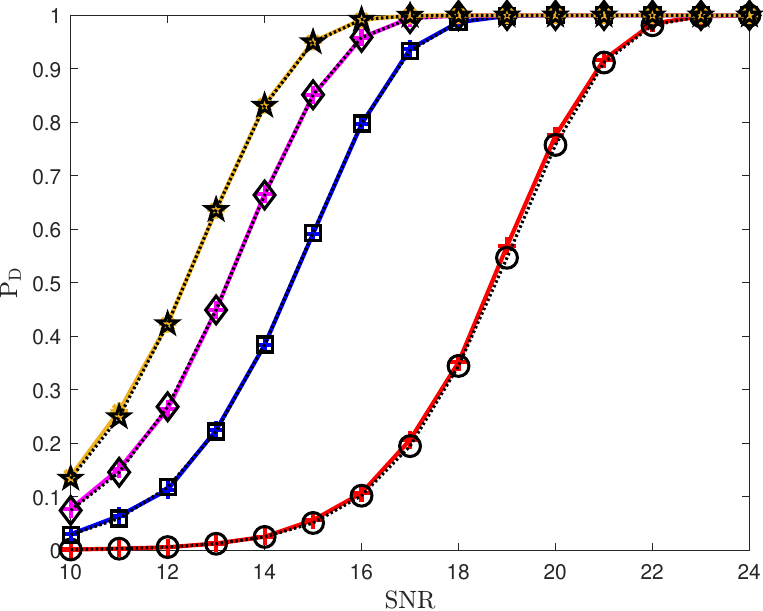}}
	\caption{Measured and computed probability of detection versus the integrated SNR averaged over $10^4$ MC trials.}
	\label{Sin_P_D_figure}
\end{figure}
\subsection{Validate the CFAR Property}
We evaluate the performance of GNOMP using two criteria: the overestimation probability ${\rm P}_{\rm OE}$ and the false alarm probability ${\rm P}_{\rm FA}$. An overestimation occurs when GNOMP overestimates the model order $K$, and a false alarm occurs when the minimum wrap-around distance between a given estimated frequency and all the true frequencies exceeds $\pi/N$. All $K$ targets have identical integrated SNRs ${\rm SNR}$ and their frequencies satisfy the minimum frequency separation $\Delta\omega_{\rm min}=2.5\Delta_{\rm DFT}$.
%
The ``measured'' overestimation and the ``measured'' false alarm probability versus the "nominal" false alarm rates under different bit-depths and integrated SNRs are shown in Fig. \ref{MUL_P_FA_OE_20_dB_figure}. Each point in the plot is generated by 300 runs of the GNOMP algorithm for estimating frequencies in a mixture of $K=8$ sinusoids of $\rm{SNR}=20$ dB. As shown in Fig. \ref{MUL_P_FA_OE_20_dB_figure}, both the empirical false alarm probability and overestimation probability closely follow the nominal value under $\rm{SNR}=20$ dB, demonstrating the high estimation accuracy of GNOMP.

\begin{figure}[htbp]
	\centering
	\subfigure[False alarm probability (${\rm SNR}=20$ dB)]{
		\label{MUL_P_FA_20dB_figure}
		\includegraphics[width=2.6in]{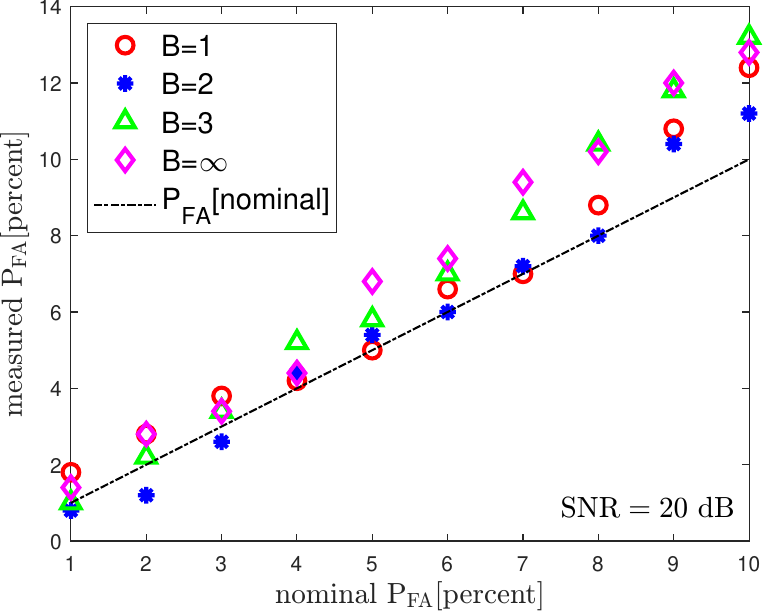}}
	\subfigure[Overestimating probability (${\rm SNR}=20$ dB)]{
		\label{MUL_P_OE_20dB_figure}
		\includegraphics[width=2.6in]{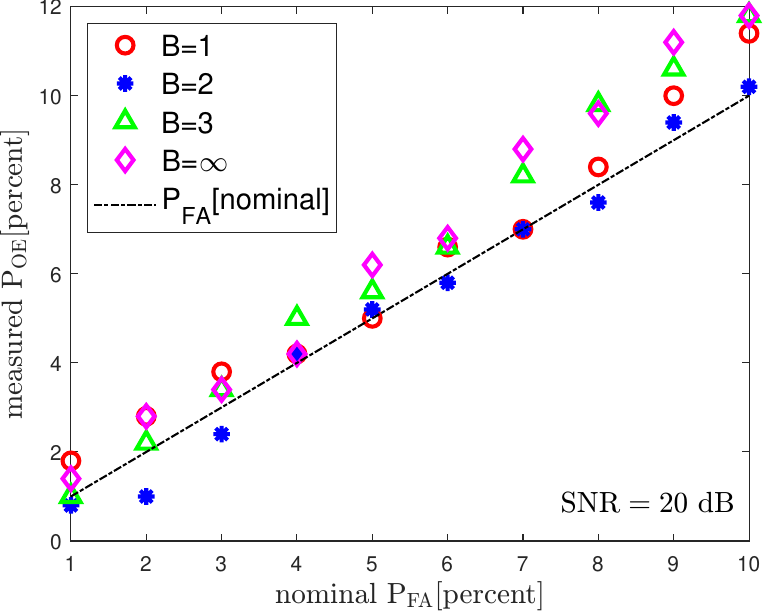}}
	\caption{The false alarm probability and overestimating probability versus the nominal false alarm probability, where all $K$ targets have identical integrated SNRs ${\rm SNR}$ and their frequencies satisfy the minimum frequency separation $\Delta\omega_{\rm min}=2.5\Delta_{\rm DFT}$.}
	\label{MUL_P_FA_OE_20_dB_figure}
\end{figure}
\subsection{Dynamic Range}
We consider two signals, and $N=512$. The amplitude of the first signal is stronger than that of the second signal. Define ${\rm SNR}_i$ as the integrated SNR of the $i$th sinusoid, and the DR as ${\rm DR}\triangleq {\rm SNR}_1-{\rm SNR}_2$. We evaluate the detection probability of the weakest signal versus ${\rm SNR}_1$ and ${\rm DR}$ for a given false alarm rate. The detection probability of target 1 (the weaker target) under different bit-depths is shown in Fig. \ref{DR_fig}. It can be seen that the instantaneous DR is about $10$ dB under $1$ bit quantization\footnote{In fact, the performance of FFT-based approaches on two simultaneous signals is investigated in \cite{DWR2016}. It is shown that when two signals are of the same amplitude, the receiver does not report them all the time. The receiver reports both signals only about 24\% of the time. About 76\% of the time, the receiver only reports one signal. Besides, the instantaneous DR of the monobit receiver is about $5$ dB, and the receiver measures the weak signal whose amplitude is $5$ dB weaker than that of the strong signal in 33/1000 trials. Here we show that our nonlinear approach performs significantly better than that of the FFT-based linear approach.}. For $2$ bits and $3$ bits quantization, the DRs are about $22$ dB and $30$ dB. This demonstrates that the proposed GNOMP enlarges the DR compared to the FFT-based linear approach.
\begin{figure*}[htbp]
	\centering
	\subfigure[B=1]{
		\label{DR_1bit_fig}
		\includegraphics[width=5cm]{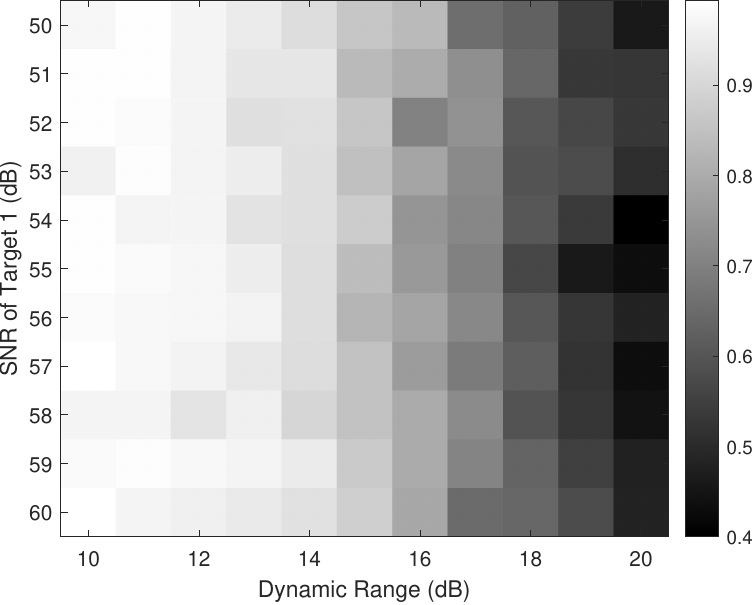}}\subfigure[B=2]{
		\label{DR_2bit_fig}
		\includegraphics[width=5cm]{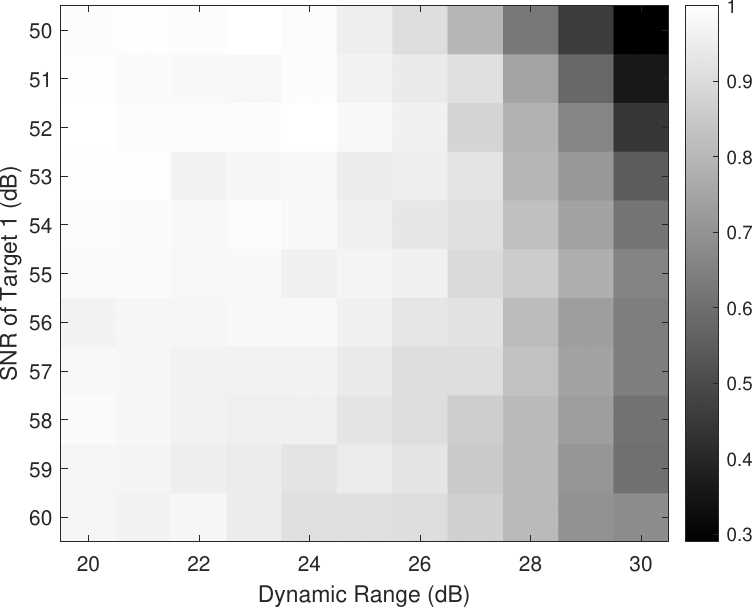}}\subfigure[B=3]{
		\label{DR_3bit_fig}
		\includegraphics[width=5cm]{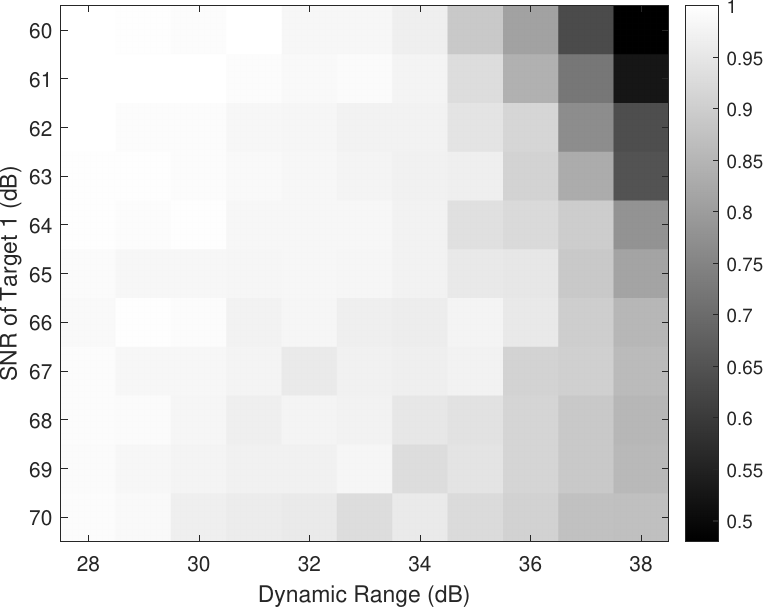}}
	\caption{Instantaneous DR under different bit-depth computed over $300$ MC trials.}
	\label{DR_fig}
\end{figure*}

\subsection{The MSE Performance of the Frequency}
In Fig. \ref{MUL_MSE_fig}, we plot the ``measured'' MSE and CRB versus the integrated SNR under different bit-depths. Each point is generated by 300 runs of the GNOMP algorithm for estimating frequencies in a mixture of $K=8$ sinusoids, with identical integrated SNRs ranging from $14$ dB to $40$ dB, and $N=512$. It can be seen that as SNR increases from $14$ dB to about $34$ dB, the GNOMP asymptotically approaches the CRB under 1-3 bit quantization. As SNR increases further, the GNOMP deviates away from the CRB, except in the no quantization situation.
\begin{figure}[htbp]
	\centering
	\includegraphics[width=4in]{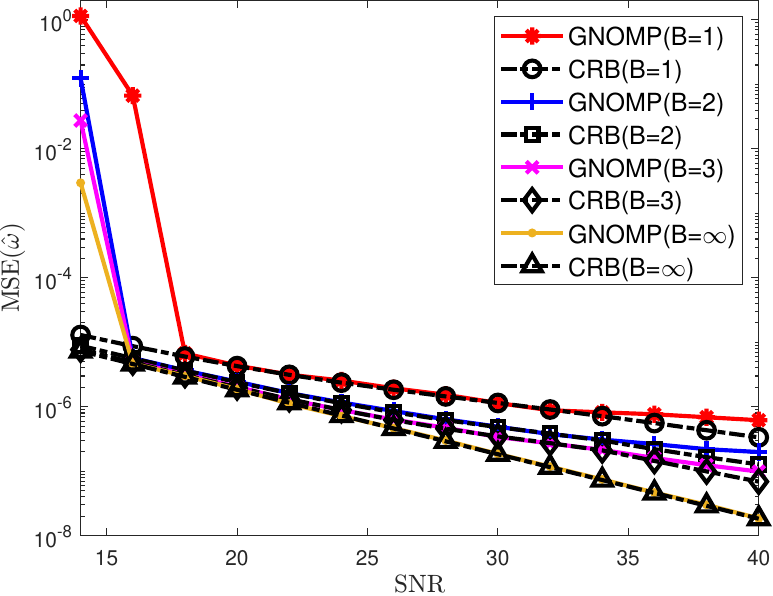}
	\caption{The measured MSE and CRB of frequencies versus the integrated SNRs under different bit-depth.}
	\label{MUL_MSE_fig}
\end{figure}

\subsection{The Detection Probability versus the SNR}\label{DetvsSNR}
We generate $8$ sinusoids, where $7$ of them have an identical integrated ${\rm SNR}= 30$ dB and $N=512$. The integrated SNR of the remaining sinusoid increases from $10$ dB to $26$ dB. The measured false alarm probability, the measured overestimating probability, and the detection probability of the remaining target are shown in Fig. \ref{MUL_Sin_P_FA_OE_D_figure}, where each point is generated by 1000 MC trials. It can be observed that the measured false alarm rate is close to the nominal false alarm rate, and the overestimating probability tends to be lower than the nominal false alarm rate. Moreover, the measured detection probability is close to the oracle detection probability, which is computed through (\ref{unknownfreq}) with $\boldsymbol\zeta=\sum_{i=1}^7x_i{\mathbf a}(\omega_i)$ excluding the weakest sinusoid, demonstrating the excellent performance of GNOMP.
\begin{figure*}[htbp]
	\centering
	\subfigure[False Alarm Rate]{
		\label{MUL_Sin_P_FA_figure}
		\includegraphics[width=5.2cm]{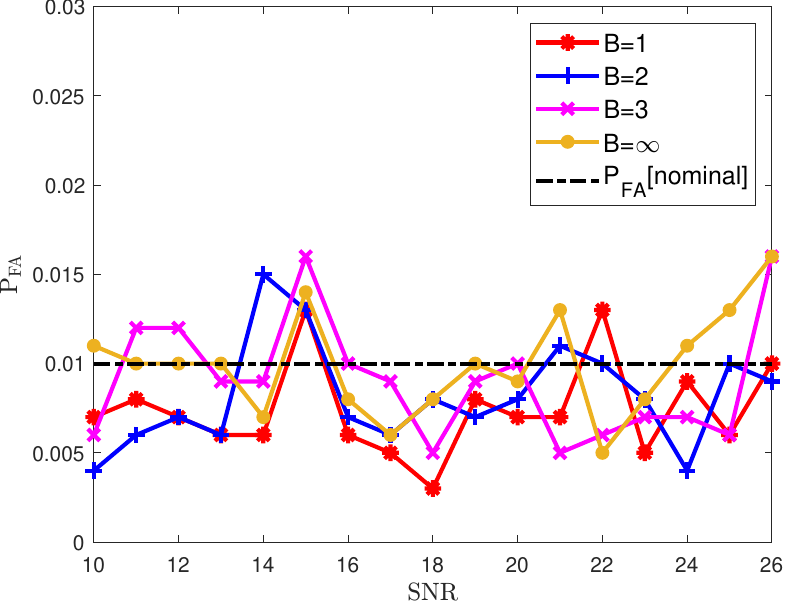}}
	\subfigure[Overestimating Probability]{
		\label{MUL_Sin_P_OE_figure}
		\includegraphics[width=5.2cm]{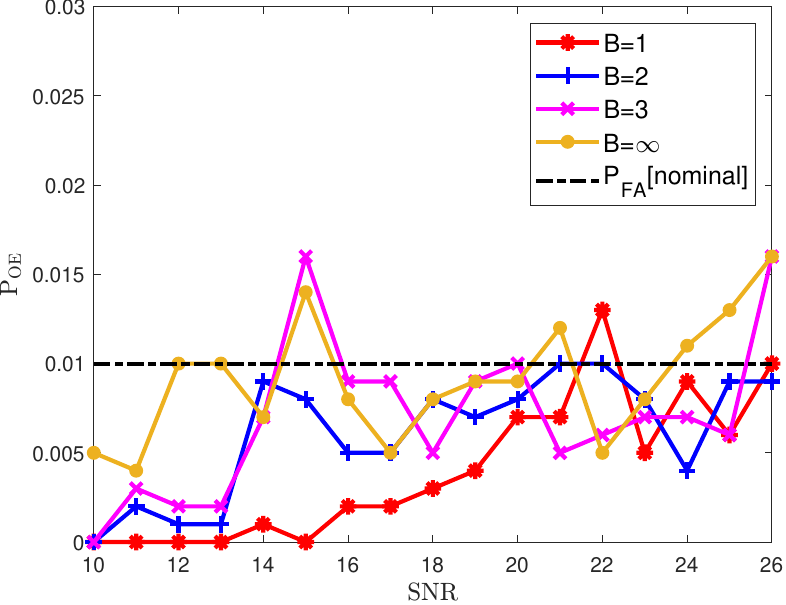}}
	\subfigure[Detection Probability]{
		\label{MUL_Sin_P_D_figure}
		\includegraphics[width=5.2cm]{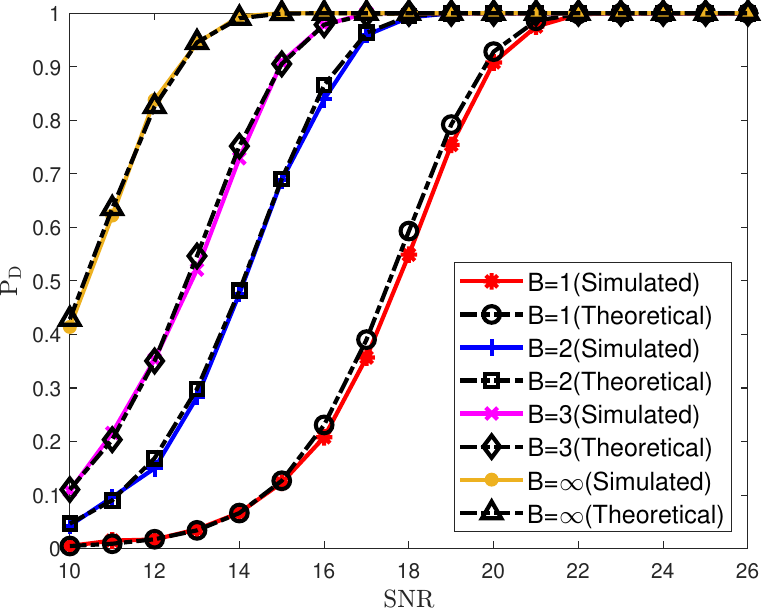}}
	\caption{The false alarm rate, the overestimating probability, the measured and theoretical detection probability of the remaining target versus the integrated SNR.}
	\label{MUL_Sin_P_FA_OE_D_figure}
\end{figure*}
\subsection{The Detection Probability versus the SNR from Signed Measurements with A Time-Varying Threshold}
The false alarm probability, overestimating probability, and detection probability versus SNR are investigated using signed measurements with a time-varying threshold. The thresholds used to obtain the signed measurements are randomly generated from a discrete set of 8 values uniformly distributed over $[-1,1]$. The number of sinusoids is $8$, with $7$ of them having the identical integrated ${\rm SNR}= 25$ dB. The integrated SNR of the remaining sinusoid increases from $14$ dB to $24$ dB. The noise variance is $\sigma^2=1$, and the results are averaged over $500$ MC trials. The GNOMP is compared with the MVALSE-EP and 1bMMRELAX, where the GNOMP and the MVALSE-EP are modified to address the LSE\&D in this setting. It is worth noting that both MVALSE-EP and 1bMMRELAX estimate the noise variance, and the proposed GNOMP is more flexible as the false alarm rate can be specified. We set $R_c=3$ for GNOMP and $R_c=4$ for GNOMP ($\sigma^2$ unknown). The results are shown in Fig. \ref{PFAOEDvsSNR}. Fig. \ref{PFAvsSNR} and Fig. \ref{POEvsSNR} show that the false alarm probabilities and overestimating probabilities of 1bMMRELAX, GNOMP (${\rm P}_{\rm FA}=0.01$), and GNOMP (${\rm P}_{\rm FA}=0.01$, $\sigma^2$ unknown) are close to $0.01$, except for GNOMP (${\rm P}_{\rm FA}=0.01$, $\sigma^2$ unknown) at ${\rm SNR}=25$ dB. The false alarm probability and overestimating of GNOMP (${\rm P}_{\rm FA}=0.05$) and GNOMP (${\rm P}_{\rm FA}=0.05$, $\sigma^2$ unknown) are close to the nominal value $0.05$ except GNOMP (${\rm P}_{\rm FA}=0.05$, $\sigma^2$ unknown) at ${\rm SNR}=25$ dB. The reason may be that one of the conditions under which the asymptotic expressions of the Rao detector hold is that the signal is weak. As for the detection probability, GNOMP (${\rm P}_{\rm FA}=0.05$) and GNOMP (${\rm P}_{\rm FA}=0.05$, $\sigma^2$ unknown) are highest and are close to the corresponding theoretical detection probability, followed by GNOMP (${\rm P}_{\rm FA}=0.01$), GNOMP (${\rm P}_{\rm FA}=0.01$, $\sigma^2$ unknown), 1bMMRELAX, MVALSE-EP. For ${\rm P}_{\rm D}=0.5$ and ${\rm P}_{\rm FA}=0.01$, the SNRs required by GNOMP (${\rm P}_{\rm FA}=0.01$), GNOMP (${\rm P}_{\rm FA}=0.01$, $\sigma^2$ unknown), and 1bMMRELAX are about $15.7$ dB, $15.7$ dB, and $16.7$ dB, respectively, demonstrating that GNOMP (${\rm P}_{\rm FA}=0.01$, $\sigma^2$ unknown) has a $1$ dB gain compared to 1bMMRELAX. We also evaluate the running time averaged over $50$ Monte Carlo trials. The running time is obtained on a desktop computer with an Intel(R) Core(TM) i7-12700 4.90 GHz CPU and 32 GB of RAM, running the operating system Microsoft Windows 11. It can be seen that the running time (mean$\pm$3 standard deviation) of GNOMP (${\rm P}_{\rm FA}=0.01$), GNOMP (${\rm P}_{\rm FA}=0.05$), GNOMP (${\rm P}_{\rm FA}=0.01$, $\sigma^2$ unknown), GNOMP (${\rm P}_{\rm FA}=0.05$, $\sigma^2$ unknown), 1bMMRELAX, MVALSE-EP are $1.70\pm 0.25$ sec., $1.74\pm 0.21$ sec., $6.12\pm 1.56$ sec., $6.18\pm 1.38$ sec., $7.20\pm 0.84$ sec., $12.84\pm 0.15$ sec., respectively. This demonstrates the computational efficiency of the proposed GNOMP approach.
\begin{figure*}[htbp]
	\centering
	\subfigure[False Alarm Rate]{
		\label{PFAvsSNR}
		\includegraphics[width=5.2cm]{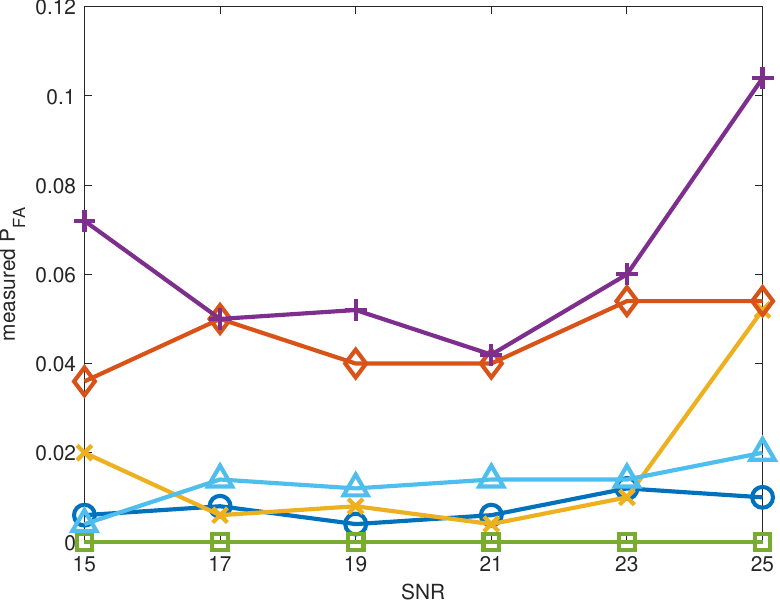}}
	\subfigure[Overestimating Probability]{
		\label{POEvsSNR}
		\includegraphics[width=5.2cm]{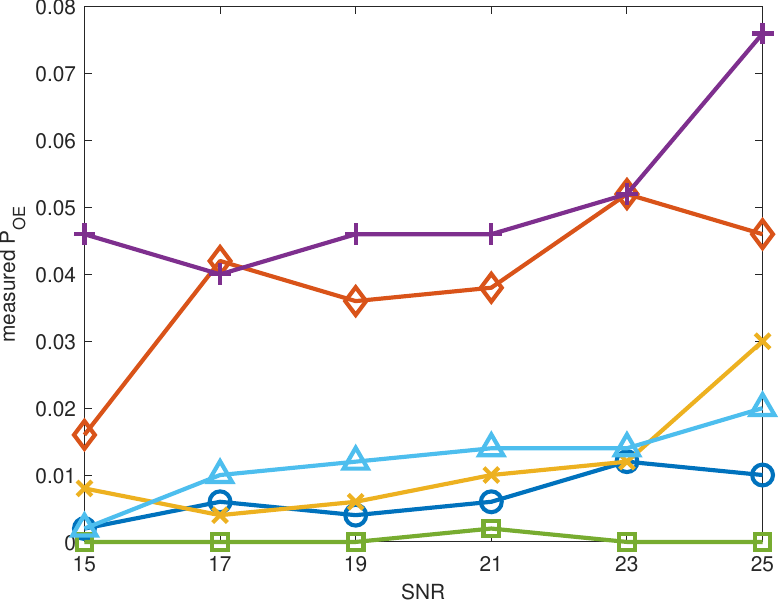}}
	\subfigure[Detection Probability]{
		\label{PDvsSNR}
		\includegraphics[width=5.2cm]{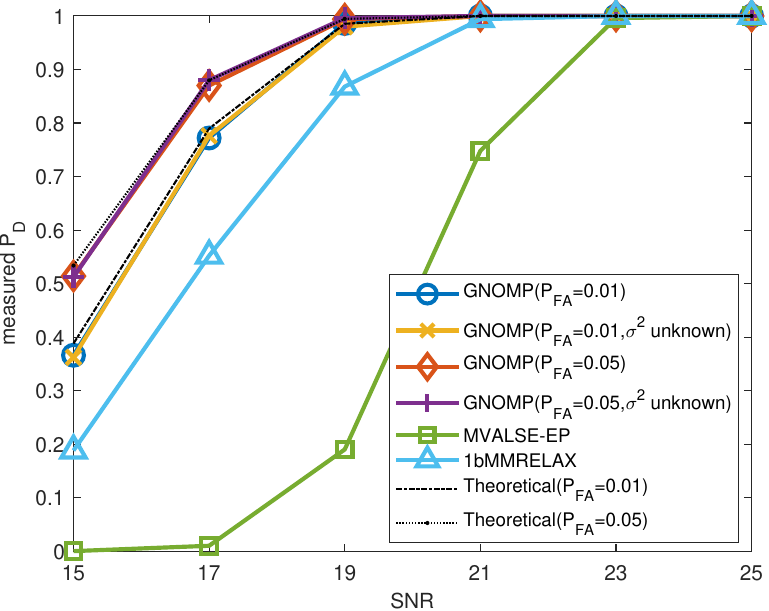}}
	\caption{The false alarm rate, the overestimating probability, the measured and theoretical detection probability of the remaining target versus the integrated SNR from signed measurements with a time-varying threshold.}
	\label{PFAOEDvsSNR}
\end{figure*}
\subsection{LSE\&D from quantized random linear measurements}
The problem can be formulated as ${\mathbf y}=\mathcal{Q}_{\rm C}\left({\boldsymbol \Phi}\left(\sum\limits_{k=1}^K{\mathbf{a}}({\omega }_k){{x}_k}\right)+{\boldsymbol\epsilon}\right)$. By defining $\breve{\mathbf{a}}({\omega }_k)\triangleq{\boldsymbol \Phi}\mathbf{a}({\omega }_k)$, the original model can be equivalently formulated as
\begin{align}\label{y_a_tilde}
   {\mathbf y} =\mathcal{Q}_{\rm C}\left(\sum\limits_{k=1}^K\breve{\mathbf{a}}({\omega }_k){{x}_k}+{\boldsymbol\epsilon}\right).
\end{align}
Consequently, the GNOMP can be applied directly by replacing ${\mathbf a}$ with $\breve{\mathbf{a}}$, and the Rao test $T_{{\rm R},B}(\boldsymbol{\varphi}_B,{\boldsymbol\zeta},\omega)$ in (\ref{generalRao}) can be calculated by replacing ${\mathbf{a}}$ with $\breve{\mathbf{a}}$.
In this way, one could solve and analyze this problem without much difficulty. We find that GNOMP and atomic norm soft-thresholding (AST) \cite{Fu} both perform well when $M>>N$. When $M$ is smaller, AST does not perform well, while GNOMP still performs well. Here we conduct two challenging numerical simulations to evaluate the performance of  GNOMP and AST under oversampling ($M=200>N=128$) and undersampling ($M=100<N=128$) setups. The parameters are set as follows: $N=128$, $K=2$, $\omega_1=1.5$, $\omega_2=3.2$, the time domain SNRs are $\rm{SNR}_1=\rm{SNR}_2=0$ dB, $\boldsymbol{\Phi}\in \mathbb{C}^{M\times N} $is a random matrix in which each element's real part or imaginary part follows standard Gaussian distribution.

For the oversampling scenario, AST fails to estimate the two targets as shown in Fig. \ref{AST_M_200}, while GNOMP perfectly estimates the two targets and terminates in the third iteration as shown in Fig.  \ref{GNOMP_M_200}. The estimation results are $\hat\omega_1=1.5008$ and $\hat\omega_2=3.2002$. For the undersampling scenario, the results of two algorithm are shown in Fig. \ref{AST_M_100} and Fig. \ref{GNOMP_M_100}, whose phenomena are similar. AST still fails to estimate the two targets. GNOMP perfectly estimates the two targets and terminated in the third iteration. The estimation results are $\hat\omega_1=1.4997$ and $\hat\omega_2=3.2015$. In summary, the GNOMP algorithm works well in both oversampling and undersampling scenarios.

\begin{figure}[htbp!]
	\centering
	\includegraphics[width=2.3in]{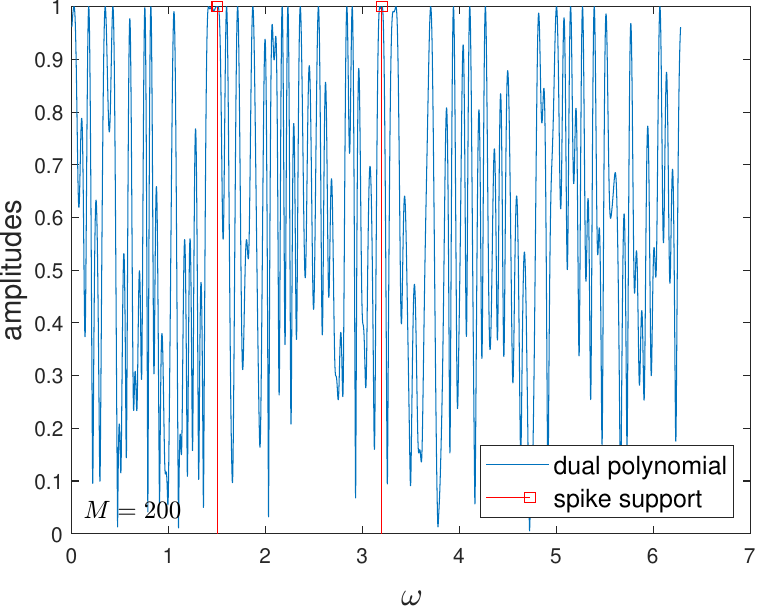}
	\caption{AST results ($M = 200$).}
 \label{AST_M_200}
\end{figure}
\begin{figure*}[htbp!]
	\centering
	\subfigure[The 1st iteration]{
		\label{Demo_1iter_M_200}
		\includegraphics[width=5cm]{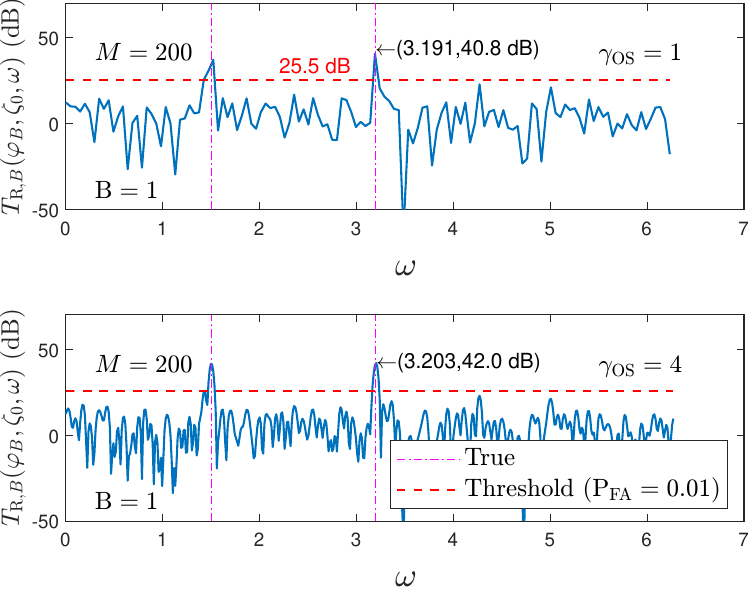}}\subfigure[The 2nd iteration]{
		\label{Demo_2iter_M_200}
		\includegraphics[width=5cm]{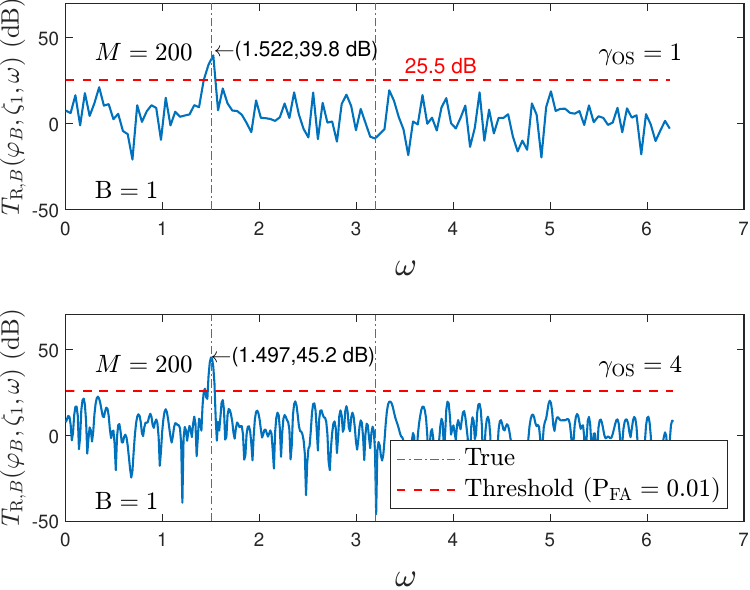}}
		\subfigure[The 3rd iteration]{
		\label{Demo_3iter_M_200}
		\includegraphics[width=5cm]{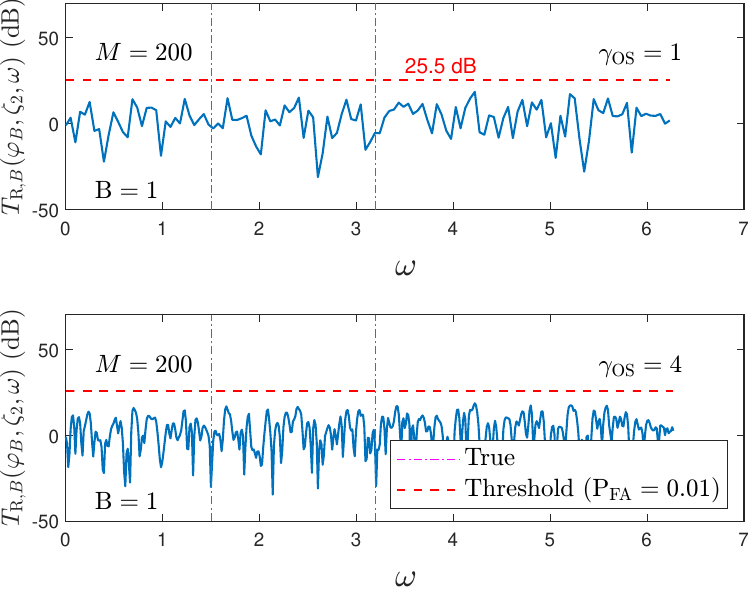}}
	\caption{GNOMP iterations ($M=200$).}
	\label{GNOMP_M_200}
\end{figure*}

\begin{figure}[htbp!]
	\centering
	\includegraphics[width=2.3in]{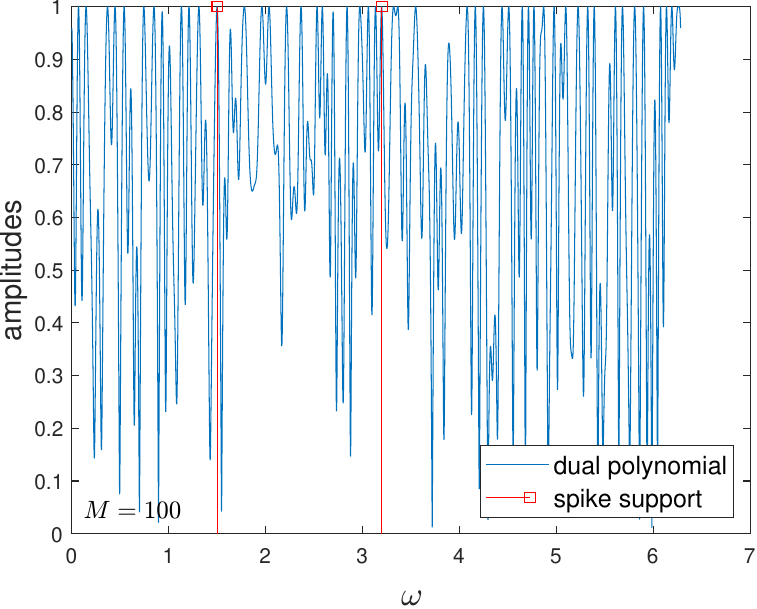}
	\caption{AST results ($M = 100$)}
 \label{AST_M_100}
\end{figure}
\begin{figure*}[htbp!]
	\centering
	\subfigure[The 1st iteration]{
		\label{Demo_1iter_M_100}
		\includegraphics[width=5cm]{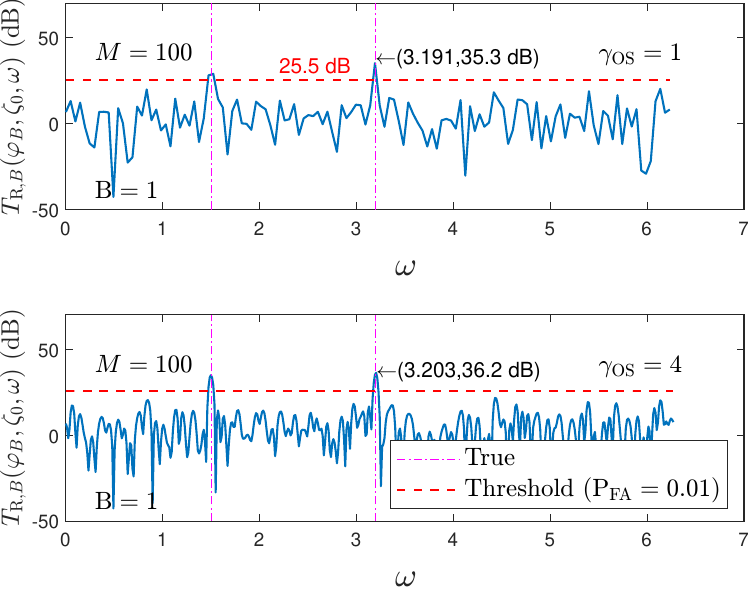}}\subfigure[The 2nd iteration]{
		\label{Demo_2iter_M_100}
		\includegraphics[width=5cm]{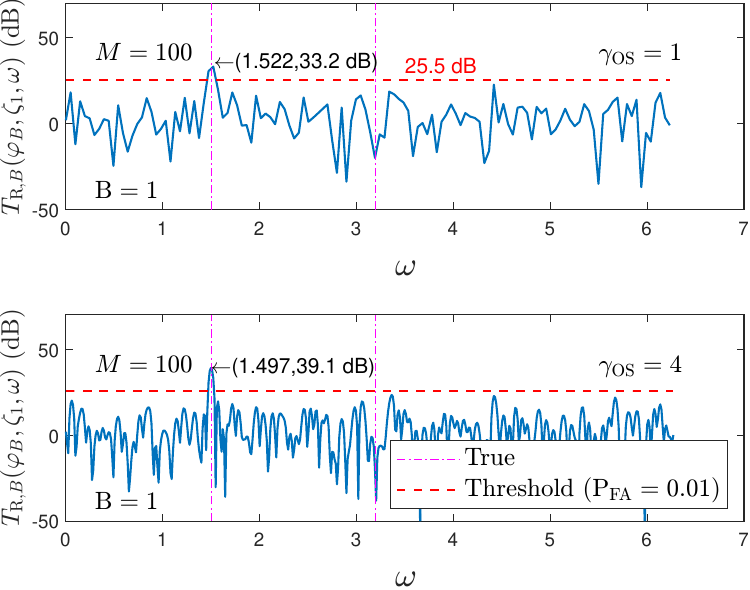}}
	\subfigure[The 3rd iteration]{
		\label{Demo_3iter_M_100}
		\includegraphics[width=5cm]{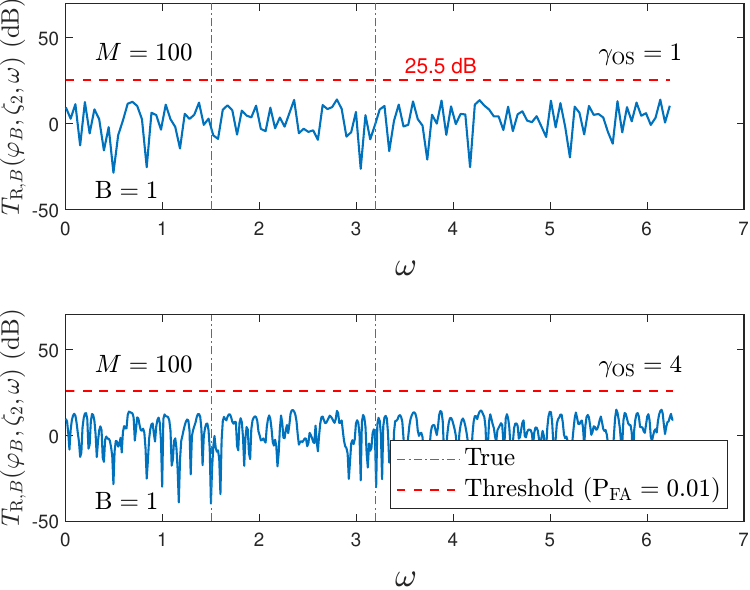}}
	\caption{GNOMP iterations ($M=100$).}
	\label{GNOMP_M_100}
\end{figure*}

\section{Real Experiment}\label{RE}
In this section, we use an AWR1642 radar to demonstrate the performance of proposed GNOMP.
The AWR1642 radar consists of two transmitters and four receivers, and it is an FMCW MIMO radar. The parameters and waveform details we used in the experiments are presented in Table \ref{tab:Radar_parameters}.
\newcommand{\tabincell}[2]{\begin{tabular}{@{}#1@{}}#2\end{tabular}}
\begin{table}[!t]
  \centering
  \scriptsize
  \caption{Parameters Setting of the Two Experiments}
  \label{tab:Radar_parameters}
  \begin{tabular}{ll}
    \\[-2mm]
    \hline
    \hline\\[-2mm]
    { \small Parameters}&\qquad {\small Value}\\
    \hline
    \vspace{1mm}\\[-3mm]
    Number of Receivers $L$      &   4 \\
    Carrier frequency $f_c$   &   77 GHz\\
    \vspace{1mm}
    Frequency modulation slope $\mu$          &  $29.982$ MHz/s\\
    \vspace{1mm}
    sweep time $T_p$         &  $60 {\rm \mu s}$\\
    \vspace{1mm}
    Pulse repeat interval $T_{\text{r}}$         &  $160 {\rm \mu s}$\\
    \vspace{1mm}
    Bandwidth $B$          &  1798.92 MHz\\
    \vspace{1mm}
    Sampling frequency $f_s = 1 / T_s$         &  10 MHz\\
    \vspace{1mm}
    Number of pulses in one CPI $M$          & 128\\
    \vspace{1mm}
    Number of fast time samples $N$          &  256 \\
    \hline
    \hline
  \end{tabular}
\end{table}

The real data acquired in \cite{NOMPCFAR} is utilized to investigate the estimation and detection performance of GNOMP. For GNOMP, the FFT is used to estimate the noise variance, which is then input to the GNOMP algorithm. Additionally, GNOMP is also implemented without the knowledge of the noise variance to perform target estimation and detection with bit-depth greater than $1$. The experimental setup involves setting $\sigma^2=250$, $\tilde{\rm P}_{\rm{FA},B}=0.01$, $B=1,2,3,\infty$, and the maximum full-scale range of the quantizer is $\gamma=60$, with $N=256$.

\begin{figure}[ht!]
	\centering
	\subfigure[The setup of the first experiment]{
		\label{experiment_1}
		\includegraphics[height=1.7in, width=2.4in]{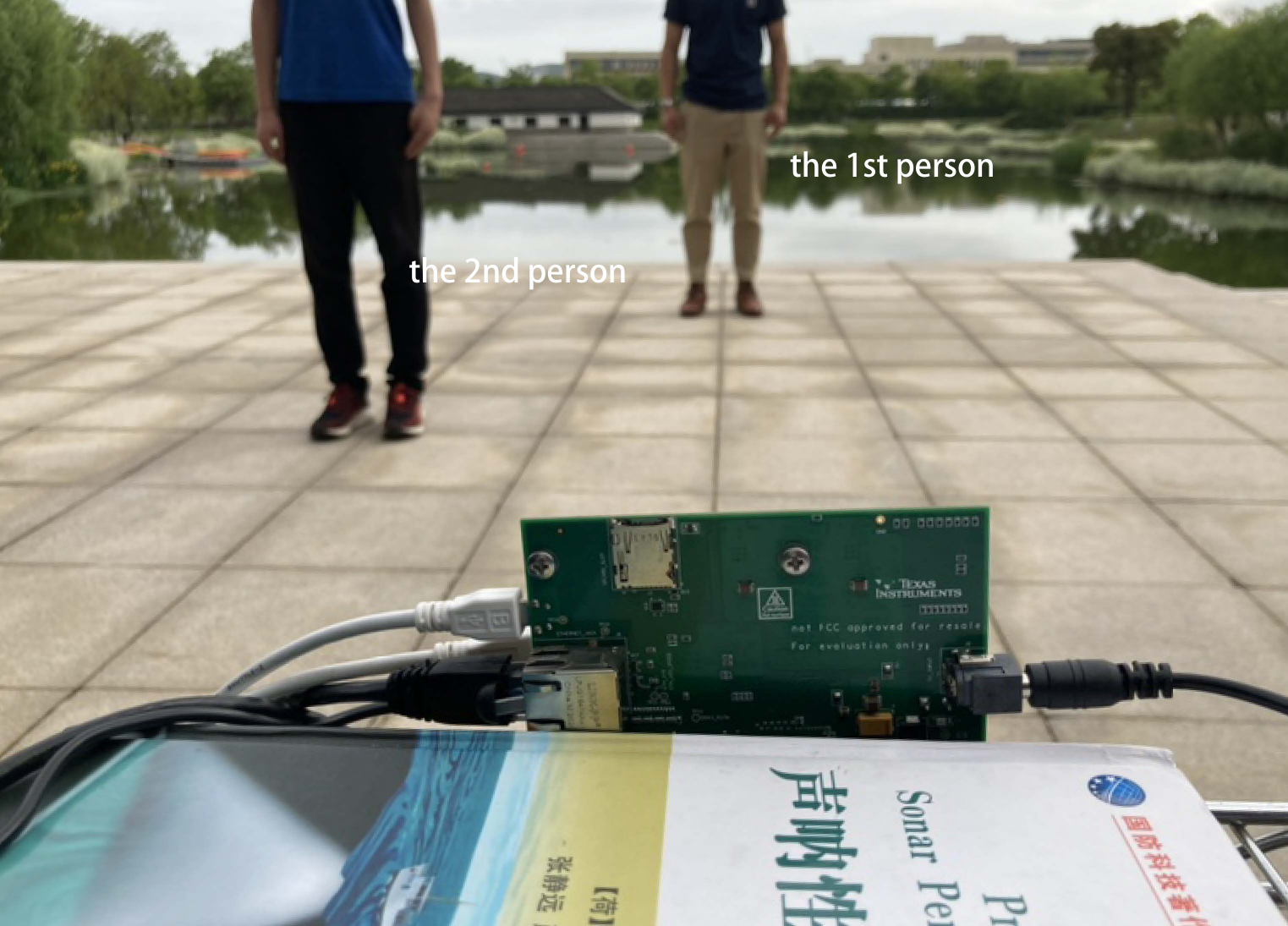}}
	\subfigure[The setup of the second experiment]{
		\label{experiment_2}
		\includegraphics[height=1.7in, width=2.4in]{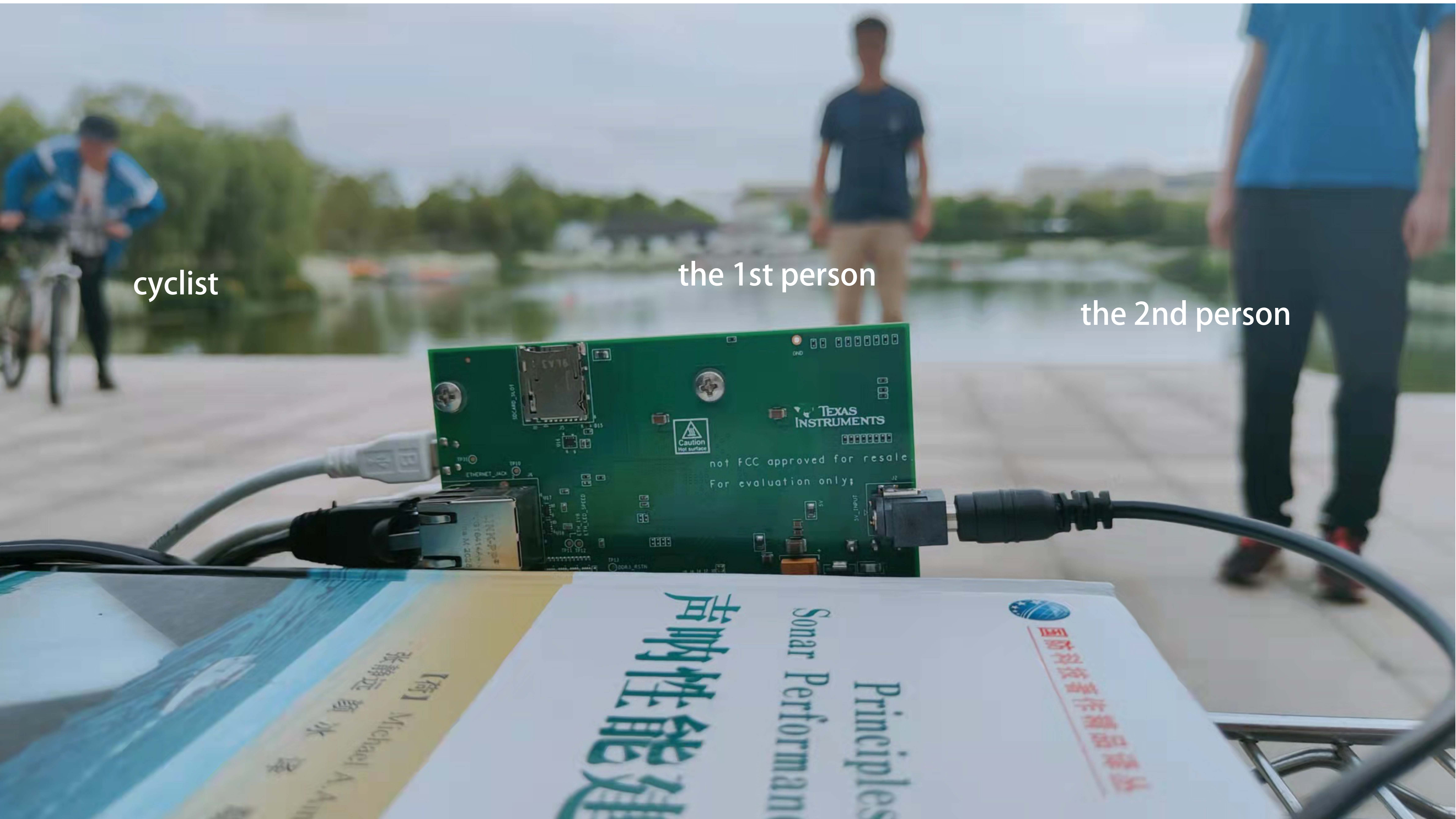}}
	\caption{The setup of the two experiments.}
	\label{experiment_12}
\end{figure}

\begin{figure}[ht!]
	\centering
	\subfigure[$B=1$]{
		\label{GNOMP_people2_1bit}
		\includegraphics[width=2.5in]{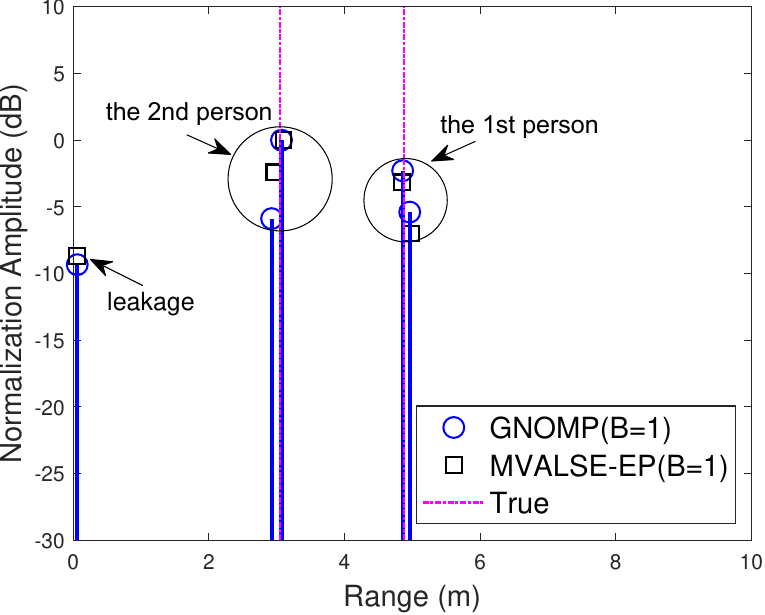}}\subfigure[$B=2$]{
		\label{GNOMP_people2_2bit}
		\includegraphics[width=2.5in]{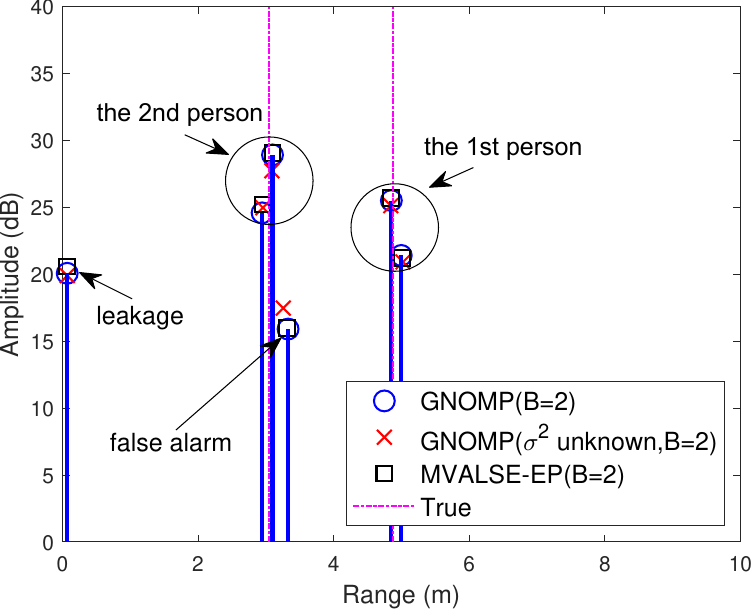}}
	\subfigure[$B=3$]{
		\label{GNOMP_people2_3bit}
		\includegraphics[width=2.5in]{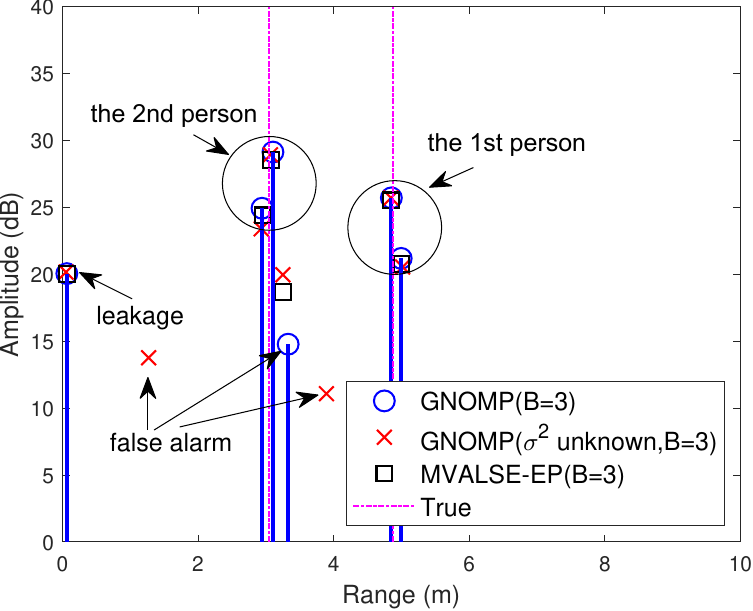}}\subfigure[$B=\infty$]{
		\label{GNOMP_people2_infty}
		\includegraphics[width=2.5in]{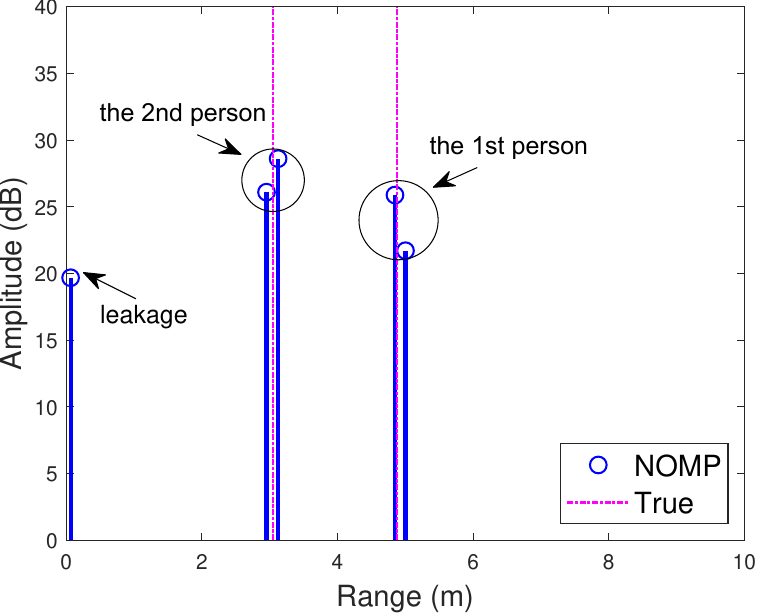}}
	\caption{The range estimation and detection results of the first experiment.}
	\label{GNOMP_real_people2}
\end{figure}
\subsection{The first Experiment}
In the first experiment, as depicted in Fig. \ref{experiment_1}, two individuals, the first person and the second person, are positioned in front of the radar at radial distances of approximately $4.88$ m and $3.05$ m, respectively. The results presented in Fig. \ref{GNOMP_real_people2} demonstrate the detection of the first person, the second person, and the leakage component. GNOMP and GNOMP ($\sigma^2$ unknown) generate false alarms under $2$ and $3$ bit quantization. Regarding running time, NOMP takes $0.004$ sec, GNOMP ($B=1$) takes $0.11$ sec, GNOMP ($B=2$) takes $0.15$ sec (and $0.85$ sec for unknown noise variance), and GNOMP ($B=3$) takes $0.14$ sec (and $1.05$ sec for unknown noise variance). On the other hand, MVALSE-EP ($B=1$) takes $1.18$ sec, MVALSE-EP ($B=2$) takes $1.5$ sec, and MVALSE-EP ($B=3$) takes $1.47$ sec. Therefore, GNOMP exhibits much faster processing times compared to MVALSE-EP.

\begin{figure}[htbp]
	\centering
	\subfigure[$B=1$]{
		\label{GNOMP_multiple_1bit}
		\includegraphics[width=2.5in]{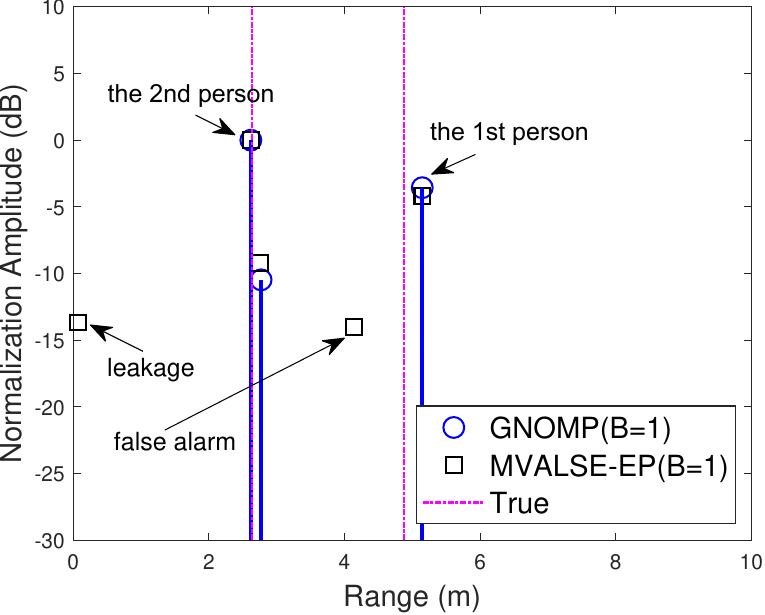}}\subfigure[$B=                                   2$]{
		\label{GNOMP_multiple_2bit}
		\includegraphics[width=2.5in]{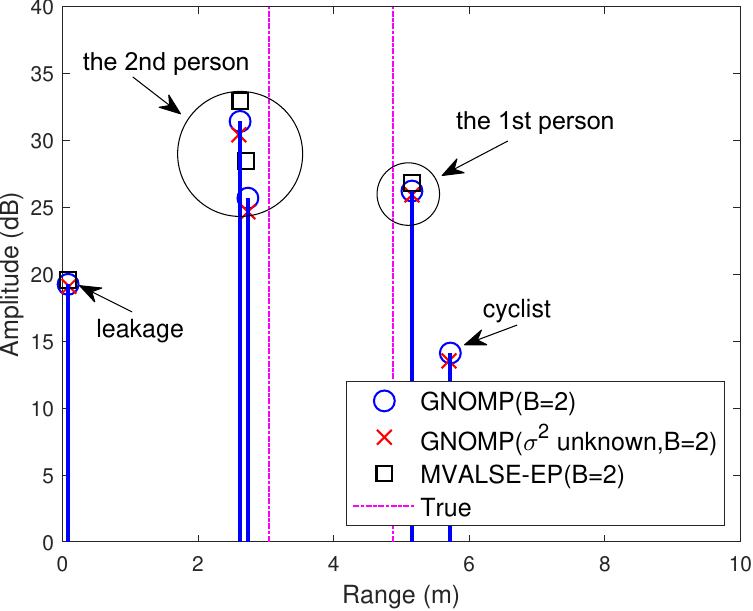}}
	\subfigure[$B=3$]{
		\label{GNOMP_multiple_3bit}
		\includegraphics[width=2.5in]{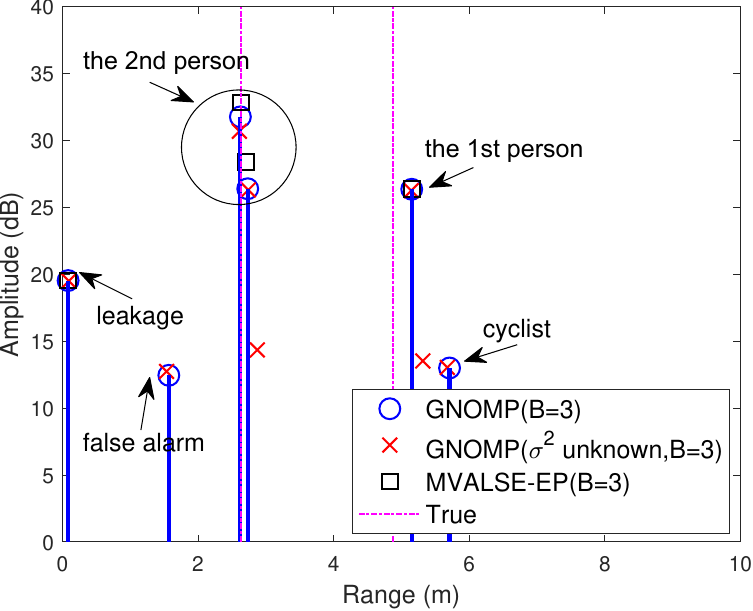}}\subfigure[$B=\infty$]{
		\label{GNOMP_multiple_infty}
		\includegraphics[width=2.5in]{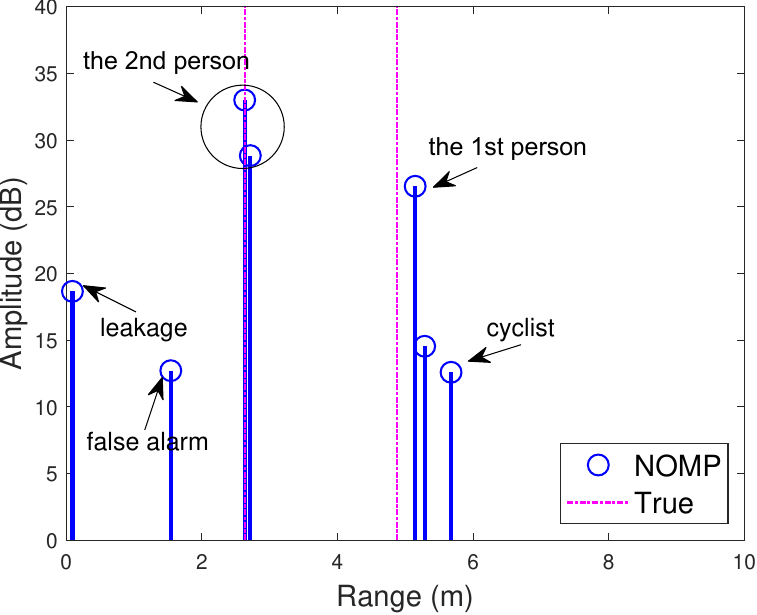}}
	\caption{The estimation and detection results of the second experiment.}
	\label{GNOMP_real_multiple}
\end{figure}
\subsection{The Second Experiment}
In the second field experiment, as illustrated in Fig. \ref{experiment_2}, two stationary individuals,  the first person and the second person, are positioned at radial distances of about $4.87$ m and $2.63$ m. Meanwhile, a cyclist is moving toward the radar with a radial distance ranging from $7$ m to $2$ m and a velocity of approximately $2$ m/s. The range estimation and detection results are depicted in Fig. \ref{GNOMP_real_multiple}. Observations indicate that under $1$ bit quantization, both GNOMP and MVALSE-EP fail to detect the cyclist due to the inherent low DR associated with $1$ bit quantization. MVALSE-EP also generates a false alarm. For $B\geq 2$, GNOMP successfully detects both persons and the cyclist, while MVALSE-EP continues to miss the cyclist. Additionally, under $3$ bit quantization, both GNOMP and GNOMP ($\sigma^2$ unknown), along with NOMP, produce a false alarm.
\section{Conclusion}\label{Cons}
This paper presents a theoretical analysis of false alarm probability and detection probability under low-resolution quantization. The SNR loss revealing the impact of both low-resolution quantization and intersinusoidal interference on the detection of weak signals is introduced. Furthermore, a fast and super-resolution algorithm, GNOMP, which utilizes FFT and IFFT for implementing the Rao detector, achieving LSE\&D while maintaining CFAR behavior, is proposed. The theoretical findings are verified and the efficiency and excellent performance of GNOMP are demonstrated through extensive numerical simulations and real experiments, comparing with state-of-the-art algorithms, the CRB, and the detection probability bound. Considering the low computation complexity and excellent estimation and detection performance of GNOMP, evaluating its computation resource and efficiency on the hardware will be done in future work.
\section{Acknowledgement}
The authors acknowledge Dr. Jiaying Ren and Dr. Haoyu Fu for sharing their codes to make performance comparisons.
\section{Appendix}
\subsection{The FIM for the Single Sinusoid Model (\ref{BHT}) with Nonzero Thresholds}\label{FIM0single}
We now evaluate the FIM for the model (\ref{BHT}) under the hypothesis ${\mathcal {H}}_1$  using the following lemma.
\begin{lemma}\label{lemmaFIM}
\cite{Fu, VALSEEP, NingTAES22} Let ${\boldsymbol \kappa}\in {\mathbb R}^P$ denote the set of unknown deterministic parameters. Note that in the case of quantized observations ${\mathbf y}={\mathcal Q}({\mathbf r})\in{\mathbb R}^N$ where ${\mathbf r}\sim {\mathcal {N}}({\boldsymbol \mu}({\boldsymbol \kappa}),{\sigma}^2{\mathbf I}_N/2)$, the FIM is given by
\begin{align}\label{FIMcal}
{\mathbf I}({\boldsymbol \kappa})=\frac{2}{\sigma^2}\left[\frac{\partial {\boldsymbol \mu}({\boldsymbol \kappa})}{\partial {\boldsymbol \kappa}^{\rm T}}\right]^{\rm T}{\boldsymbol \Lambda}\left[\frac{\partial {\boldsymbol \mu}({\boldsymbol \kappa})}{\partial {\boldsymbol \kappa}^{\rm T}}\right],
\end{align}
where
\begin{align}
\frac{\partial {\boldsymbol \mu}({\boldsymbol \kappa})}{\partial {\boldsymbol \kappa}^{\rm T}}=\left[
                                                                                         \begin{array}{cccc}
                                                                                           \frac{\partial [{\boldsymbol \mu}({\boldsymbol \kappa})]}{\partial \kappa_1} & \frac{\partial [{\boldsymbol \mu}({\boldsymbol \kappa})]}{\partial \kappa_2} & \cdots & \frac{\partial [{\boldsymbol \mu}({\boldsymbol \kappa})]}{\partial \kappa_P}
                                                                                         \end{array}
                                                                                       \right]\in {\mathbb R}^{N\times P},
\end{align}
and ${\boldsymbol \Lambda}$ is a diagonal matrix with the $(i,i)$th element
\begin{align}
\Lambda_{i,i}=h_B(\mu_i(\boldsymbol \kappa),\sigma^2),
\end{align}
$h_B(x,\sigma^2)$ is given by (\ref{hxsigma}) and $B$ is the bit-depth of the quantizer. For unquantized system, the FIM (\ref{FIMcal}) is obtained with ${\boldsymbol \Lambda}={\mathbf I}_N$.
\end{lemma}

In our setting, the observations are \small{$\left[\Re\{{\mathbf y}\};\Im\{{\mathbf y}\} \right]$}. Note that ${\boldsymbol \kappa}\in{\mathbb R}^{2}$ and ${\boldsymbol \mu}({\boldsymbol \kappa})\in{\mathbb R}^{2N}$ are
\begin{align}\label{partres}
{\boldsymbol \kappa}=\left[\begin{array}{c}
  \Re\{{x}\} \\
  \Im\{{x}\}
\end{array}
 \right],
 {\boldsymbol \mu}({\boldsymbol \kappa})=\left[\begin{array}{c}
  \Re\{{\boldsymbol\zeta}+{\mathbf a}{x}\} \\
  \Im\{{\boldsymbol\zeta}+{\mathbf a}{x}\}
\end{array}
 \right],
\end{align}
respectively. Thus
\begin{align}\label{parmukappasingle}
\frac{\partial {\boldsymbol \mu}({\boldsymbol \kappa})}{\partial {\boldsymbol \kappa}^{\rm T}}=\left[
                  \begin{array}{cc}
                    \Re\{{\mathbf a}\} & -\Im\{{\mathbf a}\} \\
                    \Im\{{\mathbf a}\} & \Re\{{\mathbf a}\} \\
                  \end{array}
                \right]\in{\mathbb R}^{2N\times 2}.
\end{align}
Substituting (\ref{parmukappasingle}) in (\ref{FIMcal}), the FIM ${\mathbf I}_B({\boldsymbol \theta})$ is
\begin{align}\label{FIM0}
	{\mathbf I}_B({\boldsymbol \theta})=&\frac{2}{\sigma^2}\left[
	\begin{array}{cc}
		\Re\{{\mathbf a}\} & -\Im\{{\mathbf a}\} \\
		\Im\{{\mathbf a}\} & \Re\{{\mathbf a}\} \\
	\end{array}
	\right]^{\rm T}\rm {diag}
	\begin{pmatrix}
		h_B(\Re\{{\mathbf a}x+\boldsymbol{\zeta}\},\sigma^2)\\
		h_B(\Im\{{\mathbf a}x+\boldsymbol{\zeta}\},\sigma^2)\\
	\end{pmatrix}\left[
	\begin{array}{cc}
		\Re\{{\mathbf a}\} & -\Im\{{\mathbf a}\} \\
		\Im\{{\mathbf a}\} & \Re\{{\mathbf a}\} \\
	\end{array}
	\right]
\end{align}
Simplifying (\ref{FIM0}) yields (\ref{FIMsingle0}). The inverse of ${\mathbf I}_B({\boldsymbol \theta})$ is shown to be (\ref{CRBsingle0}).

\begin{thebibliography}{99}
\bibitem{Mishra19SPM}
K. V. Mishra, M. R. Bhavani Shankar, V. Koivunen, B. Ottersten, and S. A. Vorobyov, ``Toward millimeterwave joint radar communications: A signal processing perspective,'' \emph{IEEE Signal Process. Mag.}, vol. 36, no. 5,
pp. 100-114, Sep. 2019.
\bibitem{KumariICASSP2020}
P. Kumari, A. Mezghani, R. W. Heath, ``A low-resolution ADC proof-of-concept development for a fully-digital millimeter-wave joint communication-Radar,'' \emph{ICASSP}, 2020.
\bibitem{Ciuonzo2013SPL}
D. Ciuonzo,  G. Papa, G. Romano, P. Salvo Rossi, and P. Willett, ``One-bit decentralized detection with a Rao test for multisensor fusion,'' \emph{IEEE Signal Process. Lett.}, vol. 20, no. 9, pp. 861-864, 2013.
\bibitem{NiTSP2023}
L. Ni, D. Zhang, Y. Sun, N. Liu, J. Liang and Q. Wan, ``Detection and localization of one-bit signal in multiple distributed subarray systems,'' \emph{IEEE Trans. Signal Process.}, vol. 71, pp. 2776-2791, 2023.\\
\bibitem{Xu2021TSIPN}
X. Cheng, D. Ciuonzo, P. Salvo Rossi, X. Wang and W. Wang, ``Multi-bit \& sequential decentralized detection of a noncooperative moving target through a generalized Rao test,'' \emph{IEEE Transactions on Signal and Information Processing over Networks}, vol. 7, pp. 740-753, 2021.
\bibitem{DWR2016p25}
J. Tsui and C. H. Cheng, ``Digital Techniques for Wideband Receivers,'' 3rd Edition, Scitech Publishing, p. 25, 2015.
\bibitem{LFMCWTAES20}
B. Jin, J. Zhu, Q. Wu, Y. Zhang and Z. Xu, ``One-bit LFMCW radar: spectrum analysis and target detection,'' \emph{IEEE Trans. Aerospace and Electronic Syst.},  vol. 56. no. 4, pp. 2732-2750, 2020.
\bibitem{SAR1991}
G. Franceschetti, V. Pascazio and G. Schirinzi, ``Processing of signum coded SAR signal: theory and experiments,'' \emph{IEE Proceedings F - Radar and Signal Processing }, vol. 138, no. 3, pp. 192-198, 1991.
\bibitem{onebitDBF}
X. Chen, L. Huang, H. Zhou, Q. Li, K. Yu and W. Yu, ``One-bit digital beamforming,'' \emph{IEEE Trans. Aerospace and Electronic Syst.}, vol. 59. no. 1, pp. 555-567, 2020.
\bibitem{Papatit}
H.~C.~Papadopoulos, G.~W.~Wornell and A.~V.~Oppenheim, ``Sequential signal encoding from noisy measurements using quantizers with dynamic bias control,'' \emph{IEEE Trans. Inf. Theory}, vol. 47, no. 3, pp. 978-1002, Mar. 2001.
\bibitem{SingletoneTSP}
A. H. Madsen and P. H\"{a}ndel, ``Effects of sampling and quantization on single-tone frequency estimation,'' \emph{IEEE Trans. Signal Process.}, vol. 48, no. 3, pp. 650-662, Mar. 2000.
\bibitem{DOA1bit02}
O. B. Shalom and A. J. Weiss, ``DOA estimation using one-bit quantized measurements,'' \emph{IEEE Trans. Aerosp. Electron. Syst.}, vol. 38, no. 3, pp. 868-884, Jul. 2002.
\bibitem{NingTAES22}
N. Zhang, J. Zhu and Z. Xu, ``Gridless multisnapshot variational line spectral estimation from coarsely quantized samples,''  \emph{IEEE Trans. Aerosp. Electron. Syst.},  vol. 59, no. 3, pp. 2979-2993, 2022.
\bibitem{Yu}
K. Yu, Y. Zhang, M. Bao, Y. Hu and Z. Wang, ``DOA estimation from one-bit compressed array data via joint sparse representation,'' \emph{IEEE Signal Process. Lett.}, vol. 23, no. 9, pp. 1279-1283, 2016.
\bibitem{MengZhu}
X. Meng, J. Zhu, ``A generalized sparse Bayesian learning algorithm for one-bit DOA estimation,'' \emph{IEEE Commun. Lett.}, vol. 22, no. 7, pp. 1414-1417, 2018.
\bibitem{mismatch}
Y. Chi, L. L. Scharf, A. Pezeshki and A. R. Calderbank, ``Sensitivity to basis mismatch in compressed sensing,'' \emph{IEEE Trans. Signal Process.}, vol. 59, no. 5, May 2011.
\bibitem{Yangzaireview}
Z. Yang, J. Li, P. Stoica and L. Xie, ``Sparse methods for direction-of-arrival estimation,'' \emph{Academic Press Library in Signal Processing}, vol. 7, pp. 509-581, 2018.
%
\bibitem{Fu}
H. Fu and Y. Chi, ``Quantized spectral compressed sensing: Cram\'{e}r-Rao bounds and recovery algorithms,'' \emph{IEEE Trans. Signal Process.}, vol. 66, no. 12, pp. 3268-3279, 2018.
\bibitem{Wen}
C. J. Wang, C. K. Wen, S. Jin, S. H. Tsai, ``Gridless channel estimation for mixed one-bit antenna array systems,'' \emph{IEEE Trans. Wireless Commun.}, vol. 17, no. 12, pp. 8485-8501, 2018.
\bibitem{Gaoyu}
Y. Gao, D. Hu, Y. Chen, Y. Ma, ``Gridless 1-b DOA estimation exploiting SVM approach,'' \emph{IEEE Commun. Lett.}, vol. 21, no. 10, pp. 2210-2213, 2017.
\bibitem{Gianelli2}
C. Gianelli, L. Xu, J. Li, P. Stoica, ``One-bit compressive sampling with time-varying thresholds for multiple sinusoids,'' \emph{CAMSAP}, 10-13 Dec. 2017, Curacao, Netherlands Antilles.
\bibitem{LiJian18SPL}
C. Li, R. Zhang, J. Li, and P. Stoica, ``Bayesian information criterion for
signed measurements with application to sinusoidal signals,'' \emph{IEEE Signal
Process. Lett.}, vol. 25, no. 8, pp. 1251-1255, Aug. 2018.
\bibitem{LiJian19TSP}
J. Ren, T. Zhang, J. Li and P. Stoica, ``Sinusoidal parameter estimation from signed measurements via Majorization-Minimization based RELAX,''
\emph{IEEE Trans. Signal Process.}, vol. 67, no. 8, pp. 2173-2186, 2019.
\bibitem{Zhang2019}
R. Zhang, C. Li, J. Li and G. Wang, ``Range estimation and range-Doppler imaging using signed measurements in LFMCW radar,'' \emph{IEEE Trans. Aerosp. Electron. Syst.}, vol. 55, no. 6, pp. 3531-3550, 2019.
\bibitem{VALSEEP}
J. Zhu, Q. Zhang and X. Meng, ``Gridless variational Bayesian inference of line spectral from quantized samples,'' \emph{China Commun.}, vol. 18, no. 10, pp. 77-95, 2021.
\bibitem{zhuTWC}
J. Zhu, C. Wen, J. Tong, C. Xu and S. Jin, ``Grid-less variational Bayesian channel estimation for antenna array systems with low resolution ADCs,'' \emph{IEEE Trans. Wireless Commun.}, vol. 19, no. 3, pp.  1549-1562, 2020.
\bibitem{Madhow16TSP}
B. Mamandipoor, D. Ramasamy and U. Madhow, ``Newtonized orthogonal matching pursuit: Frequency estimation over the continuum,'' \emph{IEEE Trans. Signal Process.}, vol. 64, no. 19, pp. 5066-5081, 2016.
\bibitem{Nir2019}
N. Shlezinger, Y. C. Eldar and M. R. D. Rodrigues, ``Hardware-limited task-based quantization,'' \emph{IEEE Trans. Signal Process.}, vol. 67, no. 20, pp. 5223-5238, 2019.


\bibitem{KayEst}
S. M. Kay, ``Fundamentals of Statistical Signal Processing: Estimation Theory,'' Prentice-Hall, Englewood Cliffs, N.J., 1993.
\bibitem{KayDet}
S. M. Kay, ``Fundamentals of Statistical Signal Processing: Detection Theory,'' Prentice-Hall, Englewood Cliffs, N.J., p. 205, 1993.
\bibitem{Proakis}
J. G. Proakis, \emph{Digital Communications}. New York, NY, USA: McGraw-Hill, 2001.
\bibitem{JianLiquant}
Y. Cheng, X. Shang, J. Li and P. Stoica, ``Interval design for signal parameter estimation from quantized data,'' \emph{IEEE Trans. Signal Process.}, vol. 70, pp. 6011-6020, 2022.
\bibitem{guanyudet}
G. Wang, J. Zhu and  Z. Xu, ``Asymptotically optimal one-bit quantizer design for weak-signal detection in generalized Gaussian noise and lossy binary communication channel,'' \emph{Signal Process.}, vol. 154, pp. 207-216, 2018.
\bibitem{JiangzhuTSP15}
J. Zhu, X. Lin, R. S. Blum and Y. Gu, ``Parameter estimation from quantized observations in multiplicative noise environments,'' \emph{IEEE Trans. Signal Process.}, vol. 63, no. 15, pp. 4037-4050, 2015.
\bibitem{NOMPCFAR}
M. Xu, J. Zhu, J. Fang, N. Zhang and Z. Xu, ``CFAR based NOMP for line spectral estimation and detection,'' \emph{IEEE Trans. Aerosp. Electron. Syst.}, vol. 59, no. 5, pp. 6971-6990, 2023.






\bibitem{onebitMIMOradaedet}
L. Huang, ``One-bit sampling based target detection in MIMO radar system,'' Seminal Report, 2021, available: https://www.eet-china.com/mp/a101872.html.
\bibitem{CSquant}
A. Zymnis, S. Boyd and E. Cand\`{e}s, ``Compressed sensing with quantized measurements,'' \emph{IEEE Signal Process. Lett.}, vol. 17, no. 2, pp. 149-152, 2010.
\bibitem{DWR2016}
J. Tsui and C. H. Cheng, ``Digital Techniques for Wideband Receivers,'' 3rd Edition, Scitech Publishing, p. 280, 2015.
\end{thebibliography}
\end{document}